\documentclass[aps,reprint,floatfix,superscriptaddress]{revtex4-2}
\usepackage{graphicx}
\usepackage{dcolumn}
\usepackage{bm}
\usepackage[colorlinks=true,urlcolor=blue,citecolor=blue,linkcolor=blue]{hyperref}
\usepackage{suffix}

\usepackage{xr}
\usepackage{xfrac}
\usepackage{multirow}
\usepackage{amsmath}
\usepackage{amssymb}
\usepackage{braket}
\usepackage{physics}
\usepackage{amsthm}
\usepackage{natbib}
\usepackage{mathtools}
\usepackage{diagbox}
\usepackage{hyperref}
\usepackage{cleveref}
\usepackage[utf8]{inputenc}
\usepackage[english]{babel}
\usepackage{scalerel}
\usepackage{stackengine}
\usepackage{quantikz}
\usepackage{algorithm}
\usepackage[noend]{algpseudocode} 
\usepackage{subfigure}
\usepackage{enumitem}
\usepackage{comment}
\usepackage{etoolbox}

\makeatletter
\patchcmd{\@bibdataout@aps}{author="08"}{author="48"}{}{%
  \PackageError{aps-bib-control}{Cannot patch long-bibliography author control}{}%
}
\patchcmd{\@bibdataout@aps}{author="08"}{author="48"}{}{%
  \PackageError{aps-bib-control}{Cannot patch short-bibliography author control}{}%
}
\makeatother

\newcommand{\algmargin}{\the\ALG@thistlm}
\makeatother
\newlength{\whilewidth}
\algdef{SE}[parWHILE]{parWhile}{EndparWhile}[1]
  {\parbox[t]{\dimexpr\linewidth-\algmargin}{%
     \hangindent\whilewidth\strut\algorithmicwhile\ #1\ \algorithmicdo\strut}}{\algorithmicend\ \algorithmicwhile}%
\algnewcommand{\parState}[1]{\State%
  \parbox[t]{\dimexpr\linewidth-\algmargin}{\strut #1\strut}}

\newtheorem{theorem}{Theorem}[]
\newtheorem*{remark}{Remark}
\newtheorem{corollary}{Corollary}[]
\newtheorem{lemma}[]{Lemma}
\newtheorem{prop}{Proposition}

\newtheorem{definition}{Definition}

\theoremstyle{definition}

\makeatletter
\@addtoreset{proofpart}{theorem}
\@addtoreset{proofcase}{theorem}
\makeatother
\renewcommand{\exp}[1]{\text{exp}\left( #1 \right)}

\begin{document}

\title{Complexity of tensor network simulation for noisy quantum circuits}

\author{Yuguo Shao}
\thanks{These authors contributed equally to this work.}
\affiliation{Yau Mathematical Sciences Center, Tsinghua University, Beijing 100084, China}
\affiliation{Beijing Institute of Mathematical Sciences and Applications, Beijing 100407, China}

\author{Zishuo Zhao}
\thanks{These authors contributed equally to this work.}
\affiliation{Yau Mathematical Sciences Center, Tsinghua University, Beijing 100084, China}
\affiliation{Department of Mathematics, Tsinghua University, Beijing 100084, China}

\author{Song Cheng}
\email{chengsong@bimsa.cn}
\affiliation{Beijing Institute of Mathematical Sciences and Applications, Beijing 100407, China}

\author{Zhengwei Liu}
\email{liuzhengwei@mail.tsinghua.edu.cn}
\affiliation{Yau Mathematical Sciences Center, Tsinghua University, Beijing 100084, China}
\affiliation{Beijing Institute of Mathematical Sciences and Applications, Beijing 100407, China}


\begin{abstract}

    We aim to rigorously address how local noise affects the classical simulability of quantum dynamics benchmarked by tensor-network methods. 
    Using operator entanglement entropy (OEE) as a unifying criterion, we prove the following:
    (1) For single-qubit depolarizing noise, geometry-independent $\order{1}$-layer and $\Theta(\log n)$-layer thresholds for fixed absolute and relative Hilbert--Schmidt error are established, beyond which polynomial bond dimension suffices; on 1D local circuits with a product initial input, whole-trajectory MPO approximations are constructed at fixed purified distance with $\mathrm{poly}(n)$ bond dimension.
    (2) For general single-qubit noise, we identify sufficient contractivity conditions for MPO approximations with constant bond dimension, in both the average and worst cases.
    (3) We further extend the above results to higher dimensions, yielding $\mathrm{poly}(n)$ average boundary-bond dimensions for projected entangled pair operators (PEPO) across every cut.
    Our results establish a rigorous connection between certain noise models, circuit types, and their classical simulability.
\end{abstract}

\maketitle

\section{Introduction}\label{sec:intro}
The boundary of classical simulability for noisy quantum dynamics is a central question of the NISQ era. Local noise suppresses coherences and drives the dynamics toward lower-complexity structures, reducing the cost of classical tensor-network (TN) methods~\cite{berezutskii2025tensor,noh2020efficient,cheng2021simulating,wellnitz2022rise,zhang2022entanglement,li2023entanglement,wei2026noise,yan2025limitations,lee2025classical}.
Whether the dynamics is representable by a matrix product operator~(MPO) of moderate bond dimension, however, depends on the error criterion and on the noise model~\cite{vidal2004efficient,verstraete2004mpdo,zwolak2004mixed,schuch2008entropy,schollwock2011density,orus2014practical,wellnitz2022rise}. The simulability boundary is therefore a family of thresholds, jointly set by accuracy target and noise class.

Operator entanglement serves as the natural diagnostic for these thresholds.
Vectorizing a density operator as a pure state in the doubled Hilbert space maps an MPO to a matrix product state (MPS), whose bond dimension across a bipartition is controlled by the entropy of the operator Schmidt coefficients~\cite{noh2020efficient,dowling2024operational,dowling2025operatorbridge}.
Two distinct entropies are relevant here: the unnormalized operator entanglement $S_{OE}$, which incorporates purity decay and governs absolute Hilbert-Schmidt truncation, and its purity-rescaled normalized counterpart $\widetilde S_{OE}$, which governs relative-error truncation~\cite{noh2020efficient}.

Recent work on the classical simulability of noisy quantum circuits has followed two complementary lines. The first uses operator entanglement as the cost proxy for mixed-state TN simulation and studies its behavior in specific noise models and circuit classes~\cite{noh2020efficient,wellnitz2022rise,zhang2022entanglement,li2023entanglement,dowling2024operational,dowling2025operatorbridge,cheng2021simulating,dowling2026classical,wei2026noise}. The second, built around quasiprobability decompositions, Pauli-path methods, and observable-backpropagation techniques, addresses tasks including observable estimation and output sampling~\cite{bennink2017unbiased,aharonov2023polynomial,gonzalezgarcia2025pauli,shao2024obppp,angrisani2025arbitrarynoise,fontana2025variational,shao2025diagnosing,mele2026noise,gao2018efficient,angrisani2025classically,schuster2025polynomial,martinez2025efficient,denzler2026simulation,dowling2026noise}. Under depolarizing noise, Pauli-path methods enable efficient sampling for typical random circuits~\cite{aharonov2023polynomial} and observable estimation for arbitrary circuits with errors bounded on average over input states~\cite{schuster2025polynomial}. In this work we contribute to the first line by proving rigorous operator-entanglement bounds — across multiple noise classes and both absolute and relative accuracy criteria — and translating these into explicit MPO and projected entangled pair operators~(PEPO) bond dimension estimates. In particular, for depolarizing noise we obtain an efficient MPO simulation of the full quantum state on a 1D chain, without averaging over circuits or inputs.

\begin{figure}[t]
    \centering
    \includegraphics[width=0.45\textwidth]{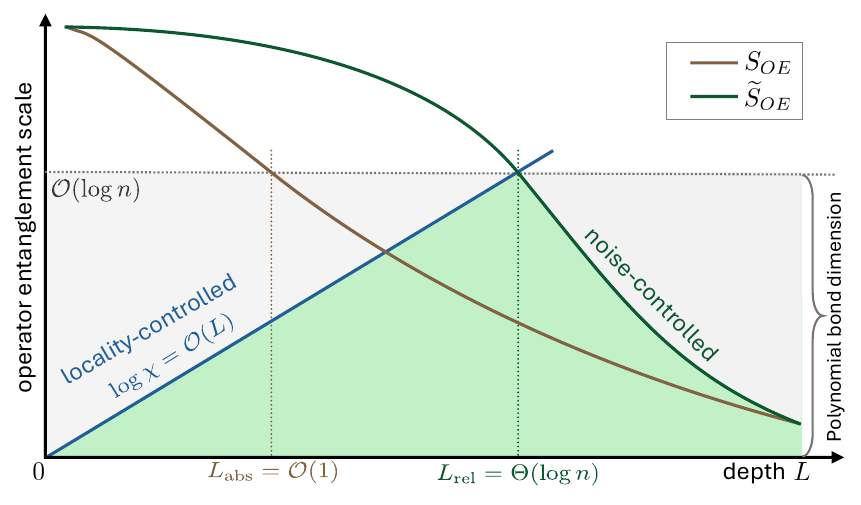}
    \caption{Depolarizing noise guarantees polynomial bond dimension of TN simulation at all circuit depths.
    For arbitrary local circuits, the bond dimension is controlled by locality before the respective crossover and by the noise-induced OEE thresholds after it:
    $L_{\mathrm{abs}}=\order{1}$ for absolute Hilbert-Schmidt error and $L_{\mathrm{rel}}=\Theta(\log n)$ for relative error, both geometry independent.
    }
    \label{fig:main_thresholds}
\end{figure}

As summarized in \cref{tab:main_results}, the bounds reveal a simple pattern: noise-induced suppression of operator entanglement lowers TN's complexity.
For single-qubit depolarizing noise, we find a geometry-independent crossover: polynomial-bond TN descriptions at fixed absolute error are guaranteed after $\order{1}$ layers, whereas fixed relative error requires $\Theta(\log n)$ layers. 
In one dimension, combining the locality-controlled bound with the late-time noise-controlled bound yields a constructive $\mathrm{poly}(n)$-bond MPO approximation of the entire noisy trajectory from a product initial state, with controlled purified distance.
For general single-qubit noise, the contraction coefficient sets the boundary: $c<1$ yields an $\order{1}$ OEE plateau for random $2$-design brickwall circuits, while $c<1/6$ yields an $\order{1}$ worst-case bound for arbitrary gates.
The same cutwise entropy estimates also extend to higher-dimensional local circuits, giving polynomial average PEPO boundary-bond dimensions.

\begin{table*}[t]
\footnotesize
\centering
\setlength{\tabcolsep}{4pt}
\renewcommand{\arraystretch}{1.15}
\begin{tabular}{@{}lll@{}}
    \hline\hline
    \textbf{Noise Conditions} & \textbf{Circuit \& Gate Conditions} & \textbf{Main Bounds for Simulability} \\
    \hline
    \multicolumn{3}{c}{\textbf{Operator entanglement entropy}} \\
    \hline
    single-qubit depolarizing & any gates (geometry independent) &
    \begin{tabular}[t]{@{}l@{}}
    $S_{OE}=\order{\log n}$ for $L\geq \order{1}$;\\
    $\widetilde S_{OE}=\order{\log n}$ for $L\geq \Theta(\log n)$; Thm.~\ref{thm:main_depolarizing_thresholds}
    \end{tabular} \\
    general single-qubit, $c<1$ & 1D brickwall, 2-design gates & $S_{OE}=\order{1}$; Thm.~\ref{thm:main_average_case_plateau} \\
    general single-qubit, $c<1/6$ & any gates (geometry independent) & $S_{OE}=\order{1}$; Thm.~\ref{thm:main_worst_case_purity} \\
    \hline\hline
    \multicolumn{3}{c}{\textbf{Average boundary-bond dimension $\overline\chi_\partial$ in higher dimensions}} \\
    \hline
    single-qubit depolarizing & local circuits & abs/rel: $\log\overline\chi_\partial=\order{\log n}$  \\
    general single-qubit, $c<1/6$ & local circuits & abs: $\overline\chi_\partial=\order{1}$ after one noise layer \\
    \hline\hline
\end{tabular}
\caption{Main results for operator-entanglement bounds. The upper block gives the OEE scaling for each noise class and circuit condition, which controls the required bond dimension of TN. The lower block records the corresponding average PEPO boundary-bond dimension $\overline\chi_\partial$ in higher dimensions.}
\label{tab:main_results}
\end{table*}

\section{Preliminaries}\label{sec:preliminaries}
We consider an $n$-qubit density matrix $\rho$ and a bipartition $A\mid B$. The operator Schmidt decomposition of $\rho$ across this cut reads
\begin{equation}\label{eq:prelim_operator_schmidt}
    \rho
    =
    \sum_{\alpha=1}^{r}
    \lambda_\alpha L_{\alpha}^{[A]} \otimes R_{\alpha}^{[B]},
\end{equation}
where $\{L_{\alpha}^{[A]}\}\subset\mathcal{B}(\mathcal{H}_A)$ and $\{R_{\alpha}^{[B]}\}\subset\mathcal{B}(\mathcal{H}_B)$ are Hilbert-Schmidt orthonormal operators on $A$ and $B$, $r$ is the operator Schmidt rank and $\lambda_\alpha>0$ are the Schmidt coefficients. 

Following Ref.~\cite{wellnitz2022rise}, the operator entanglement entropy~(OEE) is defined as
\begin{equation}\label{eq:prelim_oee}
    S_{OE}(\rho)
    =
    -\sum_{\alpha=1}^{r}
    \left(\lambda_\alpha\right)^2
    \log_2 \left(\lambda_\alpha\right)^2 .
\end{equation}
Normalizing this spectrum by the purity gives the normalized OEE
\begin{equation}
    \widetilde{S}_{OE}(\rho)
    =
    -\sum_{\alpha=1}^{r} p_\alpha \log_2 p_\alpha,
    \quad
    p_\alpha=\frac{\lambda_\alpha^2}{\tr{\rho^2}},
\end{equation}
which is the MPO entanglement entropy of Ref.~\cite{noh2020efficient}; it removes the overall purity loss and measures only the structure of the operator weight. 
Equivalently, $\tilde{S}_{OE}(\rho)$ is the entanglement entropy of the vectorized state $\tr{\rho^2}^{-1/2}\|\rho\rangle\!\rangle$ in the doubled Hilbert space.  
The two entropies obey
\begin{equation}\label{eq:prelim_oee_relation}
    S_{OE}(\rho) = \tr{\rho^2} \widetilde{S}_{OE}(\rho) - \tr{\rho^2}\log_2 \tr{\rho^2}.
\end{equation}

This spectrum $\lambda_\alpha$ is also the full-environment of tensor networks across A and B, the two entropies $S_{OE}$ and $\widetilde S_{OE}$ thus optimally govern the absolute and relative truncation error of the operator Schmidt decomposition, respectively~\cite{vidal2003efficient,vidal2004efficient,zwolak2004mixed,schuch2008entropy,schollwock2011density,orus2014practical,dowling2026classical}. 
As shown in \cref{sec:complexity_tensor_network}, a relative truncation error $\delta\in(0,1)$ is guaranteed by any bond dimension $\chi$ satisfying
\begin{equation}\label{eq:prelim_bond_dimension_scaling_rel}
    \chi \geq \max\left\{1,\left\lceil (1-\delta)2^{\frac{\widetilde{S}_{OE}(\rho)}{\delta}}\right\rceil\right\},
\end{equation}
whereas a fixed absolute truncation error $0<\varepsilon<\tr{\rho^2}$ is guaranteed by any
\begin{equation}\label{eq:prelim_bond_dimension_scaling_abs}
    \chi \geq \max\left\{1,\left\lceil \left(\tr{\rho^2}-\varepsilon\right) 2^{\frac{S_{OE}(\rho)}{\varepsilon}}\right\rceil\right\}.
\end{equation}
Roughly, once the error tolerance is fixed, the required bond dimension scales exponentially in the relevant operator-entanglement scale,
\begin{equation}\label{eq:prelim_bond_dimension_scaling}
    \chi_{\mathrm{rel}} = \order{2^{\alpha \widetilde{S}_{OE}}}, \qquad
    \chi_{\mathrm{abs}} = \order{2^{\alpha S_{OE}}}.
\end{equation}
up to constants set by the chosen tolerance.

Polynomial-$\chi$ tensor-network approximations are therefore guaranteed by at most logarithmic OEE.
However, the OEE of a local circuit in general case grows linearly with the circuit depth~\cite{nahum2017quantum,vonkeyserlingk2018operator}. 

By \textit{local circuit} we mean a quantum circuit on a bounded-degree interaction graph where each gate acts on $\order{1}$ qubits.
Consequently, for any bipartition $A\mid A^c$, each layer contains at most $\order{a(A)}$ gates crossing the cut, where $a(A)$ is the number of boundary bonds.

In the following, we investigate under what conditions the operator entanglement would be bounded to this logarithmic scale.

\section{Dissipation driven by depolarizing noise}\label{sec:depolarizing}
We consider an $n$-qubit circuit of depth $L$ built from noisy layers
\begin{equation}
    \Phi_{\ell}=\mathcal D_{\lambda}^{\otimes n}\circ\mathcal U_{\ell},
    \qquad \ell=1,\ldots,L,
\end{equation}
where $\mathcal U_{\ell}(\cdot)=U_{\ell}(\cdot)U_{\ell}^{\dagger}$ is the unitary layer and $\mathcal D_{\lambda}^{\otimes n}$ is independent single-qubit depolarizing noise, with
\begin{equation}
    \mathcal{D}_{\lambda}(\sigma)=(1-\lambda)\sigma+\lambda\frac{\mathbb{I}}{2}.
\end{equation}
For an initial $n$-qubit state $\rho_0$, we write
\begin{equation}
    \rho_{\ell}=\Phi_{\ell}(\rho_{\ell-1}),
    \qquad \ell=1,\ldots,L,
\end{equation}
so that $\rho_\ell$ is the state after the $\ell$-th unitary layer and the following noise layer. The final state is
\begin{equation}
    \rho_L=\left(\Phi_L\circ\cdots\circ\Phi_1\right)(\rho_0).
\end{equation}
At this stage the unitary layers are otherwise unrestricted; locality, randomness or brickwall structure will be specified only when it is used.

Depolarizing noise simultaneously dissipates OEE and reduces purity. The two effects enter approximation differently: absolute error benefits directly from the purity decay, whereas relative error asks only whether the normalized operator structure has simplified. The two accuracy criteria therefore reach efficient simulability at asymptotically different depths.
\begin{theorem}[Depolarization-induced separation of simulation thresholds]\label{thm:main_depolarizing_thresholds}
For the above depth-$L$ depolarizing circuit with fixed strength $\lambda>0$, fix an arbitrary bipartition $A\mid B$ of the qubits. The OEE of $\rho_\ell$ exhibits two separated crossover depths,
\begin{equation}
    L_{\mathrm{abs}} = \order{1}
    \qquad
    L_{\mathrm{rel}} = \Theta(\log n),
\end{equation}
such that $S_{OE}(\rho_\ell)=\order{\log n}$ when $\ell\geq  L_{\mathrm{abs}}$, whereas $\widetilde{S}_{OE}(\rho_\ell)=\order{\log n}$ when $\ell\geq  L_{\mathrm{rel}}$. Moreover, the logarithmic scale of $L_{\mathrm{rel}}$ is optimal for achieving logarithmic $\widetilde{S}_{OE}(\rho_\ell)$.
\end{theorem}

The proof has two main steps. First, hypercontractivity~\cite{king2014hypercontractivity} together with unitary invariance of Schatten norms yields a geometry-independent description of purity decay. Second, purity-controlled maximum-OEE bounds convert this decay into logarithmic OEE bounds at the crossover depths.

For a product initial state on a 1D chain, locality bounds the operator Schmidt rank after $L$ layers by $\exp{\order{L}}$.
Combined with \cref{thm:main_depolarizing_thresholds}, this ensures that both OEEs remain at most $\order{\log n}$ throughout the trajectory, as illustrated in \cref{fig:main_thresholds}.
The following proposition gives a constructive MPO approximation with polynomial bond dimension and controlled error at every layer.
We quantify the accuracy of physical states by the purified distance $P(\rho,\sigma)=\sqrt{1-F(\rho,\sigma)^2}$, where $F(\rho,\sigma)=\norm{\sqrt\rho\sqrt\sigma}_1$.

\begin{prop}[Whole-trajectory noisy circuit simulation]
\label{prop:main_whole_trajectory_rel_depolarizing}
Consider an $n$-qubit noisy trajectory on a one-dimensional local circuit,
\begin{equation}
    \rho_{\ell}=\mathcal D_{\lambda}^{\otimes n}\circ \mathcal U_{\ell}(\rho_{\ell-1}),
    \qquad
    \rho_0=\ketbra{0}{0}^{\otimes n},
\end{equation}
where $\lambda\in(0,1)$ is a fixed depolarizing strength, $\mathcal U_{\ell}(\cdot)=U_{\ell}(\cdot)U_{\ell}^{\dagger}$, and each $U_{\ell}$ is a product of disjoint local gates from the $\ell$-th layer.
For any error tolerance $\varepsilon\in(0,1)$, a constructive sequential MPO simulation produces density operators $\{\widehat\rho_{\ell}\}_{\ell=0}^{L}$ with maximum bond dimension $\mathrm{poly}(n,1/\varepsilon)$, such that, for every $\ell$,
\begin{equation}
    P(\rho_\ell,\widehat\rho_\ell)\leq\varepsilon,\qquad
    \norm{\rho_{\ell}-\widehat\rho_{\ell}}_2^2\leq\varepsilon^2\norm{\rho_\ell}_2^2.
\end{equation}
\end{prop}

The construction in \cref{app:whole_trajectory} uses exact MPO evolution up to depth $T=\order{\log(n/\varepsilon)}$ and outputs $\mathbb I/2^n$ thereafter.
Locality ensures polynomial bond dimension during the exact evolution, while noise guarantees both error bounds at later depths.
The purified-distance guarantee also implies $\norm{\rho_\ell-\widehat\rho_\ell}_1\leq2\varepsilon$ and $\abs{\tr[O(\rho_\ell-\widehat\rho_\ell)]}\leq2\varepsilon\norm{O}_\infty$ for every observable $O$.

\section{Dissipation driven by general noise}\label{sec:general_noise}
Unlike depolarization, generic local noise is neither isotropic nor necessarily purity-reducing. For random 1D brickwall circuits, gate averaging nevertheless yields a constant OEE plateau under non-unitary noise. Under sufficiently strong noise, purity suppression also gives a geometry-independent worst-case bound for arbitrary unitary layers.

Each time step applies $\Phi_\ell=\mathcal{N}^{\otimes n}\circ\mathcal{U}_\ell$, where $\mathcal U_\ell$ is the unitary layer. The same single-qubit noise channel $\mathcal N$ is applied to every site after each unitary layer. Starting from an initial state $\rho_0$, the trajectory is $\rho_L=\Phi_L\circ\cdots\circ\Phi_1(\rho_0)$.

After pre- and post-unitary rotations $\mathcal{N}(\cdot) = U \mathcal{N}'(V^\dagger (\cdot) V) U^\dagger$, any single-qubit channel admits the canonical form~\cite{king2001minimal,ruskai2002analysis} of the Pauli transfer matrix $\mathcal{S}(\mathcal{N}')_{ij}=\frac{1}{2}\tr{\sigma_i\,\mathcal{N}'(\sigma_j)}$, with $(\sigma_0,\sigma_1,\sigma_2,\sigma_3)=(\mathbb{I},X,Y,Z)$,
\begin{equation}
    \mathcal{S}(\mathcal{N}')=\begin{pmatrix}
        1 & 0 & 0 & 0 \\
        t_X & D_X & 0 & 0 \\
        t_Y & 0 & D_Y & 0 \\
        t_Z & 0 & 0 & D_Z
    \end{pmatrix}.
\end{equation}
The noise strength is quantified by the contraction coefficient~\cite{mele2026noise}
\begin{equation}
    c(\mathcal{N})=\frac{1}{3}\left(t_X^2+t_Y^2+t_Z^2+D_X^2+D_Y^2+D_Z^2\right),
\end{equation}
$D_i$ are damping coefficients, $t_i$ are displacements, and all $D_i$ share the same sign. It satisfies $c(\mathcal{N})\leq 1$, with equality iff $\mathcal{N}$ is unitary; smaller $c$ means stronger suppression of Pauli coefficients per layer.

\subsection{Average case}
For the average-case dynamics, we consider a one-dimensional noisy brickwall circuit — a concrete instance of a local circuit — built from staggered layers of nearest-neighbour two-qubit gates. The randomness is local: every gate is drawn independently from a unitary $2$-design, and the same single-qubit channel $\mathcal N$ is applied on all sites after each layer.
This geometry gives the unitary dynamics a ballistic light cone~\cite{nachtergaele2010lieb,poulin2010lieb,barthel2012quasilocality}, while the noise continually contracts the local operator components carried along it. For any non-unitary channel, the averaged dynamics admits a system-size independent OEE plateau.

\begin{theorem}[Average-case OEE plateau in brickwall circuits]\label{thm:main_average_case_plateau}
Consider a 1D noisy brickwall circuit with all gates drawn independently from a unitary $2$-design, a product initial state, and a fixed arbitrary cut of the chain.
If $c(\mathcal{N})<1$, then with probability at least $1-Le^{-\Omega(n)}$ over the choice of gates,
\begin{equation}
    \sup_{\ell \in \{1,\dots,L\}} S_{OE}(\rho_\ell) = \order{1}.
\end{equation}
\end{theorem}

Since $c(\mathcal N)=1$ holds exactly for unitary single-qubit channels, the theorem applies to every fixed non-unitary single-qubit channel.
Theorem~\ref{thm:main_average_case_plateau} shows that, for $L = \order{\mathrm{poly}(n)}$, there is an area law of OEE with high probability across the full parameter space, induced by local noise.

This route to low OEE differs from the depolarizing case by not requiring relaxation to the maximally mixed state. The proof in \cref{app:general_average_case} bounds the purity averaged over random gates and uses it to control OEE. Pure state replacement channels are handled separately, since their outputs are product states with zero OEE.

\subsection{Worst case}
For arbitrary gates, sufficiently strong noise bounds the unnormalized OEE directly through purity. This argument does not require locality or averaging over gates.

\begin{theorem}[Worst-case OEE bound under strong contraction]\label{thm:main_worst_case_purity}
Let $n\geq2$ and let the fixed single-qubit channel $\mathcal N$ satisfy $c(\mathcal N)<1/6$.
For any initial state, arbitrary unitary layers, and any bipartition, the output at depth $\ell\geq1$ satisfies
\begin{equation}
    S_{OE}(\rho_\ell)=\order{1}.
\end{equation}
\end{theorem}

In the two-copy swap representation of purity, strong contraction bounds each local noise contribution by a factor smaller than one. The tensor-product noise therefore makes purity exponentially small in $n$, and the purity-controlled OEE bound converts this into a constant bound on $S_{OE}$. This argument requires neither gate averaging nor geometric locality.

\section{Implications for average PEPO boundary-bond dimension}\label{sec:pepo}
The higher dimensional analogue of MPO is the projected entangled-pair operator~(PEPO) \cite{Ciracetal2020}. Consider a PEPO on a graph and a bipartition $A\mid A^c$. Let $a(A)$ be the number of virtual PEPO bonds crossing the cut. 
Let $R$ denote the operator Schmidt rank across $A\mid A^c$, we define the average boundary-bond dimension as
\begin{equation}
    \overline\chi_{\partial}(A)=R^{1/a(A)} .
\end{equation}
This is the operator-space analogue of the area-law of the entanglement entropy in the projected entangled-pair states~(PEPS) with boundary bond dimension $\chi_{\partial}$, where the Schmidt-rank is at most $\chi_{\partial}^{a(A)}$ across the cut~\cite{Ciracetal2020,EntanglementAreaLaws2010}.
In the PEPO setting, $\overline\chi_{\partial}(A)$ should therefore be understood as a cutwise representability scale for the vectorized density operator.
At fixed truncation tolerance, the OEE-to-rank truncation bounds (\cref{eq:prelim_bond_dimension_scaling_abs,eq:prelim_bond_dimension_scaling_rel}) imply that such an approximate representation can be achieved with average boundary-bond dimension scale
\begin{equation}
    \log \overline\chi_{\partial}(A)=\order{\widetilde S_{OE}(A)/a(A)}
\end{equation}
for fixed relative accuracy, and with $S_{OE}$ in place of $\widetilde S_{OE}$ for fixed absolute accuracy.

For single-qubit depolarizing noises, the geometry-independent entropy estimates of \cref{thm:main_depolarizing_thresholds} apply to any PEPO cut and retain the same separation between absolute and relative crossover depths.
This gives a uniform-in-depth boundary-dimension bound once combined with locality before the crossover.
Indeed, a depth-$L$ local circuit can increase $\log\overline\chi_{\partial}(A)$ by at most $\order{L}$ per boundary bond; the pre-crossover windows needed are $L_{\mathrm{abs}}=\order{1}$ for absolute accuracy and $L_{\mathrm{rel}}=\order{\log n}$ for relative accuracy.
After the corresponding crossover, the entropy estimates give $\log\overline\chi_{\partial}(A)=\order{\log n/a(A)}\leq \order{\log n}$.
Hence, for every depth, every nontrivial cut $A\mid A^c$, and either fixed error criterion,
\begin{equation}
    \overline\chi_{\partial}^{\mathrm{abs/rel}}(A) = \mathrm{poly}(n).
\end{equation}

For general local noise, the worst-case OEE bound under the strong-contraction assumption $c(\mathcal N)<1/6$ gives $S_{OE}(\rho_\ell)=\order{1}$.
Therefore fixed absolute Hilbert-Schmidt accuracy requires
\begin{equation}
    \overline\chi_{\partial}^{\mathrm{abs}}(A)=\order{1},
\end{equation}
across every nontrivial cut.

Unlike 1D, these cutwise estimates of average PEPO boundary-bond dimension do not directly imply an efficient simulation. Nevertheless, this result remains significant as it provides the potential of a compact representation for entanglement structure of many-body systems.

\section{Discussion}\label{sec:discussion}
Our results establish a rigorous connection between noise, circuits, operator entanglement, and classical simulability via TN. For depolarizing noise this gives geometry-independent thresholds: $\order{1}$ depth for absolute accuracy and $\order{\log n}$ depth for relative accuracy. Consequently, it allows efficient whole-trajectory MPO simulation for 1D circuits with a product initial state. This is a worst-case bound for fixed circuits and a fixed input state, in contrast to previous Pauli-path-based methods that typically require averaging over circuit or input ensembles~\cite{aharonov2023polynomial,schuster2025polynomial}.

For general single-qubit noise, the average-case and worst-case conclusions have different sufficient conditions. Random $2$-design brickwall circuits exhibit an $\order{1}$ OEE plateau with high probability for any non-unitary channel. For arbitrary gates, $c<1/6$ guarantees a geometry-independent $\order{1}$ OEE bound. This worst-case bound also gives constant cutwise average PEPO boundary dimensions at fixed absolute accuracy.

Both noise analyses use purity estimates to control OEE. For depolarizing noise, hypercontractivity~\cite{king2014hypercontractivity} controls the depth dependence of purity and certifies the whole-trajectory MPO approximation. For general single-qubit noise, we derive purity bounds from a two-copy representation, using gate averaging for random circuits and operator-norm estimates for arbitrary gates. These bounds yield the corresponding $\order{1}$ OEE bounds.

The worst-case threshold $c<1/6$ may not be optimal. The bounds carry implicit prefactors that can become large near their thresholds or demanding tolerances, potentially exceeding feasible bond dimensions even for moderate $n$. The general-noise analysis is restricted to single-qubit noise. Long-range interactions and coherent noise remain open. Finally, our depolarizing results assume i.i.d.~noise.

Future work could explore several directions. Identifying channel-level features—anisotropy, contraction directions—that dissipate operator entanglement would clarify which noise structures simplify classical simulation. On the algorithmic side, a more practical adaptive truncation scheme could reduce the cost of the exact initial segment while certifying accumulated errors throughout the trajectory. Sharpening the sufficient threshold $c<1/6$ would further clarify the complexity phase diagram for noisy dynamics. Although the whole-trajectory construction already bounds observable errors, methods tailored to a small number of observables could further reduce the simulation cost, for instance in shallow tomography. Progress on these questions would move noisy TN simulability from a collection of model-specific bounds toward a more unified theory of how dissipation reshapes quantum complexity.

\begin{acknowledgments}
We thank Zi-wen Liu, Fuchuan Wei and Jinsong Wu for helpful comments on the manuscript.
S.C. was supported by the National Science Foundation of China (Grant No. 12004205 and Grant No. 12574253). 
Z.L. was supported by NKPs (Grant No. 2020YFA0713000).
Y.S., Z.Z. and Z.L. were supported by BMSTC and ACZSP (Grant No. Z221100002722017). 
S.C. and Z.L. were supported by Beijing Natural Science Foundation (Grant No. Z220002).
\end{acknowledgments}

\bibliography{main}

@article{wellnitz2022rise,
  title = {Rise and Fall, and Slow Rise Again, of Operator Entanglement under Dephasing},
  author = {Wellnitz, D. and Preisser, G. and Alba, V. and Dubail, J. and Schachenmayer, J.},
  journal = {Phys. Rev. Lett.},
  volume = {129},
  issue = {17},
  pages = {170401},
  numpages = {8},
  year = {2022},
  month = {Oct},
  publisher = {American Physical Society},
  doi = {10.1103/PhysRevLett.129.170401},
  url = {https://link.aps.org/doi/10.1103/PhysRevLett.129.170401}
}

@article{eckart1936approximation,
  title={The approximation of one matrix by another of lower rank},
  author={Eckart, Carl and Young, Gale},
  journal={Psychometrika},
  volume={1},
  number={3},
  pages={211--218},
  year={1936},
  publisher={Springer-Verlag},
  doi={10.1007/BF02288367},
  url={https://doi.org/10.1007/BF02288367}
}

@article{vidal2003efficient,
  title = {Efficient Classical Simulation of Slightly Entangled Quantum Computations},
  author = {Vidal, Guifr\'{e}},
  journal = {Phys. Rev. Lett.},
  volume = {91},
  issue = {14},
  pages = {147902},
  year = {2003},
  month = {Oct},
  publisher = {American Physical Society},
  doi = {10.1103/PhysRevLett.91.147902},
  url = {https://link.aps.org/doi/10.1103/PhysRevLett.91.147902}
}

@article{vidal2004efficient,
  title = {Efficient Simulation of One-Dimensional Quantum Many-Body Systems},
  author = {Vidal, Guifr\'{e}},
  journal = {Phys. Rev. Lett.},
  volume = {93},
  issue = {4},
  pages = {040502},
  year = {2004},
  month = {Jul},
  publisher = {American Physical Society},
  doi = {10.1103/PhysRevLett.93.040502},
  url = {https://link.aps.org/doi/10.1103/PhysRevLett.93.040502}
}

@article{zwolak2004mixed,
  title = {Mixed-State Dynamics in One-Dimensional Quantum Lattice Systems: A Time-Dependent Superoperator Renormalization Algorithm},
  author = {Zwolak, Michael and Vidal, Guifr\'{e}},
  journal = {Phys. Rev. Lett.},
  volume = {93},
  issue = {20},
  pages = {207205},
  year = {2004},
  month = {Nov},
  publisher = {American Physical Society},
  doi = {10.1103/PhysRevLett.93.207205},
  url = {https://link.aps.org/doi/10.1103/PhysRevLett.93.207205}
}

@article{verstraete2004mpdo,
  title = {Matrix Product Density Operators: Simulation of Finite-Temperature and Dissipative Systems},
  author = {Verstraete, Frank and Garc\'{i}a-Ripoll, Juan J. and Cirac, J. Ignacio},
  journal = {Phys. Rev. Lett.},
  volume = {93},
  issue = {20},
  pages = {207204},
  year = {2004},
  month = {Nov},
  publisher = {American Physical Society},
  doi = {10.1103/PhysRevLett.93.207204},
  url = {https://link.aps.org/doi/10.1103/PhysRevLett.93.207204}
}

@article{schuch2008entropy,
  title = {Entropy Scaling and Simulability by Matrix Product States},
  author = {Schuch, Norbert and Wolf, Michael M. and Verstraete, Frank and Cirac, J. Ignacio},
  journal = {Phys. Rev. Lett.},
  volume = {100},
  issue = {3},
  pages = {030504},
  year = {2008},
  month = {Jan},
  publisher = {American Physical Society},
  doi = {10.1103/PhysRevLett.100.030504},
  url = {https://link.aps.org/doi/10.1103/PhysRevLett.100.030504}
}

@article{schollwock2011density,
  title = {The density-matrix renormalization group in the age of matrix product states},
  author = {Schollw{\"o}ck, Ulrich},
  journal = {Annals of Physics},
  volume = {326},
  number = {1},
  pages = {96--192},
  year = {2011},
  doi = {10.1016/j.aop.2010.09.012},
  url = {https://doi.org/10.1016/j.aop.2010.09.012}
}

@article{orus2014practical,
  title = {A Practical Introduction to Tensor Networks: Matrix Product States and Projected Entangled Pair States},
  author = {Or\'{u}s, Rom\'{a}n},
  journal = {Annals of Physics},
  volume = {349},
  pages = {117--158},
  year = {2014},
  doi = {10.1016/j.aop.2014.06.013},
  url = {https://doi.org/10.1016/j.aop.2014.06.013}
}

@article{kshetrimayum2017simple,
  title = {A simple tensor network algorithm for two-dimensional steady states},
  author = {Kshetrimayum, Augustine and Weimer, Hendrik and Or\'{u}s, Rom\'{a}n},
  journal = {Nature Communications},
  volume = {8},
  number = {1},
  pages = {1291},
  year = {2017},
  doi = {10.1038/s41467-017-01511-6},
  url = {https://doi.org/10.1038/s41467-017-01511-6}
}

@article{noh2020efficient,
  author = {Noh, Kyungjoo and Jiang, Liang and Fefferman, Bill},
  title = {Efficient Classical Simulation of Noisy Random Quantum Circuits in One Dimension},
  journal = {Quantum},
  volume = {4},
  pages = {318},
  year = {2020},
  url = {https://quantum-journal.org/papers/q-2020-09-11-318/},
  urldate = {2026-03-17},
  doi = {10.22331/q-2020-09-11-318}
}

@article{bennink2017unbiased,
  author = {Bennink, Ryan S. and Ferragut, Erik M. and Humble, Travis S. and Laska, Jason A. and Nutaro, James J. and Pleszkoch, Mark G. and Pooser, Raphael C.},
  title = {Unbiased simulation of near-Clifford quantum circuits},
  journal = {Phys. Rev. A},
  volume = {95},
  issue = {6},
  pages = {062337},
  year = {2017},
  month = {Jun},
  publisher = {American Physical Society},
  doi = {10.1103/PhysRevA.95.062337},
  url = {https://doi.org/10.1103/PhysRevA.95.062337}
}

@inproceedings{aharonov2023polynomial,
  author = {Aharonov, Dorit and Gao, Xun and Landau, Zeph and Liu, Yunchao and Vazirani, Umesh},
  title = {A Polynomial-Time Classical Algorithm for Noisy Random Circuit Sampling},
  booktitle = {Proceedings of the 55th Annual ACM Symposium on Theory of Computing},
  series = {STOC '23},
  pages = {945--957},
  year = {2023},
  publisher = {Association for Computing Machinery},
  doi = {10.1145/3564246.3585234},
  url = {https://doi.org/10.1145/3564246.3585234}
}

@article{gonzalezgarcia2025pauli,
  author = {Gonz\'alez-Garc\'ia, Guillermo and Cirac, J. Ignacio and Trivedi, Rahul},
  title = {Pauli path simulations of noisy quantum circuits beyond average case},
  journal = {Quantum},
  volume = {9},
  pages = {1730},
  year = {2025},
  doi = {10.22331/q-2025-05-05-1730},
  url = {https://doi.org/10.22331/q-2025-05-05-1730}
}

@article{angrisani2025arbitrarynoise,
  title = {Simulating Quantum Circuits with Arbitrary Local Noise Using Pauli Propagation},
  author = {Angrisani, Armando and Mele, Antonio A. and Rudolph, Manuel S. and Cerezo, M. and Holmes, Zo\"e},
  journal = {PRX Quantum},
  volume = {7},
  issue = {2},
  pages = {020313},
  numpages = {46},
  year = {2026},
  month = {Apr},
  publisher = {American Physical Society},
  doi = {10.1103/fb28-wlv2},
  url = {https://link.aps.org/doi/10.1103/fb28-wlv2}
}

@article{fontana2025variational,
  author = {Fontana, Enrico and Rudolph, Manuel S. and Duncan, Ross and Rungger, Ivan and C\^{\i}rstoiu, Cristina},
  title = {Classical simulations of noisy variational quantum circuits},
  journal = {npj Quantum Information},
  volume = {11},
  pages = {84},
  year = {2025},
  doi = {10.1038/s41534-024-00955-1},
  url = {https://doi.org/10.1038/s41534-024-00955-1}
}

@article{shao2024obppp,
  author = {Shao, Yuguo and Wei, Fuchuan and Cheng, Song and Liu, Zhengwei},
  title = {Simulating Noisy Variational Quantum Algorithms: A Polynomial Approach},
  journal = {Phys. Rev. Lett.},
  volume = {133},
  issue = {12},
  pages = {120603},
  year = {2024},
  month = {Sep},
  publisher = {American Physical Society},
  doi = {10.1103/PhysRevLett.133.120603},
  url = {https://doi.org/10.1103/PhysRevLett.133.120603}
}

@article{shao2025diagnosing,
  author = {Shao, Yuguo and Chen, Zhengyu and Wei, Zhaohui and Liu, Zhengwei},
  title = {Diagnosing Quantum Circuits: Noise Robustness, Trainability, and Expressibility},
  journal = {arXiv preprint arXiv:2509.11307},
  year = {2025},
  doi = {10.48550/arXiv.2509.11307},
  url = {https://arxiv.org/abs/2509.11307}
}

@article{dowling2024operational,
  author = {Dowling, Neil and Modi, Kavan},
  title = {Operational Metric for Quantum Chaos and the Corresponding Spatiotemporal-Entanglement Structure},
  journal = {PRX Quantum},
  volume = {5},
  pages = {010314},
  year = {2024},
  doi = {10.1103/PRXQuantum.5.010314},
  url = {https://doi.org/10.1103/PRXQuantum.5.010314}
}

@article{dowling2025operatorbridge,
  author = {Dowling, Neil and Modi, Kavan and White, Gregory A. L.},
  title = {Bridging Entanglement and Magic Resources within Operator Space},
  journal = {Phys. Rev. Lett.},
  volume = {135},
  pages = {160201},
  year = {2025},
  doi = {10.1103/c7k1-xcwy},
  url = {https://doi.org/10.1103/c7k1-xcwy}
}

@article{dowling2026noise,
  title={Noise-induced Simulability Transition from Operator Scrambling},
  author={Dowling, Neil and Turkeshi, Xhek and De Nardis, Jacopo and Lami, Guglielmo},
  journal={arXiv preprint arXiv:2605.18943},
  year={2026},
  doi={10.48550/arXiv.2605.18943},
  url={https://arxiv.org/abs/2605.18943}
}

@article{dowling2026classical,
  title={Classical simulability from operator entanglement scaling},
  author={Dowling, Neil},
  journal={arXiv preprint arXiv:2603.05656},
  year={2026},
  doi={10.48550/arXiv.2603.05656},
  url={https://arxiv.org/abs/2603.05656}
}

@article{DuttaFaulkner2021ReflectedEntropy,
    title={A canonical purification for the entanglement wedge cross-section},
    volume={2021},
    ISSN={1029-8479},
    url={http://dx.doi.org/10.1007/JHEP03(2021)178},
    DOI={10.1007/jhep03(2021)178},
    number={3},
    journal={Journal of High Energy Physics},
    publisher={Springer Science and Business Media LLC},
    author={Dutta, Souvik and Faulkner, Thomas},
    year={2021},
    month={mar},
    pages={178}
}

@article{king2001minimal,
  author = {King, C. and Ruskai, M. B.},
  title = {Minimal entropy of states emerging from noisy quantum channels},
  year = {2001},
  issue_date = {January 2001},
  publisher = {IEEE Press},
  volume = {47},
  number = {1},
  issn = {0018-9448},
  url = {https://doi.org/10.1109/18.904522},
  doi = {10.1109/18.904522},
  journal = {IEEE Trans. Inf. Theor.},
  month = jan,
  pages = {192–209},
  numpages = {18}
}

@article{ruskai2002analysis,
  title={An analysis of completely-positive trace-preserving maps on M2},
  author={Ruskai, Mary Beth and Szarek, Stanislaw and Werner, Elisabeth},
  journal={Linear algebra and its applications},
  volume={347},
  number={1-3},
  pages={159--187},
  year={2002},
  publisher={Elsevier},
  doi={10.1016/S0024-3795(01)00547-X},
  url={https://doi.org/10.1016/S0024-3795(01)00547-X}
}

@article{mele2026noise,
  title={Noise-induced shallow circuits and the absence of barren plateaus},
  author={Mele, Antonio Anna and Angrisani, Armando and Ghosh, Soumik and Khatri, Sumeet and Eisert, Jens and Stilck Fran{\c{c}}a, Daniel and Quek, Yihui},
  journal={Nature Physics},
  pages={1--6},
  year={2026},
  publisher={Nature Publishing Group UK London},
  doi={10.1038/s41567-026-03245-z},
  url={https://doi.org/10.1038/s41567-026-03245-z}
}

@article{audenaert2007,
  title={A sharp continuity estimate for the von Neumann entropy},
  author={Audenaert, Koenraad M. R.},
  journal={Journal of Physics A: Mathematical and Theoretical},
  volume={40},
  pages={8127},
  year={2007},
  publisher={IOP Publishing},
  doi={10.1088/1751-8113/40/28/S18},
  url={https://doi.org/10.1088/1751-8113/40/28/S18}
}

@article{Mirsky1960,
  author  = {Mirsky, L.},
  title   = {Symmetric Gauge Functions and Unitarily Invariant Norms},
  journal = {The Quarterly Journal of Mathematics},
  year    = {1960},
  volume  = {11},
  pages   = {50--59},
  doi     = {10.1093/qmath/11.1.50},
  url     = {https://doi.org/10.1093/qmath/11.1.50}
}

@article{Fannes1973,
  author  = {Fannes, Mark},
  title   = {A continuity property of the entropy density for spin lattice systems},
  journal = {Communications in Mathematical Physics},
  year    = {1973},
  volume  = {31},
  number  = {4},
  pages   = {291--294},
  doi     = {10.1007/BF01646490},
  url     = {https://doi.org/10.1007/BF01646490}
}

@article{poulin2010lieb,
  title = {Lieb-Robinson Bound and Locality for General Markovian Quantum Dynamics},
  author = {Poulin, David},
  journal = {Physical Review Letters},
  volume = {104},
  number = {19},
  pages = {190401},
  year = {2010},
  publisher = {American Physical Society},
  doi = {10.1103/PhysRevLett.104.190401},
  url = {https://link.aps.org/doi/10.1103/PhysRevLett.104.190401}
}

@article{barthel2012quasilocality,
  title = {Quasilocality and Efficient Simulation of Markovian Quantum Dynamics},
  author = {Barthel, Thomas and Kliesch, Martin},
  journal = {Physical Review Letters},
  volume = {108},
  number = {23},
  pages = {230504},
  year = {2012},
  publisher = {American Physical Society},
  doi = {10.1103/PhysRevLett.108.230504},
  url = {https://link.aps.org/doi/10.1103/PhysRevLett.108.230504}
}

@incollection{nachtergaele2010lieb,
  title = {Lieb-Robinson Bounds in Quantum Many-Body Physics},
  author = {Nachtergaele, Bruno and Sims, Robert},
  booktitle = {Entropy and the Quantum},
  series = {Contemporary Mathematics},
  volume = {529},
  pages = {141--176},
  year = {2010},
  publisher = {American Mathematical Society},
  doi = {10.1090/conm/529/10429},
  url = {https://doi.org/10.1090/conm/529}
}

@article{king2014hypercontractivity,
  title={Hypercontractivity for semigroups of unital qubit channels},
  author={King, Christopher},
  journal={Communications in Mathematical Physics},
  volume={328},
  number={1},
  pages={285--301},
  year={2014},
  publisher={Springer},
  doi={10.1007/s00220-014-1982-4},
  url={https://doi.org/10.1007/s00220-014-1982-4}
}

@article{Ciracetal2020,
  author  = {Cirac, Ignacio and Perez-Garcia, David and Schuch, Norbert and Verstraete, Frank},
  title   = {Matrix Product States and Projected Entangled Pair States: Concepts, Symmetries, and Theorems},
  journal = {Rev. Mod. Phys. 93, 045003 (2021)},
  year    = {2021},
  doi     = {10.1103/RevModPhys.93.045003},
  url     = {https://doi.org/10.1103/RevModPhys.93.045003}
}

@article{EntanglementAreaLaws2010,
  title = {Colloquium: Area laws for the entanglement entropy},
  author = {Eisert, J. and Cramer, M. and Plenio, M. B.},
  journal = {Rev. Mod. Phys.},
  volume = {82},
  issue = {1},
  pages = {277--306},
  numpages = {0},
  year = {2010},
  month = {Feb},
  publisher = {American Physical Society},
  doi = {10.1103/RevModPhys.82.277},
  url = {https://link.aps.org/doi/10.1103/RevModPhys.82.277}
}

@article{cheng2021simulating,
  title={Simulating noisy quantum circuits with matrix product density operators},
  author={Cheng, Song and Cao, Chenfeng and Zhang, Chao and Liu, Yongxiang and Hou, Shi-Yao and Xu, Pengxiang and Zeng, Bei},
  journal={Physical Review Research},
  volume={3},
  number={2},
  pages={023005},
  year={2021},
  publisher={APS},
  doi={10.1103/physrevresearch.3.023005},
  url={https://doi.org/10.1103/physrevresearch.3.023005}
}

@article{dankert2009exact,
  title = {Exact and approximate unitary 2-designs and their application to fidelity estimation},
  author = {Dankert, Christoph and Cleve, Richard and Emerson, Joseph and Livine, Etera},
  journal = {Physical Review A},
  volume = {80},
  number = {1},
  pages = {012304},
  year = {2009},
  doi = {10.1103/PhysRevA.80.012304},
  url = {https://doi.org/10.1103/PhysRevA.80.012304}
}

@article{gao2018efficient,
  title={Efficient classical simulation of noisy quantum computation},
  author={Gao, Xun and Duan, Luming},
  journal={arXiv preprint arXiv:1810.03176},
  year={2018},
  doi={10.48550/arXiv.1810.03176},
  url={https://arxiv.org/abs/1810.03176}
}

@article{schuster2025polynomial,
  title={A polynomial-time classical algorithm for noisy quantum circuits},
  author={Schuster, Thomas and Yin, Chao and Gao, Xun and Yao, Norman Y},
  journal={Physical Review X},
  volume={15},
  number={4},
  pages={041018},
  year={2025},
  publisher={APS},
  doi={10.1103/xct1-7kf2},
  url={https://doi.org/10.1103/xct1-7kf2}
}

@article{angrisani2025classically,
  title={Classically estimating observables of noiseless quantum circuits},
  author={Angrisani, Armando and Schmidhuber, Alexander and Rudolph, Manuel S and Cerezo, Marco and Holmes, Zo{\"e} and Huang, Hsin-Yuan},
  journal={Physical review letters},
  volume={135},
  number={17},
  pages={170602},
  year={2025},
  publisher={APS},
  doi={10.1103/lh6x-7rc3},
  url={https://doi.org/10.1103/lh6x-7rc3}
}

@article{martinez2025efficient,
  title={Efficient simulation of parametrized quantum circuits under nonunital noise through pauli backpropagation},
  author={Martinez, Victor and Angrisani, Armando and Pankovets, Ekaterina and Fawzi, Omar and Stilck Fran{\c{c}}a, Daniel},
  journal={Physical Review Letters},
  volume={134},
  number={25},
  pages={250602},
  year={2025},
  publisher={APS},
  doi={10.1103/j1gg-s6zb},
  url={https://doi.org/10.1103/j1gg-s6zb}
}

@article{nahum2017quantum,
  title = {Quantum Entanglement Growth under Random Unitary Dynamics},
  author = {Nahum, Adam and Vijay, Sagar and Haah, Jeongwan},
  journal = {Phys. Rev. X},
  volume = {7},
  issue = {3},
  pages = {031016},
  numpages = {14},
  year = {2017},
  month = {Aug},
  publisher = {American Physical Society},
  doi = {10.1103/PhysRevX.7.031016},
  url = {https://link.aps.org/doi/10.1103/PhysRevX.7.031016}
}

@article{vonkeyserlingk2018operator,
  title = {Operator Hydrodynamics, OTOCs, and Entanglement Growth in Systems without Conservation Laws},
  author = {von Keyserlingk, C. W. and Rakovszky, T. and Pollmann, F. and Sondhi, S. L.},
  journal = {Phys. Rev. X},
  volume = {8},
  issue = {2},
  pages = {021013},
  numpages = {22},
  year = {2018},
  month = {Apr},
  publisher = {American Physical Society},
  doi = {10.1103/PhysRevX.8.021013},
  url = {https://link.aps.org/doi/10.1103/PhysRevX.8.021013}
}

@article{wei2026noise,
  title={Noise-induced contraction of MPO truncation errors in noisy random circuits and Lindbladian dynamics},
  author={Wei, Zhi-Yuan and Rajakumar, Joel and Nelson, Jon and Malz, Daniel and Gullans, Michael J and Gorshkov, Alexey V},
  journal={arXiv preprint arXiv:2603.20400},
  year={2026},
  doi={10.48550/arXiv.2603.20400},
  url={https://arxiv.org/abs/2603.20400}
}

@article{berezutskii2025tensor,
  title={Tensor networks for quantum computing},
  author={Berezutskii, Aleksandr and Liu, Minzhao and Acharya, Atithi and Ellerbrock, Roman and Gray, Johnnie and Haghshenas, Reza and He, Zichang and Khan, Abid and Kuzmin, Viacheslav and Lyakh, Dmitry and others},
  journal={Nature Reviews Physics},
  volume={7},
  number={10},
  pages={581--593},
  year={2025},
  publisher={Nature Publishing Group UK London},
  doi={10.1038/s42254-025-00853-1},
  url={https://doi.org/10.1038/s42254-025-00853-1}
}

@article{li2023entanglement,
  title={Entanglement dynamics of noisy random circuits},
  author={Li, Zhi and Sang, Shengqi and Hsieh, Timothy H},
  journal={Physical Review B},
  volume={107},
  number={1},
  pages={014307},
  year={2023},
  publisher={APS},
  doi={10.1103/PhysRevB.107.014307},
  url={https://doi.org/10.1103/PhysRevB.107.014307}
}

@article{yan2025limitations,
  title={Limitations of noisy quantum devices in computing and entangling power},
  author={Yan, Yuxuan and Du, Zhenyu and Chen, Junjie and Ma, Xiongfeng},
  journal={npj Quantum Information},
  volume={11},
  number={1},
  pages={188},
  year={2025},
  publisher={Nature Publishing Group UK London},
  doi={10.1038/s41534-025-01136-4},
  url={https://doi.org/10.1038/s41534-025-01136-4}
}

@article{zhang2022entanglement,
  title={Entanglement entropy scaling of noisy random quantum circuits in two dimensions},
  author={Zhang, Meng and Wang, Chao and Dong, Shaojun and Zhang, Hao and Han, Yongjian and He, Lixin},
  journal={Physical Review A},
  volume={106},
  number={5},
  pages={052430},
  year={2022},
  publisher={APS},
  doi={10.1103/PhysRevA.106.052430},
  url={https://doi.org/10.1103/PhysRevA.106.052430}
}

@article{lee2025classical,
  title={Classical simulation of noisy random circuits from exponential decay of correlation},
  author={Lee, Su-un and Ghosh, Soumik and Oh, Changhun and Noh, Kyungjoo and Fefferman, Bill and Jiang, Liang},
  journal={arXiv preprint arXiv:2510.06328},
  year={2025},
  doi={10.48550/arXiv.2510.06328},
  url={https://arxiv.org/abs/2510.06328}
}

@article{denzler2026simulation,
  title={Simulation of noisy quantum circuits using frame representations},
  author={Denzler, Janek and Carrasco, Jose and Eisert, Jens and Guaita, Tommaso},
  journal={arXiv preprint arXiv:2601.05131},
  year={2026},
  doi={10.48550/arXiv.2601.05131},
  url={https://arxiv.org/abs/2601.05131}
}

@phdthesis{tomamichel2012framework,
  author = {Tomamichel, Marco},
  title = {A Framework for Non-Asymptotic Quantum Information Theory},
  school = {ETH Zurich},
  year = {2012},
  eprint = {1203.2142},
  archivePrefix = {arXiv},
  primaryClass = {quant-ph},
  url = {https://arxiv.org/abs/1203.2142}
}

\clearpage
\onecolumngrid
\appendix

\section{Operator entanglement entropy}
Consider a bipartition of an $n$-qubit system into subsystems $A$ and $B$, with $n_A$ and $n_B$ qubits respectively ($n=n_A+n_B$).
The density matrix $\rho$ can be expressed as:
\begin{equation}
    \rho = \sum_{i=1}^{R} \lambda_i \tau_i^{[A]} \otimes \tau_i^{[B]},
\end{equation}
where $\left\{\tau_i^{[A]} \right\}$ and $\left\{\tau_i^{[B]} \right\}$ are orthonormal operator bases for subsystems $A$ and $B$, with $\tr{\tau_i^{[A]} \tau_j^{[A]}} = \tr{\tau_i^{[B]} \tau_j^{[B]}} = \delta_{ij}$. 
We call $\lambda_i$ the Schmidt coefficients, and $R$ the Schmidt rank. 

The bipartite entropy of this decomposition is given by the \textit{operator space entanglement entropy} or simply \textit{operator entanglement entropy} (OEE) defined as~\cite{wellnitz2022rise}:
\begin{equation}\label{eq:oe_entropy}
    S_{OE}(\rho) = -\sum_{i=1}^{R} \lambda_i^2 \log_2 \lambda_i^2.
\end{equation}

Furthermore, we can define the \textit{normalized operator entanglement entropy} as follows: 
\begin{equation}
    \widetilde{S}_{OE}(\rho) = -\sum_{i=1}^{R} \frac{\lambda^2_i}{t} \log_2 \frac{\lambda^2_i}{t},
\end{equation}
where $t = \tr{\rho^2} = \sum_{i=1}^{R} \lambda_i^2$ is the purity of $\rho$. 
The relation between the normalized operator entanglement entropy $\widetilde{S}_{OE}(\rho)$ and unnormalized operator entanglement entropy $S_{OE}(\rho)$ is given by:
\begin{lemma}\label{lemma:normalized_oe_entropy}
    For any bipartite state $\rho$ with purity $t = \tr{\rho^2}$, the normalized operator entanglement entropy $\widetilde{S}_{OE}(\rho)$ can be expressed in terms of the unnormalized operator entanglement entropy $S_{OE}(\rho)$ as follows:
    \begin{equation}
        \widetilde{S}_{OE}(\rho) = \frac{S_{OE}(\rho)}{t} + \log_2 t.
    \end{equation}
\end{lemma}
\begin{proof}
    We have:
    \begin{equation}
        \begin{aligned}
                \widetilde{S}_{OE}(\rho) & = -\sum_{i=1}^{R} \frac{\lambda^2_i}{t} \log_2 \frac{\lambda^2_i}{t} \\
                & = -\sum_{i=1}^{R} \frac{\lambda^2_i}{t} \log_2 \lambda^2_i + \sum_{i=1}^{R} \frac{\lambda^2_i}{t} \log_2 t \\
                & = -\frac{1}{t} \sum_{i=1}^{R} \lambda^2_i \log_2 \lambda^2_i + \log_2 t \\
                & = \frac{S_{OE}(\rho)}{t} + \log_2 t.
        \end{aligned}
    \end{equation}
\end{proof}

This normalized operator entanglement entropy appeared in \cite{noh2020efficient} under the name \emph{MPO entanglement entropy}, and is used as a metric for the computational power of MPO representations of $\rho$.

\section{Bounding the bond dimension via operator entanglement entropy}\label{sec:complexity_tensor_network}

Operator entanglement becomes a complexity diagnostic only after the computational task is specified.
Consider a density operator $\rho$ admitting an operator Schmidt decomposition across a bipartition,
\begin{equation}
    \rho
    =
    \sum_{\alpha=1}^{R}
    \lambda_\alpha L_{\alpha}^{[A]}\otimes R_{\alpha}^{[B]},
\end{equation}
where $\{L_{\alpha}^{[A]}\}$ and $\{R_{\alpha}^{[B]}\}$ are Hilbert-Schmidt orthonormal operator families.
Let $\{\lambda_\alpha\}_{\alpha=1}^R$ be sorted in descending order, such that $\lambda_1 \geq \lambda_2 \geq \dots \geq \lambda_R > 0$. 
We denote the square of the unnormalized Schmidt weights as $p_\alpha = \lambda_\alpha^2$.
Following the notation established in the preliminaries, we denote purity as $t = \tr\{\rho^2\} = \sum_{\alpha=1}^R p_\alpha \leq 1$, unnormalized operator entanglement entropy as $S_{OE}(\rho) = -\sum_{\alpha=1}^R p_\alpha \log_2 p_\alpha$, normalized operator entanglement entropy as $\widetilde{S}_{OE}(\rho) = -\sum_{\alpha=1}^R \frac{p_\alpha}{t} \log_2 \frac{p_\alpha}{t}$. 
The approximation $\widehat\rho$ used below is the ordered Schmidt truncation. For a target bond dimension $\chi$, define
\begin{equation}
    \widehat\rho_\chi
    =
    \sum_{\alpha=1}^{\chi}
    \lambda_\alpha L_{\alpha}^{[A]}\otimes R_{\alpha}^{[B]},
\end{equation}
with $\widehat\rho_0=0$ when the zero-rank approximation is allowed. Since the operators $L_{\alpha}^{[A]}\otimes R_{\alpha}^{[B]}$ are Hilbert-Schmidt orthonormal, the squared truncation error is exactly the discarded tail,
\begin{equation}
    e_\chi
    =
    \norm{\rho-\widehat\rho_\chi}_{2}^{2}
    =
    \sum_{\alpha=\chi+1}^R p_\alpha .
\end{equation}
The absolute criterion asks for an additive tolerance,
\begin{equation}
    e_\chi=\norm{\rho-\widehat\rho_\chi}_{2}^{2}\leq \varepsilon .
\end{equation}
This is the natural specification when the downstream task only requires additive accuracy in the unnormalized state, or in a finite set of linear functionals with controlled Hilbert-Schmidt norm. For any such observable $O$,
\begin{equation}
    \abs{\tr[(\rho-\widehat\rho_\chi)O]}
    \leq
    \norm{\rho-\widehat\rho_\chi}_{2}\norm{O}_{2}
    \leq
    \sqrt{\varepsilon}\norm{O}_{2}.
\end{equation}
Absolute precision therefore resolves only the part of the noisy state that lies above an externally imposed additive error floor. In particular, if $t\leq\varepsilon$, discarding the whole operator already satisfies the criterion; the nontrivial absolute-error regime is $0<\varepsilon<t$.

The relative criterion instead asks for a fixed fraction of the remaining Hilbert-Schmidt weight,
\begin{equation}
    \frac{e_\chi}{t}
    =
    \frac{\norm{\rho-\widehat\rho_\chi}_{2}^{2}}{\norm{\rho}_{2}^{2}}
    \leq \delta ,
    \qquad
    \delta=\sum_{\alpha=\chi+1}^R\frac{p_\alpha}{t}.
\end{equation}
This is appropriate when the object of interest is the normalized vectorized state $\ket{\rho}/\norm{\rho}_{2}$, the normalized operator Schmidt spectrum, or any quantity obtained only after a later normalization step. It is also the conservative requirement when states of different purity must be compared at the same fractional accuracy.

Thus the two criteria probe different structures in a noisy mixed state. Fixed relative accuracy remains sensitive to the internal Schmidt distribution of the surviving operator weight, and is governed by $\widetilde S_{OE}$. Fixed absolute accuracy also sees the decay of the total Hilbert-Schmidt weight, and is governed by $S_{OE}$. The following bounds make this separation quantitative.

\begin{theorem}\label{thm:D_abs_error}
    Let $\rho$ be a bipartite state with purity $t=\tr\{\rho^2\}\leq 1$, and let $0<\varepsilon<t$.
    Let $R_{\mathrm{abs}}(\varepsilon)$ be the smallest integer $\chi$ such that truncating the ordered operator Schmidt decomposition of $\rho$ to its largest $\chi$ coefficients has absolute squared Hilbert-Schmidt error at most $\varepsilon$.
    Then
    \begin{equation}
        R_{\mathrm{abs}}(\varepsilon)
        \leq
        \max\left\{1,\left\lceil (t-\varepsilon)2^{S_{OE}(\rho)/\varepsilon}\right\rceil\right\}
        = \order{(t-\varepsilon)2^{S_{OE}(\rho)/\varepsilon}}.
    \end{equation}
    If $\varepsilon\geq t$, the zero-rank approximation already has squared Hilbert-Schmidt error at most $\varepsilon$.
\end{theorem}
\begin{proof}
    Write $p_\alpha=\lambda_\alpha^2$ and assume $p_1\geq p_2\geq\cdots$.
    For an integer $k\geq1$, let
    \begin{equation}
        e_k=\sum_{\alpha>k}p_\alpha
    \end{equation}
    be the truncation error after keeping the largest $k$ weights.
    We prove the contrapositive.
    Suppose $e_k>\varepsilon$.
    Then the retained mass satisfies
    \begin{equation}
        \sum_{\alpha=1}^{k}p_\alpha=t-e_k<t-\varepsilon .
    \end{equation}
    Since the sequence is non-increasing, the first discarded weight is bounded by the average retained weight,
    \begin{equation}
        p_{k+1}
        \leq \frac{1}{k}\sum_{\alpha=1}^{k}p_\alpha
        < \frac{t-\varepsilon}{k}.
    \end{equation}
    Therefore $p_\alpha\leq p_{k+1}< (t-\varepsilon)/k$ for every $\alpha>k$.
    Because $0<t-\varepsilon<1$ and $k\geq1$, the logarithm below is positive, and the tail entropy obeys
    \begin{equation}
        \begin{aligned}
        S_{OE}(\rho)
        &\geq -\sum_{\alpha>k}p_\alpha\log_2 p_\alpha  \\
        &> e_k\log_2\left(\frac{k}{t-\varepsilon}\right)
        > \varepsilon\log_2\left(\frac{k}{t-\varepsilon}\right).
        \end{aligned}
    \end{equation}
    Hence $k<(t-\varepsilon)2^{S_{OE}(\rho)/\varepsilon}$ whenever the tail exceeds $\varepsilon$.
    Consequently, every integer
    \begin{equation}
        k\geq \max\left\{1,\left\lceil (t-\varepsilon)2^{S_{OE}(\rho)/\varepsilon}\right\rceil\right\}
    \end{equation}
    has $e_k\leq\varepsilon$, which proves the claim.
\end{proof}

\begin{theorem}\label{thm:D_rel_error}
    Let $\rho$ be a bipartite state with nonzero purity $t=\tr\{\rho^2\}$, and let $\delta\in(0,1)$.
    Let $R_{\mathrm{rel}}(\delta)$ be the smallest integer $\chi$ such that truncating the ordered operator Schmidt decomposition of $\rho$ to its largest $\chi$ coefficients has relative squared Hilbert-Schmidt error at most $\delta$.
    Then
    \begin{equation}
        R_{\mathrm{rel}}(\delta)
        \leq
        \max\left\{1,\left\lceil (1-\delta)2^{\widetilde{S}_{OE}(\rho)/\delta}\right\rceil\right\} = \order{(1-\delta)2^{\widetilde{S}_{OE}(\rho)/\delta}}.
    \end{equation}
\end{theorem}
\begin{proof}
    Let $\tilde p_\alpha=p_\alpha/t$, sorted in non-increasing order, and define
    \begin{equation}
        \tilde e_k=\sum_{\alpha>k}\tilde p_\alpha .
    \end{equation}
    Suppose $\tilde e_k>\delta$.
    Then the retained normalized mass is less than $1-\delta$, and monotonicity gives
    \begin{equation}
        \tilde p_{k+1}
        \leq
        \frac{1}{k}\sum_{\alpha=1}^{k}\tilde p_\alpha
        <
        \frac{1-\delta}{k}.
    \end{equation}
    Hence every discarded normalized weight is smaller than $(1-\delta)/k$, and
    \begin{equation}
        \begin{aligned}
        \widetilde S_{OE}(\rho)
        &\geq -\sum_{\alpha>k}\tilde p_\alpha\log_2\tilde p_\alpha \\
        &> \tilde e_k\log_2\left(\frac{k}{1-\delta}\right)
        > \delta\log_2\left(\frac{k}{1-\delta}\right).
        \end{aligned}
    \end{equation}
    Therefore $\tilde e_k>\delta$ implies
    \begin{equation}
        k<(1-\delta)2^{\widetilde S_{OE}(\rho)/\delta}.
    \end{equation}
    Thus every integer
    \begin{equation}
        k\geq \max\left\{1,\left\lceil (1-\delta)2^{\widetilde S_{OE}(\rho)/\delta}\right\rceil\right\}
    \end{equation}
    has relative tail at most $\delta$.
\end{proof}

\section{Entropy Bounds}\label{appendix:entropy}

\subsection{Purity-controlled bound}

Before proceeding, we consider a useful 3-block state family and derive its purity and Schmidt coefficients as follows:
\begin{lemma}\label{lemma:3_block_state}
    Let $D_A$ and $D_B$ be the dimensions of subsystems $A$ and $B$ respectively, without loss of generality, we assume that $D_A=m \leq n = D_B$.
    Consider a 3-block state family defined as:
    \begin{equation}
        \hat{\rho}(p_0, p_1, p_2) = p_0 \ketbra{\Phi_m}{\Phi_m} + \frac{p_1}{m^2 - 1} \left( \mathbb{I}_{A\otimes B_0} - \ketbra{\Phi_m}{\Phi_m} \right) + \frac{p_2}{m (n - m)} \left( \mathbb{I}_{A\otimes B_0^{\perp}}\right),
    \end{equation}
    where $p_0, p_1, p_2 \geq 0$, $p_0 + p_1 + p_2 = 1$, $B_0$ is a subspace of $B$ with dimension $m$, $B_0^{\perp}$ is the orthogonal complement of $B_0$, $\ket{\Phi_m} = \frac{1}{\sqrt{m}} \sum_{i=1}^{m} \ket{ii}$ is the maximally entangled state on the subspace $A \otimes B_0$.

    The purity of this state is given by:
    \begin{equation}
        \tr{\hat{\rho}(p_0, p_1, p_2)^2} = p_0^2 + \frac{p_1^2}{m^2 - 1} + \frac{p_2^2}{m (n - m)}.
    \end{equation}
    The Schmidt coefficients of this state are:
    \begin{equation}
        \lambda_1 = \sqrt{\frac{(1-p_2)^2}{m^2} + \frac{p_2^2}{m(n - m)}}, \quad \lambda_i = \abs{\frac{p_0-\frac{p_1}{m^2 - 1}}{m}}, \quad \text{for } i = 2, \ldots, m^2.
    \end{equation}
\end{lemma}

\begin{proof}
    The purity can be calculated as:
    \begin{equation}
        \begin{aligned}
            \tr{\hat{\rho}(p_0, p_1, p_2)^2} & = p_0^2 \tr{\ketbra{\Phi_m}{\Phi_m}^2} + \left(\frac{p_1}{m^2 - 1}\right)^2 \tr{\left( \mathbb{I}_{A\otimes B_0} - \ketbra{\Phi_m}{\Phi_m} \right)^2} + \left(\frac{p_2}{m (n - m)}\right)^2 \tr{\left( \mathbb{I}_{A\otimes B_0^{\perp}}\right)^2}\\
            & = p_0^2 + \frac{p_1^2}{m^2 - 1} + \frac{p_2^2}{m (n - m)},
        \end{aligned}
    \end{equation}
    since $\ketbra{\Phi_m}{\Phi_m}$, $\mathbb{I}_{A\otimes B_0} - \ketbra{\Phi_m}{\Phi_m}$ and $\mathbb{I}_{A\otimes B_0^{\perp}}$ are orthogonal to each other.

    To calculate the Schmidt coefficients, we define a set of orthonormal operator bases for subsystem $A$ as $\{F_i\}_{i=1}^{m^2}$, where:
    \begin{equation}
        F_1 = \frac{1}{\sqrt{m}} \mathbb{I}_A.
    \end{equation}
    And a set of orthonormal operator bases for subsystem $B_0$ as $\{G_j\}_{j=1}^{m^2}$, where: 
    \begin{equation}
        G_1 = \frac{1}{\sqrt{m}} \mathbb{I}_{B_0}.
    \end{equation}
    A set of orthonormal operator bases for subsystem $B_0^{\perp}$ can be defined as $\{H_k\}_{k=1}^{m(n - m)}$, where:
    \begin{equation}
        H_1 = \frac{1}{\sqrt{n - m}} \mathbb{I}_{B_0^{\perp}}.
    \end{equation}

    The sets $\{F_i\}_{i=1}^{m^2}$, $\{G_j\}_{j=1}^{m^2}$ and $\{H_k\}_{k=1}^{m(n - m)}$  satisfy the orthonormality conditions:
    \begin{equation}
        \tr{F_i F_j^\dagger} = \delta_{ij}, \quad \tr{G_i G_j^\dagger} = \delta_{ij}, \quad \tr{H_i H_j^\dagger} = \delta_{ij}.
    \end{equation}

    Then we can express the state $\ketbra{\Phi_m}{\Phi_m}$ as:
    \begin{equation}
        \ketbra{\Phi_m}{\Phi_m} = \sum_{i=1}^{m^2} \frac{1}{m} F_i \otimes G_i.
    \end{equation}
    The state $\mathbb{I}_{A\otimes B_0}$ can be expressed as:
    \begin{equation}
        \mathbb{I}_{A\otimes B_0} = m F_1 \otimes G_1.
    \end{equation}
    The state $\mathbb{I}_{A\otimes B_0^{\perp}}$ can be expressed as:
    \begin{equation}
        \mathbb{I}_{A\otimes B_0^{\perp}} = \sqrt{m(n - m)} F_1 \otimes H_1.
    \end{equation}

    Thus, the state $\hat{\rho}(p_0, p_1, p_2)$ can be expressed as:
    \begin{equation}\label{ap:eq:3_block_schmidt_decomposition}
        \begin{aligned}
            \hat{\rho}(p_0, p_1, p_2) & = p_0 \sum_{i=1}^{m^2} \frac{1}{m} F_i \otimes G_i + \frac{p_1}{m^2 - 1} \left( m F_1 \otimes G_1 - \sum_{i=1}^{m^2} \frac{1}{m} F_i \otimes G_i \right) + \frac{p_2}{m (n - m)} \left( \sqrt{m(n - m)} F_1 \otimes H_1 \right)\\
            & = \sum_{i=1}^{m^2} \left( \frac{p_0}{m} - \frac{p_1}{m (m^2 - 1)} \right) F_i \otimes G_i + \left( \frac{p_1 m}{m^2 - 1} \right) F_1 \otimes G_1 + \frac{p_2}{\sqrt{m (n - m)}} F_1 \otimes H_1\\
            & = \left( \frac{p_0}{m} - \frac{p_1}{m (m^2 - 1)} + \frac{p_1 m}{m^2 - 1}\right) F_1 \otimes G_1 + \sum_{i=2}^{m^2} \left( \frac{p_0}{m} - \frac{p_1}{m (m^2 - 1)} \right) F_i \otimes G_i + \frac{p_2}{\sqrt{m (n - m)}} F_1 \otimes H_1\\
            & =  F_1 \otimes \left[\left( \frac{p_0}{m} - \frac{p_1}{m (m^2 - 1)} + \frac{p_1 m}{m^2 - 1}\right) G_1 + \frac{p_2}{\sqrt{m (n - m)}} H_1 \right] + \sum_{i=2}^{m^2} \left( \frac{p_0}{m} - \frac{p_1}{m (m^2 - 1)} \right) F_i \otimes G_i.
        \end{aligned}
    \end{equation}
    It is easy to see that the set $\left\{ \left( \frac{p_0}{m} - \frac{p_1}{m (m^2 - 1)} + \frac{p_1 m}{m^2 - 1}\right) G_1 + \frac{p_2}{\sqrt{m (n - m)}} H_1, G_2, \ldots, G_{m^2} \right\}$ are orthogonal to each other, thus \cref{ap:eq:3_block_schmidt_decomposition} is the Schmidt decomposition of the state $\hat{\rho}(p_0, p_1, p_2)$.

    On the other hand, we have:
    \begin{equation}
        \begin{aligned}
            & \tr{\left[\left( \frac{p_0}{m} - \frac{p_1}{m (m^2 - 1)} + \frac{p_1 m}{m^2 - 1}\right) G_1 + \frac{p_2}{\sqrt{m (n - m)}} H_1\right]^2} \\
            = & \left( \frac{p_0}{m} - \frac{p_1}{m (m^2 - 1)} + \frac{p_1 m}{m^2 - 1}\right)^2 \tr{G_1^2} + \left(\frac{p_2}{\sqrt{m (n - m)}}\right)^2 \tr{H_1^2} \\
            = & \left( \frac{p_0}{m} - \frac{p_1}{m (m^2 - 1)} + \frac{p_1 m}{m^2 - 1}\right)^2 + \frac{p_2^2}{m (n - m)} \\
            = & \frac{(1-p_2)^2}{m^2} + \frac{p_2^2}{m(n - m)},
        \end{aligned}
    \end{equation}
    where we have used $p_0 + p_1 + p_2 = 1$. 
    
    Thus the Schmidt coefficients of this state are:
    \begin{equation}
        \lambda_1 = \sqrt{\frac{(1-p_2)^2}{m^2} + \frac{p_2^2}{m(n - m)}}, \quad \lambda_i = \abs{\frac{p_0-\frac{p_1}{m^2 - 1}}{m}}, \quad \text{for } i = 2, \ldots, m^2.
    \end{equation}
\end{proof}

For a given purity $t = \tr{\rho^2}$, we can characterize the maximum operator entanglement entropy as follows:
\begin{theorem}\label{theorem:maximum_entropy}
    For any bipartite state $\rho$ with purity $t = \tr{\rho^2}$ on subsystems $A$ and $B$ of dimension $D_A$ and $D_B$ respectively.
    The operator entanglement entropy is upper bounded by $S_{OE}(\rho) \leq S_{OE}^{\max}(t)$, where
    \begin{equation}
        S_{OE}^{\max}(t) 
        \begin{cases}
        = t\log_2 D_A D_B + \left(t - \frac{1}{D_A D_B}\right) \log_2 \left(\frac{D_{\min}^2-1}{tD_A D_B - 1}\right), & \text{if } t \leq \frac{1}{D_{\min}D_{\max}} + \frac{D_{\min}^2 - 1}{D_{\max}^2}, \\
        \leq t\log_2 D_A D_B + \left(t - \frac{1}{D_A D_B}\right) \log_2 \left(\frac{D_{\min}^2-1}{tD_A D_B - 1}\right), & \text{if } \frac{1}{D_{\min}D_{\max}} + \frac{D_{\min}^2 - 1}{D_{\max}^2} \leq t \leq \frac{D_{\min}}{D_{\max}}, \\
        \leq t\log_2 \frac{D_{\min}^2}{t}, & \text{if } \frac{D_{\min}}{D_{\max}} \leq t \leq 1,
        \end{cases}
    \end{equation}
    where $D_{\min} = \min(D_A, D_B)$ and $D_{\max} = \max(D_A, D_B)$.
    Note that the whole system has dimension $D_A D_B$, and the purity $t \in [1/(D_A D_B), 1]$.

    Similarly, the normalized operator entanglement entropy $\widetilde{S}_{OE}(\rho)$ is upper bounded by $\widetilde{S}_{OE}^{\max}(t) = \frac{1}{t} S_{OE}^{\max}(t) + \log_2 t$, which can be expressed as:
    \begin{equation}
        \widetilde{S}_{OE}^{\max}(t) 
        \begin{cases}
        = \log_2 tD_A D_B + \left(1 - \frac{1}{tD_A D_B}\right) \log_2 \left(\frac{D_{\min}^2-1}{tD_A D_B - 1}\right), & \text{if } t \leq \frac{1}{D_{\min}D_{\max}} + \frac{D_{\min}^2 - 1}{D_{\max}^2}, \\
        \leq \log_2 tD_A D_B + \left(1 - \frac{1}{tD_A D_B}\right) \log_2 \left(\frac{D_{\min}^2-1}{tD_A D_B - 1}\right), & \text{if } \frac{1}{D_{\min}D_{\max}} + \frac{D_{\min}^2 - 1}{D_{\max}^2} \leq t \leq \frac{D_{\min}}{D_{\max}}, \\
        \leq \log_2 D_{\min}^2, & \text{if } \frac{D_{\min}}{D_{\max}} \leq t \leq 1.
        \end{cases}
    \end{equation}
\end{theorem}

\begin{proof}
First we note that the purity $t = \tr{\rho^2}$ can be expressed in terms of the Schmidt coefficients $\lambda_i$ as:
\begin{equation}
    t =  \tr{\rho^2} = \sum_{i,j=1}^{\chi} \lambda_i \lambda_j \tr{\tau_i^{[A]} \tau_j^{[A]}} \tr{\tau_i^{[B]} \tau_j^{[B]}} = \sum_{i=1}^{\chi} \lambda_i^2. 
\end{equation}
To maximize the operator entanglement entropy $S_{OE}(\rho) = -\sum_{i=1}^{\chi} \lambda_i^2 \log_2 \lambda_i^2$ under the constraints $\sum_{i=1}^{\chi} \lambda_i^2 = t$, we can use the method of Lagrange multipliers.
Without loss of generality, we assume that $\lambda_1$ is the largest Schmidt coefficient $\lambda_1 \geq \lambda_i$ for all $i$, and takes the value of $\lambda_1^*$, and thus we only need to optimize over the remaining variables $\{\lambda_i\}_{i=2}^{\chi}$ with the constraint $\sum_{i=2}^{\chi} \lambda_i^2 = t - (\lambda_1^*)^2$.
The Lagrangian function is given by:
\begin{equation}
    \mathcal{L} = -\sum_{i=2}^{\chi} \lambda_i^2 \log_2 \lambda_i^2 + \alpha \left(\sum_{i=2}^{\chi} \lambda_i^2 - t + (\lambda_1^*)^2 \right),
\end{equation}
where $\alpha$ is the Lagrange multiplier.
Taking the partial derivatives of $\mathcal{L}$ with respect to $\lambda_i$ and setting them to zero gives:
\begin{equation}
    \frac{\partial \mathcal{L}}{\partial \lambda_i} = -2 \lambda_i \log_2 \lambda_i^2 - \frac{2}{\ln 2} \lambda_i + 2 \alpha \lambda_i = 0.
\end{equation}
Solving this equation, we find that the optimal $\lambda_i$ take the form:
\begin{equation}
    \lambda_i^2 = 2^{\alpha - \frac{1}{\ln 2}} \  \text{or} \  0, \quad \text{for } i = 2, \ldots, \chi.
\end{equation}
Assume that there are $n_0$ non-zero $\lambda_i$ among $\{\lambda_i\}_{i=2}^{\chi}$, then we have:
\begin{equation}
    n_0 \lambda_i^2 = t - (\lambda_1^*)^2 \implies \lambda_i^2 = \frac{t - (\lambda_1^*)^2}{n_0}, \quad \text{for } i = 2, \ldots, n_0 + 1.
\end{equation}
Substituting these optimal values back into the expression for $S_{OE}(\rho)$, we obtain:
\begin{equation}\label{ap:eq:oe_entropy_n0_lambda1}
    S_{OE}(\rho) = -(\lambda_1^*)^2 \log_2 (\lambda_1^*)^2 - (t - (\lambda_1^*)^2) \log_2 \frac{t - (\lambda_1^*)^2}{n_0},
\end{equation}
which achieves its maximum when $n_0 = D_{\min}^2 - 1$.

Considering the function $g(x) = -x \log_2 x - (t - x) \log_2 \frac{t - x}{D_{\min}^2 - 1}$ for $x \in [0, t]$, we find that it is concave in this interval, by taking the second derivative:
\begin{equation}
    g''(x) = -\frac{1}{x \ln 2} - \frac{1}{(t - x) \ln 2} = \frac{- t}{x (t - x) \ln 2} \leq 0.
\end{equation}
Thus, the maximum of $g(x)$ occurs at the critical point where $g'(x) = 0$.
Calculating the first derivative:
\begin{equation}
    g'(x) = -\log_2 x - \frac{1}{\ln 2} + \log_2 (t - x) + \frac{1}{\ln 2} - \log_2 (D_{\min}^2 - 1) = \log_2 \frac{t - x}{x(D_{\min}^2 - 1)}.
\end{equation}
Setting $g'(x) = 0$ gives:
\begin{equation}
    \frac{t - x}{x(D_{\min}^2 - 1)} = 1 \implies x = \frac{t}{D_{\min}^2}.
\end{equation}
On the other hand, the maximum Schmidt coefficient $\lambda_1^*$ can be expressed as:
\begin{equation}
    \lambda_{\max}(\rho) = \max_{\norm{A}_2 = 1, \norm{B}_2 = 1} \tr{\rho A \otimes B}.
\end{equation}
When we choose $A = \frac{1}{\sqrt{D_A}} \mathbb{I}_{D_A}$ and $B = \frac{1}{\sqrt{D_B}} \mathbb{I}_{D_B}$, we have:
\begin{equation}
    \lambda_{\max}(\rho) \geq \tr{\rho \frac{1}{\sqrt{D_A}} \mathbb{I}_{D_A} \otimes \frac{1}{\sqrt{D_B}} \mathbb{I}_{D_B}} = \frac{1}{\sqrt{D_A D_B}} \tr{\rho} = \frac{1}{\sqrt{D_A D_B}}.
\end{equation}
Thus, we have $\lambda_1^* \geq \frac{1}{\sqrt{D_A D_B}}$.

When $t \leq \frac{D_{\min}}{D_{\max}}$, we have $\frac{t}{D_{\min}^2} \leq \frac{1}{D_A D_B} \leq \lambda_1^{*2}$, and the maximum of $S_{OE}(\rho)$ is upper
bounded when $\lambda_i^2$ take the values:
\begin{equation}\label{ap:eq:lambda_values_case1}
    \lambda_1^2 = \frac{1}{D_A D_B}, \quad \lambda_i^2 = \frac{t - \frac{1}{D_A D_B}}{D_{\min}^2 - 1}, \quad \text{for } i = 2, \ldots, D_{\min}^2.
\end{equation}
Later we will show that these values can be achieved using 3-block state family defined in \cref{lemma:3_block_state}.

Substituting these values back into the expression for $S_{OE}(\rho)$, we obtain:
\begin{equation}
    S_{OE}^{\max}(t) \leq t\log_2 D_A D_B + \left(t - \frac{1}{D_A D_B}\right) \log_2 \left(\frac{D_{\min}^2-1}{tD_A D_B - 1}\right).
\end{equation}
    
When $\frac{D_{\min}}{D_{\max}} \leq t \leq 1$, we have $\frac{t}{D_{\min}^2} \geq \frac{1}{D_A D_B}$, and the maximum of $S_{OE}(\rho)$ can be upper bounded when $\lambda_i^2$ take the values:
\begin{equation}\label{ap:eq:lambda_values_case2}
    \lambda_i^2 = \frac{t}{D_{\min}^2}, \quad \text{for } i = 1, \ldots, D_{\min}^2.
\end{equation}
Substituting these values back into the expression for $S_{OE}(\rho)$, we obtain:
\begin{equation}
    S_{OE}^{\max}(t) \leq t\log_2 \frac{D_{\min}^2}{t}.
\end{equation}

The rest of the proof is to show that the Schmidt coefficients given in \cref{ap:eq:lambda_values_case1} can be achieved by a physical state when $t \leq \frac{1}{D_{\min}D_{\max}} + \frac{D_{\min}^2 - 1}{D_{\max}^2}$.
Without loss of generality, we assume that $D_A = m \leq n = D_B$, and consider the 3-block state family defined in \cref{lemma:3_block_state} with parameters $p_0, p_1, p_2$.
From \cref{lemma:3_block_state}, we know that the purity of this state is given by:
\begin{equation}
    \tr{\hat{\rho}(p_0, p_1, p_2)^2} = p_0^2 + \frac{p_1^2}{m^2 - 1} + \frac{p_2^2}{m (n - m)}.
\end{equation}
The Schmidt coefficients of this state are:
\begin{equation}
    \lambda_1 = \sqrt{\frac{(1-p_2)^2}{m^2} + \frac{p_2^2}{m(n - m)}}, \quad \lambda_i = \abs{\frac{p_0-\frac{p_1}{m^2 - 1}}{m}}, \quad \text{for } i = 2, \ldots, m^2,
\end{equation}
where $p_0, p_1, p_2 \geq 0$, $p_0 + p_1 + p_2 = 1$.

To achieve the Schmidt coefficients in \cref{ap:eq:lambda_values_case1}, we set:
\begin{equation}
    p_0 = \frac{1}{mn} + \frac{m^2 - 1}{m} \sqrt{\frac{t-\frac{1}{mn}}{m^2 - 1}}, \quad p_1 = \frac{m^2 - 1}{m} \left( \frac{1}{n} - \sqrt{\frac{t-\frac{1}{mn}}{m^2 - 1}} \right), \quad p_2 = \frac{n-m}{n},
\end{equation}
which is a genuine probability distribution when $\frac{1}{mn} \leq t \leq \frac{1}{mn} + \frac{m^2 - 1}{n^2}$.

Furthermore, for normalized operator entanglement entropy, we have \cref{lemma:normalized_oe_entropy} which states that $\widetilde{S}_{OE}(\rho) = \frac{S_{OE}(\rho)}{t} + \log_2 t$, thus we have $\widetilde{S}_{OE}(\rho) \leq \frac{1}{t} S_{OE}^{\max}(t) + \log_2 t$.

\end{proof}

For the case where $D_A = D_B$, we can still derive the maximum operator entanglement entropy as follows:
\begin{corollary}\label{cor:maximum_entropy_equal_dimension}
    For any bipartite state $\rho$ with purity $t = \tr{\rho^2}$ on subsystems $A$ and $B$ with equal dimension $D_*=D_A=D_B$.
    The maximum operator entanglement entropy is given by:
    \begin{equation}
        S_{OE}^{\max}(t) = t\log_2 D_*^2 + \left(t - \frac{1}{D_*^2}\right) \log_2 \left(\frac{D_*^2-1}{tD_*^2 - 1}\right),
    \end{equation}
    and its normalized version:
    \begin{equation}
        \widetilde{S}_{OE}^{\max}(t) = \log_2 tD_*^2 + \left(1 - \frac{1}{tD_*^2}\right) \log_2 \left(\frac{D_*^2-1}{tD_*^2 - 1}\right).
    \end{equation}
\end{corollary}

\subsection{Property of the purity-controlled bound}
Here we present some useful properties of the maximum operator entanglement entropy $S_{OE}^{\max}(t)$ and $\widetilde{S}_{OE}^{\max}(t)$.
Some people might be concerned that the purity upper bound for normalized operator entanglement entropy $\widetilde{S}_{OE}(\rho)$ is trivial, as it is purity dependent while in the definition, the purity scale has been normalized out. 

To address this concern, we take $D_A = D_B = 4$ as an example, and plot the maximum operator entanglement entropy $S_{OE}^{\max}(t)$ and its normalized version $\widetilde{S}_{OE}^{\max}(t)$ with respect to the purity $t$ in \cref{ap:fig:maximum_entropy}.
This figure shows that both $S_{OE}^{\max}(t)$ and $\widetilde{S}_{OE}^{\max}(t)$ are monotonically increasing and concave functions with respect to the purity $t$. When the purity $t$ is close to the lower bound $\frac{1}{D_A D_B}$, $\widetilde{S}_{OE}^{\max}(t)$ is close to zero, which means that $\widetilde{S}_{OE}^{\max}(t)$ is a non-trivial function and can be used to characterize the maximum operator entanglement entropy for states with different purity.
\begin{figure}[hbp]
    \centering
    \includegraphics[width=0.5\textwidth]{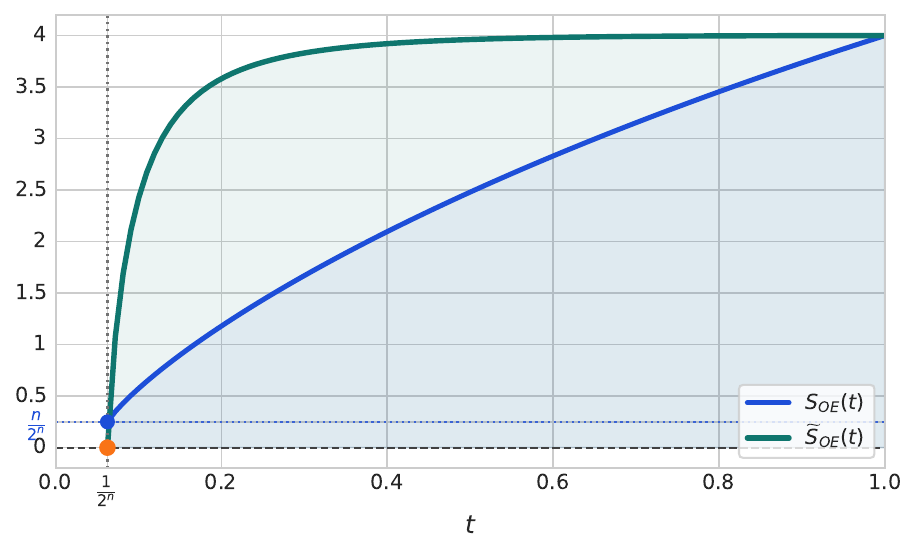}
    \caption{The maximum operator entanglement entropy $S_{OE}^{\max}(t)$ and its normalized version $\widetilde{S}_{OE}^{\max}(t)$ with respect to the purity $t$ when $D_A = D_B = 4$.}
    \label{ap:fig:maximum_entropy}
\end{figure}

In the following, we will give a rigorous proof about these properties of $S_{OE}^{\max}(t)$ and $\widetilde{S}_{OE}^{\max}(t)$.
For simplicity, we define $S_{OE_1}^{\max}(t) =  t\log_2 D_A D_B + \left(t - \frac{1}{D_A D_B}\right) \log_2 \left(\frac{D_{\min}^2-1}{tD_A D_B - 1}\right), \text{if } \frac{1}{D_A D_B} \leq t \leq \frac{D_{\min}}{D_{\max}}$ and $S_{OE_2}^{\max}(t) = t\log_2 \frac{D_{\min}^2}{t}, \text{if } \frac{D_{\min}}{D_{\max}} \leq t \leq 1$.
\begin{theorem}\label{theorem:property_maximum_entropy_part1}
    The maximum operator entanglement entropy $S_{OE_1}^{\max}(t)$ and its normalized version $\widetilde{S}_{OE_1}^{\max}(t) = \frac{1}{t} S_{OE_1}^{\max}(t) + \log_2 t$ are monotonically increasing and concave functions with respect to the purity $t \in \left[ \frac{1}{D_A D_B}, \frac{D_{\min}}{D_{\max}} \right]$. Their limits at the lower bound of purity are given by $\lim_{t \to \frac{1}{D_A D_B}} S_{OE_1}^{\max}(t) = \frac{\log_2 D_AD_B}{D_A D_B}$ and $\lim_{t \to \frac{1}{D_A D_B}} \widetilde{S}_{OE_1}^{\max}(t) = 0$.
\end{theorem}
\begin{proof}
Taking the first derivative of $S_{OE_1}^{\max}(t)$ with respect to $t$, we have:
\begin{equation}
    \derivative{S_{OE_1}^{\max}(t)}{t}= \log_2 D_A D_B + \log_2 \left(\frac{D_{\min}^2 - 1}{t D_A D_B - 1}\right) - \frac{1}{\ln 2} = \frac{\ln \left(\frac{D_{\min}^2 - 1}{t - \frac{1}{D_A D_B}}\right) - 1}{\ln 2} = \log_2 \left(\frac{D_{\min}^2 - 1}{\left(t - \frac{1}{D_A D_B}\right)e}\right).
\end{equation}
For the normalized version $\widetilde{S}_{OE_1}^{\max}(t) = \frac{1}{t} S_{OE_1}^{\max}(t) + \log_2 t = \log_2\left(tD_AD_B\right) + \left(1 - \frac{1}{tD_A D_B}\right) \log_2 \left(\frac{D_{\min}^2-1}{tD_A D_B - 1}\right)$, we have:
\begin{equation}
    \begin{aligned}
        \derivative{\widetilde{S}_{OE_1}^{\max}(t)}{t} & = \frac{1}{t \ln 2} + \frac{1}{t^2 D_A D_B} \log_2 \left(\frac{D_{\min}^2-1}{tD_A D_B - 1}\right) - \left(1 - \frac{1}{tD_A D_B}\right) \frac{D_A D_B}{\left(tD_A D_B - 1\right) \ln 2} \\
         & = \frac{1}{t \ln 2} + \frac{1}{t^2 D_A D_B} \log_2 \left(\frac{D_{\min}^2-1}{tD_A D_B - 1}\right) - \frac{1}{t \ln 2} \\
         & = \frac{1}{t^2 D_A D_B} \log_2 \left(\frac{D_{\min}^2-1}{tD_A D_B - 1}\right).
    \end{aligned}
\end{equation}

Since $t \in [\frac{1}{D_A D_B}, 1]$, we have $\frac{D_{\min}^2 - 1}{\left(t - \frac{1}{D_A D_B}\right)e} \geq \frac{D_{\min}^2 - 1}{e}$.
When $D_{\min} \geq 2$, we have $\frac{D_{\min}^2 - 1}{e} > 1$, thus $\derivative{S_{OE_1}^{\max}(t)}{t} > 0$ and $\derivative{\widetilde{S}_{OE_1}^{\max}(t)}{t} > 0$ for $t \in [\frac{1}{D_A D_B}, 1]$, which means that $S_{OE_1}^{\max}(t)$ and $\widetilde{S}_{OE_1}^{\max}(t)$ are monotonically increasing functions with respect to $t$.

Taking the second derivative of $S_{OE_1}^{\max}(t)$ and $\widetilde{S}_{OE_1}^{\max}(t)$ with respect to $t$, we have:
\begin{equation}
    \frac{d^2 S_{OE_1}^{\max}(t)}{d t^2} = -\frac{D_A D_B}{(t D_A D_B - 1) \ln 2} = -\frac{1}{\left(t - \frac{1}{D_A D_B}\right) \ln 2} < 0,
\end{equation}
and
\begin{equation}
    \frac{d^2 \widetilde{S}_{OE_1}^{\max}(t)}{d t^2} = -\frac{2}{t^3 D_A D_B} \log_2 \left(\frac{D_{\min}^2-1}{tD_A D_B - 1}\right) - \frac{1}{t^2 \left(tD_A D_B - 1\right) \ln 2} < 0,
\end{equation}
where the first term is negative since $\frac{D_{\min}^2-1}{tD_A D_B - 1} \geq 1$, with $t\leq \frac{D_{\min}}{D_{\max}}$.

Thus, $S_{OE_1}^{\max}(t)$ and $\widetilde{S}_{OE_1}^{\max}(t)$ are concave functions with respect to $t$.

Finally, we calculate the limits of $S_{OE_1}^{\max}(t)$ and $\widetilde{S}_{OE_1}^{\max}(t)$ at the lower bound of purity.
Let $x = t- \frac{1}{D_A D_B}$, we have:
\begin{equation}
    \begin{aligned}
        \lim_{t \to \frac{1}{D_A D_B}} S_{OE_1}^{\max}(t)
        &= \lim_{x \to 0} \left[\left(x + \frac{1}{D_A D_B}\right) \log_2 D_A D_B
           + x \log_2 \left(\frac{D_{\min}^2-1}{x D_A D_B}\right) \right] \\
        &= \frac{\log_2 D_A D_B}{D_A D_B} - \lim_{x \to 0} x \log_2 x
         = \frac{\log_2 D_A D_B}{D_A D_B},
    \end{aligned}
\end{equation}
and
\begin{equation}
    \lim_{t \to \frac{1}{D_A D_B}} \widetilde{S}_{OE_1}^{\max}(t) = \lim_{x \to 0} \frac{1}{x + \frac{1}{D_A D_B}} S_{OE_1}^{\max}\left(x + \frac{1}{D_A D_B}\right) + \log_2 \left(x + \frac{1}{D_A D_B}\right) = 0.
\end{equation}

\end{proof}

\begin{theorem}\label{theorem:property_maximum_entropy_part2}
    The maximum operator entanglement entropy $S_{OE_2}^{\max}(t)$ and its normalized version $\widetilde{S}_{OE_2}^{\max}(t) = \frac{1}{t} S_{OE_2}^{\max}(t) + \log_2 t$ are monotonically increasing and concave functions with respect to the purity $t \in \left[ \frac{D_{\min}}{D_{\max}}, 1 \right]$.
\end{theorem}
\begin{proof}
Taking the first derivative of $S_{OE_2}^{\max}(t)$ with respect to $t$, we have:
\begin{equation}
    \derivative{S_{OE_2}^{\max}(t)}{t} = \log_2 \frac{D_{\min}^2}{t} - \frac{1}{\ln 2} = \log_2 \frac{D_{\min}^2}{t e}.
\end{equation}
Since $t \in [\frac{D_{\min}}{D_{\max}}, 1]$, we have $\frac{D_{\min}^2}{t e} \geq \frac{D_{\min}^2 }{ e}$.
When $D_{\min} \geq 2$, we have $\frac{D_{\min}^2 }{ e} > 1$, thus $\derivative{S_{OE_2}^{\max}(t)}{t} > 0$ for $t \in [\frac{D_{\min}}{D_{\max}}, 1]$, which means that $S_{OE_2}^{\max}(t)$ is a monotonically increasing function with respect to $t$.

Taking the second derivative of $S_{OE_2}^{\max}(t)$ with respect to $t$, we have:
\begin{equation}
    \frac{d^2 S_{OE_2}^{\max}(t)}{d t^2} = -\frac{1}{t \ln 2} < 0.
\end{equation}
Thus, $S_{OE_2}^{\max}(t)$ is a concave function with respect to $t$.

For the normalized version $\widetilde{S}_{OE_2}^{\max}(t) = \frac{1}{t} S_{OE_2}^{\max}(t) + \log_2 t = \log_2 D_{\min}^2$, it is a constant function with respect to $t$, thus it is also monotonically increasing and concave with respect to $t$.
\end{proof}

\subsection{Operator entanglement entropy with tail control}

Eckart-Young theorem~\cite{eckart1936approximation} states that the best rank-$r$ approximation of a matrix $M$ in terms of the Frobenius norm is given by the sum of the first $r$ singular values and corresponding singular vectors of $M$.
A similar result is also true for operator Schmidt decomposition, as shown by the following lemma.
\begin{lemma}[Eckart-Young theorem]\label{lemma:operator_schmidt_decomposition_eckart_young}
    For a bipartite operator $\rho$ with operator Schmidt decomposition $\rho = \sum_{i=1}^{\chi} \lambda_i \tau_i^{[A]} \otimes \tau_i^{[B]}$, with Schmidt coefficients $\{\lambda_i\}_{i=1}^{\chi}$ sorted in non-increasing order, the best rank-$r$ approximation of $\rho$ in terms of the Frobenius norm is given by:
    \begin{equation}
        \rho_r = \sum_{i=1}^{r} \lambda_i \tau_i^{[A]} \otimes \tau_i^{[B]},
    \end{equation}
and the deviation between $\rho$ and $\rho_r$ is the minimum among all rank-$r$ approximations of $\rho$, which can be expressed as:
\begin{equation}
    \min_{\text{Schmidt rank}(\sigma) \leq r} \norm{\rho - \sigma}_2 = \norm{\rho - \rho_r}_2 = \sqrt{\sum_{i=r+1}^{\chi} \lambda_i^2}.
\end{equation}
\end{lemma}

\begin{proof}
    Taking normalized orthogonal bases $\{A_i\}_{i=1}^{D_A}$ and 
    $\{B_i\}_{i=1}^{D_B}$ for the subsystems $A$ and $B$, we can express $\rho$ and $\sigma$ in terms of these bases as:
\begin{equation}
    \begin{aligned}
            \rho = \sum_{i=1}^{D_A} \sum_{j=1}^{D_B} M_{ij} A_i \otimes B_j, \\
            \sigma = \sum_{i=1}^{D_A} \sum_{j=1}^{D_B} N_{ij} A_i \otimes B_j.
    \end{aligned}
\end{equation}
Denoting the coefficient matrices as $M$ and $N$, we have:
\begin{equation}
    \norm{\rho - \sigma}_2^2 = \tr{(\rho - \sigma)^2} = \sum_{i=1}^{D_A} \sum_{j=1}^{D_B} (M_{ij} - N_{ij})^2 = \norm{M - N}_F^2,
\end{equation}
where $\norm{\cdot}_F$ is the Frobenius norm.

On the other hand, the operator Schmidt rank of $\rho$ is equal to the rank of the coefficient matrix $M$, and the Schmidt coefficients $\lambda_i$ are equal to the singular values of $M$.
Thus, the problem of finding the best rank-$r$ approximation of $\rho$ in terms of the Frobenius norm is equivalent to finding the best rank-$r$ approximation of the matrix $M$ in terms of the Frobenius norm:
\begin{equation}
    \min_{\text{Schmidt rank}(\sigma) \leq r} \norm{\rho - \sigma}_2 = \min_{\text{rank}(N) \leq r} \norm{M - N}_F,
\end{equation}
where the right side is the original Eckart-Young theorem, which states that the best rank-$r$ approximation of a matrix $M$ in terms of the Frobenius norm is given by the sum of the first $r$ singular values and corresponding singular vectors of $M$.
Thus, we have:
\begin{equation}
    \min_{\text{Schmidt rank}(\sigma) \leq r} \norm{\rho - \sigma}_2 = \norm{M - M_r}_F = \sqrt{\sum_{i=r+1}^{\chi} \lambda_i^2},
\end{equation}
where $M_r$ is the best rank-$r$ approximation of $M$ given by the sum of the first $r$ singular values and corresponding singular vectors of $M$, which corresponds to the operator $\rho_r$ given by the sum of the first $r$ Schmidt coefficients and corresponding Schmidt operators of $\rho$.

\end{proof}

\begin{lemma}\label{lemma:entropy_tail_control}
    For any bipartite state $\rho$ with Schmidt coefficients $\{\lambda_i\}_{i=1}^{\chi}$ sorted in non-increasing order, and any $r \in \{1, \ldots, \chi\}$, suppose the purity is $t = \tr{\rho^2}=\sum_{i=1}^{\chi} \lambda_i^2$ and the tail is controlled by $p = \sum_{i=r+1}^{\chi} \lambda_i^2$, then we have:
    \begin{equation}
        S_{OE}(\rho) \leq (t-p) \log_2 r - (t-p) \log_2 (t-p) + p \log_2 (D_{\min}^2) - p \log_2 p,
    \end{equation}
    and its normalized version:
    \begin{equation}
        \widetilde{S}_{OE}(\rho) \leq \left( 1 - \delta \right) \log_2 r - \left(1-\delta\right) \log_2 (1-\delta) - \delta\log_2 \delta + \delta \log_2 (D_{\min}^2),
    \end{equation}
    where $\delta = \frac{p}{t}$ is the normalized tail.
\end{lemma}
\begin{proof}
    To maximize the operator entanglement entropy $S_{OE}(\rho) = -\sum_{i=1}^{\chi} \lambda_i^2 \log_2 \lambda_i^2$ under the constraints $\sum_{i=1}^{\chi} \lambda_i^2 = t$ and $\sum_{i=r+1}^{\chi} \lambda_i^2 = p$, we can use the method of Lagrange multipliers.
    Without loss of generality, we assume that $\lambda_1 \geq \lambda_i$ for all $i$, and thus we only need to optimize over the variables $\{\lambda_i\}_{i=1}^{r}$ with the constraints $\sum_{i=1}^{r} \lambda_i^2 = t - p$ and $\lambda_i^2 \geq 0$ for $i = 1, \ldots, r$.
    The Lagrangian function is given by:
    \begin{equation}
        \mathcal{L} = -\sum_{i=1}^{\chi} \lambda_i^2 \log_2 \lambda_i^2 + \alpha \left(\sum_{i=1}^{r} \lambda_i^2 - t + p \right) + \beta \left(\sum_{i=r+1}^{\chi} \lambda_i^2 - p \right),
    \end{equation}
    where $\alpha$ and $\beta$ are the Lagrange multipliers.
    Taking the partial derivatives of $\mathcal{L}$ with respect to $\lambda_i$ and setting them to zero gives:
    \begin{equation}
        \frac{\partial \mathcal{L}}{\partial \lambda_i} = -2 \lambda_i \log_2 \lambda_i^2 - \frac{2}{\ln 2} \lambda_i + 2 \alpha \lambda_i = 0, \quad \text{for } i = 1, \ldots, r,
    \end{equation}
    and
    \begin{equation}
        \frac{\partial \mathcal{L}}{\partial \lambda_i} = -2 \lambda_i \log_2 \lambda_i^2 - \frac{2}{\ln 2} \lambda_i + 2 \beta \lambda_i = 0, \quad \text{for } i = r+1, \ldots, \chi.
    \end{equation}
    Solving these equations, we find that the optimal $\lambda_i$ take the form:
    \begin{equation}
        \begin{aligned}
                & \lambda_i^2 = 2^{\alpha - \frac{1}{\ln 2}} \  \text{for } i = 1, \ldots, r, \\
                & \lambda_i^2 = 2^{\beta - \frac{1}{\ln 2}} \  \text{for } i = r+1, \ldots, \chi.
        \end{aligned}
    \end{equation}
    Thus, we have:
    \begin{equation}
        r \cdot 2^{\alpha - \frac{1}{\ln 2}} = t - p, \quad (\chi - r) \cdot 2^{\beta - \frac{1}{\ln 2}} = p.
    \end{equation}
    Substituting these optimal values back into the expression for $S_{OE}(\rho)$, we obtain:
    \begin{equation}
        S_{OE}(\rho) = -(t-p) \log_2 \frac{t-p}{r} - p \log_2 \frac{p}{\chi - r} \leq (t-p) \log_2 r - (t-p) \log_2 (t-p) + p \log_2 (D_{\min}^2) - p \log_2 p,
    \end{equation}
    where we have used the fact that $\chi \leq D_{\min}^2$.

    Using the relation between $S_{OE}(\rho)$ and $\widetilde{S}_{OE}(\rho)$ given in \cref{lemma:normalized_oe_entropy}, we can also obtain the upper bound on $\widetilde{S}_{OE}(\rho)$:
    \begin{equation}
        \begin{aligned}
            \widetilde{S}_{OE}(\rho) = \frac{1}{t} S_{OE}(\rho) + \log_2 t &\leq \frac{t-p}{t} \log_2 r - \frac{t-p}{t} \log_2 (t-p) + \frac{p}{t} \log_2 (D_{\min}^2) - \frac{p}{t} \log_2 p + \log_2 t\\
            & = \left( 1 - \frac{p}{t} \right) \log_2 r - \frac{t-p}{t} \log_2 (\frac{t-p}{t}) - \frac{p}{t} \log_2 \left(\frac{p}{t}\right) + \frac{p}{t} \log_2 (D_{\min}^2) \\
            & = \left( 1 - \delta \right) \log_2 r - \left(1-\delta\right) \log_2 (1-\delta) - \delta\log_2 \delta + \delta \log_2 (D_{\min}^2),
        \end{aligned}
    \end{equation}
    where $\delta = \frac{p}{t}$ is the fraction of the tail in the total purity.
\end{proof}

\begin{lemma}\label{lemma:monotonicity_of_F}
    Let $F(t,p,D_{\min}) = (t-p) \log_2 r - (t-p) \log_2 (t-p) + p \log_2 (D_{\min}^2) - p \log_2 p$ be the upper bound on the operator entanglement entropy given in \cref{lemma:entropy_tail_control}.
    If $r \geq 3$, then $F$ is a monotonically increasing function with respect to $D_{\min}$, $t \in [0,1]$ and $p \in [0,\frac{t}{2}]$.
\end{lemma}
\begin{proof}
    It is straightforward to verify that $\frac{\partial F}{\partial D_{\min}} = \frac{2p}{D_{\min} \ln 2} > 0$, thus $F$ is a monotonically increasing function with respect to $D_{\min}$.
    For $r \geq 3$, we also have $\frac{\partial F}{\partial t} = \log_2 r - \log_2 (t-p) - \frac{1}{\ln 2} > 0$.
    thus $F$ is a monotonically increasing function with respect to $t$.
    Finally, we have $\frac{\partial F}{\partial p} = -\log_2 r + \log_2 (t-p) + \log_2 (D_{\min}^2) - \log_2 p = \log_2 \frac{(t-p)D_{\min}^2}{rp}$. 
    Since $r \leq D_{\min}^2$ and $p \leq t-p$, we have $\frac{(t-p)D_{\min}^2}{rp} \geq 1$, thus $\frac{\partial F}{\partial p} \geq 0$, which means that $F$ is a monotonically increasing function with respect to $p$.
\end{proof}

\begin{corollary}\label{cor:l2_distance_rank_r_approximation}
    For any $n$-qubit bipartite state $\rho$ and $\sigma_r$, if the Schmidt rank of $\sigma_r$ is $r \geq 3$, and the squared $l_2$ distance between $\rho$ and $\sigma_r$ is less than $p\in [0,\frac{1}{2}]$, i.e., $\norm{\rho - \sigma_r}_2^2 \leq p$, then we have the following upper bound on the operator entanglement entropy of $\rho$:
    \begin{equation}
        S_{OE}(\rho) \leq (1-p) \log_2 r - (1-p) \log_2 (1-p) + np - p \log_2 p.
    \end{equation}
\end{corollary}
\begin{proof}
    By \cref{lemma:operator_schmidt_decomposition_eckart_young}, assuming the Schmidt coefficients of $\rho$ are sorted in non-increasing order as $\lambda_1 \geq \lambda_2 \geq \ldots \geq \lambda_{\chi}$, we have $\sum_{i=r+1}^{\chi} \lambda_i^2 =\min_{\text{Schmidt rank}(\sigma) \leq r} \norm{\rho - \sigma}_2^2 \leq \norm{\rho - \sigma_r}_2^2 \leq p$.
    Thus, by \cref{lemma:entropy_tail_control}, we have $S_{OE}(\rho) \leq F(t,\sum_{i=r+1}^{\chi} \lambda_i^2,D_{\min})$.

    By \cref{lemma:monotonicity_of_F} and $r \geq 3$, we have $F(t,\sum_{i=r+1}^{\chi} \lambda_i^2,D_{\min}) \leq F(1,\sum_{i=r+1}^{\chi} \lambda_i^2, \sqrt{2^n})$.
    If $p \leq \frac{1}{2}$, we have $F(1,\sum_{i=r+1}^{\chi} \lambda_i^2, \sqrt{2^n}) \leq F(1,p, \sqrt{2^n}) = (1-p) \log_2 r - (1-p) \log_2 (1-p) + np - p \log_2 p$, which completes the proof.
\end{proof}

\subsection{Operator entanglement and mutual information}

When $\rho$ is a pure state, we have 
\begin{equation}
    \widetilde{S}_{OE}(\rho) = 2S(A)_{\rho},
\end{equation}
where $S(A)_{\rho} = -\tr{\rho_A \log \rho_A}$ is the von Neumann entropy of the reduced density matrix $\rho_A = \tr_B \{\rho\}$. 
We now generalize this relation to mixed states, obtaining the following inequality:
\begin{theorem}\label{thm::eqn::OEandMutualInfo}
    For any bipartite state $\rho$ of system $AB$ with purity $t = \tr{\rho^2}$, the normalized operator entanglement entropy $\widetilde{S}_{OE}(\rho)$ satisfies
    \begin{equation}
        \widetilde{S}_{OE}(\rho) \geq I(A:B)_{\rho^2/t}.
    \end{equation}
\end{theorem}
Here, $I(A:B)_{\rho^2/t} = S(A)_{\rho^2/t} + S(B)_{\rho^2/t} - S(AB)_{\rho^2/t}$ is the mutual information of the bipartite state $\rho^2/t$. 
When $\rho$ is pure, we have $\rho^2/t = \rho$, and $I(A:B)_{\rho} = S(A)_{\rho} + S(B)_{\rho} = 2S(A)_{\rho}$, therefore recovering the equality discussed above. 
Note that by \cref{lemma:normalized_oe_entropy}, the inequality above can be equivalently expressed as: 
\begin{equation}
    S_{OE}(\rho) \geq t I(A:B)_{\rho^2/t} - t \log_2 t.
\end{equation}

To prove the theorem, we rely on the vectorization of $\rho$. 
Let $\ket{i}$ and $\ket{j}$ be the orthonormal bases of the Hilbert spaces $\mathcal{H}_A$ and $\mathcal{H}_B$. 
Let $\ket{\bar{i}}$ and $\ket{\bar{j}}$ be the corresponding orthonormal bases of the auxiliary Hilbert spaces $\mathcal{H}_{\overline{A}}$ and $\mathcal{H}_{\overline{B}}$. 
We define the vectorization of $\rho$ as the following vector in the doubled Hilbert space $\mathcal{H}_{AB\overline{A}\overline{B}}$:
\begin{equation}
   | \rho \rangle\!\rangle = \sum_{i,j,i',j'} [\rho ]_{i'j',ij}\ket{ij}\otimes \ket{\bar{i'} \bar{j'}}, 
\end{equation}
where $[\rho ]_{i'j',ij} = \bra{i' j'} \rho \ket{i j}$ are the matrix elements of $\rho$ in the basis $\ket{i j}$. 
The vectorization $| \rho \rangle\!\rangle$ satisfies
\begin{equation}
    \langle\!\langle \rho | \rho \rangle\!\rangle = \tr{\rho^2} = t.
\end{equation}
Thus, we can define the normalized vectorization of $\rho$ as the following pure state:
\begin{equation}
    |\tilde{\rho}\rangle\!\rangle = \frac{1}{\sqrt{t}} |\rho\rangle\!\rangle.
\end{equation}
Now we examine various reduced density matrices of the pure state $|\tilde{\rho}\rangle\!\rangle$.
First, we have $\tr_{\overline{AB}}\{ |\tilde{\rho}\rangle\!\rangle \langle\!\langle \tilde{\rho} | \} = (\rho^2/t)^{\intercal}$.
Therefore, we have 
\begin{equation}
    \tr_{A\overline{AB}}\{ |\tilde{\rho}\rangle\!\rangle \langle\!\langle \tilde{\rho} | \} = \left(\tr_{A}\{\rho^2/t\}\right)^{\intercal},\quad \tr_{B\overline{AB}}\{ |\tilde{\rho}\rangle\!\rangle \langle\!\langle \tilde{\rho} | \} = \left(\tr_{B}\{\rho^2/t\}\right)^{\intercal}
\end{equation} 
\begin{equation}
    \tr_{AB}\{ |\tilde{\rho}\rangle\!\rangle \langle\!\langle \tilde{\rho} | \} = \frac{1}{t}\sum_{i,j,i',j'}\sum_{i'',j''} [\rho]_{ij,i''j''}\overline{[\rho]_{i'j',i''j''}} \ket{\bar{i}\bar{j}}\bra{\bar{i'}\bar{j'}}. 
\end{equation}
Define the unitary operators $U_{A}: \mathcal{H}_{A} \to \mathcal{H}_{\overline{A}}$ and $U_{B}: \mathcal{H}_{B} \to \mathcal{H}_{\overline{B}}$ by $U_A\ket{i} = \ket{\bar{i}}$ and $U_B\ket{j} = \ket{\bar{j}}$. 
Then it is straightforward to verify that
\begin{equation}\label{eqn:: symmetry between AB and overline{AB}}
    \tr_{AB}\{ |\tilde{\rho}\rangle\!\rangle \langle\!\langle \tilde{\rho} | \} = (U_A\otimes U_B)\left(\tr_{\overline{AB}}\{ |\tilde{\rho}\rangle\!\rangle \langle\!\langle \tilde{\rho} | \}\right)^{\intercal}(U_A\otimes U_B)^\dagger.
\end{equation}
Consequently, we have 
\begin{equation}\label{eqn:: symmetry between A and overline{A}}
    \tr_{AB\overline{B}}\{ |\tilde{\rho}\rangle\!\rangle \langle\!\langle \tilde{\rho} | \} = U_A \left(\tr_{B\overline{AB}}\{ |\tilde{\rho}\rangle\!\rangle \langle\!\langle \tilde{\rho} | \}\right)^{\intercal} U_A^\dagger = U_A \tr_{B}\{\rho^2/t\} U_A^\dagger.
\end{equation}
\begin{proof}Proof of \cref{thm::eqn::OEandMutualInfo}
    Let us denote $|\tau^{[A]}_i\rangle\!\rangle$ in $\mathcal{H}_{A\overline{A}}$ and $|\tau^{[B]}_i\rangle\!\rangle$ in $\mathcal{H}_{B\overline{B}}$ as the vectorization of the Schmidt operators $\tau^{[A]}_i$ and $\tau^{[B]}_i$, respectively. 
    Following the operator Schmidt decomposition, we can expand $|\tilde{\rho}\rangle\!\rangle$ as
    \begin{equation}
        |\tilde{\rho}\rangle\!\rangle = \sum_{i=1}^{\chi} \frac{\lambda_i}{\sqrt{t}} |\tau^{[A]}_i\rangle\!\rangle \otimes |\tau^{[B]}_i\rangle\!\rangle.
    \end{equation}
    It follows immediately that the reduced density matrix of $|\tilde{\rho}\rangle\!\rangle$ on the subsystem $A\overline{A}$ is given by
    \begin{equation}
        \tr_{B\overline{B}}\{ |\tilde{\rho}\rangle\!\rangle \langle\!\langle \tilde{\rho} | \} = \sum_{i=1}^{\chi} \frac{\lambda_i^2}{t} |\tau^{[A]}_i\rangle\!\rangle \langle\!\langle \tau^{[A]}_i |.
    \end{equation}
    Therefore, we have 
    \begin{equation}
        S(A\overline{A})_{|\tilde{\rho}\rangle\!\rangle \langle\!\langle \tilde{\rho} |} = -\sum_{i=1}^{\chi} \frac{\lambda_i^2}{t} \log_2 \frac{\lambda_i^2}{t} = \widetilde{S}_{OE}(\rho).
    \end{equation}
    By \eqref{eqn:: symmetry between AB and overline{AB}}, we have $S(\overline{AB})_{|\tilde{\rho}\rangle\!\rangle \langle\!\langle \tilde{\rho} |} = S(\overline{AB})_{|\tilde{\rho}\rangle\!\rangle \langle\!\langle \tilde{\rho} |} = S(AB)_{\rho^2/t}$. 
    By \eqref{eqn:: symmetry between A and overline{A}}, we have $S(\overline{A})_{|\tilde{\rho}\rangle\!\rangle \langle\!\langle \tilde{\rho} |} = S(A)_{|\tilde{\rho}\rangle\!\rangle \langle\!\langle \tilde{\rho} |}$.
    Therefore, by strong subadditivity of von Neumann entropy, we have:
    \begin{equation}
        \begin{aligned}
            0&\leq I(A:\overline{B}\vert \overline{A})_{|\tilde{\rho}\rangle\!\rangle \langle\!\langle \tilde{\rho} |} = S(A\overline{A})_{|\tilde{\rho}\rangle\!\rangle \langle\!\langle \tilde{\rho} |} + S(\overline{AB})_{|\tilde{\rho}\rangle\!\rangle \langle\!\langle \tilde{\rho} |} - S(\overline{A})_{|\tilde{\rho}\rangle\!\rangle \langle\!\langle \tilde{\rho} |} - S(A\overline{A}\overline{B})_{|\tilde{\rho}\rangle\!\rangle \langle\!\langle \tilde{\rho} |}\\
            & = \widetilde{S}_{OE}(\rho) + S(AB)_{\rho^2/t} - S(A)_{\rho^2/t} - S(B)_{\rho^2/t}. 
        \end{aligned}
    \end{equation}
    where in the second equality we have the fact that $S(A\overline{A}\overline{B})_{|\tilde{\rho}\rangle\!\rangle \langle\!\langle \tilde{\rho} |} = S(B)_{\rho^2/t}$, which holds since $|\tilde{\rho}\rangle\!\rangle \langle\!\langle \tilde{\rho} |$ is a pure state. 
    Rearranging the above inequality, we obtain $\widetilde{S}_{OE}(\rho) \geq I(A:B)_{\rho^2/t}$, which completes the proof.  
\end{proof}

\begin{remark}
    From the above analysis it seems that an equally natural quantity is the normalized operator entanglement entropy of $\rho^{\frac{1}{2}}$.
    This is because the vectorization $|\rho^{\frac{1}{2}}\rangle\!\rangle$ is automatically normalized, and satisfies $\tr_{\overline{AB}} |\rho^{\frac{1}{2}}\rangle\!\rangle \langle\!\langle \rho^{\frac{1}{2}} | = \rho^{\intercal}$.
    Then following the same argument, we deduce that 
    \begin{equation}
        \widetilde{S}_{OE}(\rho^{\frac{1}{2}}) \geq I(A:B)_{\rho}. 
    \end{equation}
    We note that the quantity $\widetilde{S}_{OE}(\rho^{\frac{1}{2}})$ has been introduced in \cite{DuttaFaulkner2021ReflectedEntropy} under the name \emph{reflected entropy}. 
\end{remark}

By a similar strategy, we now derive an upper bound on the operator entanglement. 

\begin{lemma}
    For any bipartite state $\rho$ of system $AB$ with purity $t = \tr{\rho^2}$, we have the following upper bound on the operator entanglement entropy:
    \begin{equation}
        S_{OE}(\rho) \leq 2t S(A)_{\rho^2/t} - t \log_2 t.
    \end{equation}
\end{lemma}
\begin{proof}
    Let us first prove that for any $\rho_{AB}$, the normalized operator entanglement entropy satisfies 
    \begin{equation}
        \widetilde{S}_{OE}(\rho) \leq 2S(A)_{\rho^2/t}. 
    \end{equation}
    Consider again the canonical purification $|\tilde{\rho}\rangle\!\rangle$ of $\rho$ in $\mathcal{H}_{AB\overline{AB}}$. 
    Then by the subadditivity of von Neumann entropy, we have 
    \begin{equation}
       S(A)_{|\tilde{\rho}\rangle\!\rangle \langle\!\langle \tilde{\rho} |} + S(\overline{A})_{|\tilde{\rho}\rangle\!\rangle \langle\!\langle \tilde{\rho} |} - S(A\overline{A})_{|\tilde{\rho}\rangle\!\rangle \langle\!\langle \tilde{\rho} |} = I(A:\overline{A})_{|\tilde{\rho}\rangle\!\rangle \langle\!\langle \tilde{\rho} |} \geq 0.  
    \end{equation}
Since the reduced density matrix of $|\tilde{\rho}\rangle\!\rangle$ on the subsystem $\overline{A}$ is unitarily equivalent to that on $A$, we deduce that 
\begin{equation}
    2S(A)_{\rho^2/t} = 2 S(A)_{|\tilde{\rho}\rangle\!\rangle \langle\!\langle \tilde{\rho} |}\geq \widetilde{S}_{OE}(\rho). 
\end{equation}
By \cref{lemma:normalized_oe_entropy}, we then have 
\begin{equation}
    2t S(A)_{\rho^2/t} \geq S_{OE}(\rho) + t \log_2 t,
\end{equation}
which completes the proof. 
\end{proof}

\subsection{Continuity bound of Operator Entanglement}

In this section, we establish the following continuity bound on the difference in operator entanglement entropy between two bipartite quantum states in terms of their $l_2$ distance. 

\begin{lemma}\label{lemma:: Fannes-Audenaert for subnormalized distribution}
    Let $P = \{p_i\}^m_{i=1}$ and $Q = \{q_i\}^m_{i=1}$ be two nonnegative sequences, with $t_P = \sum_{i=1}^m p_i$ and $t_Q = \sum_{i=1}^m q_i$ both less than or equal to $1$.
    Define $\tilde{T} = \frac{1}{2}\left( \norm{P-Q}_1 + \abs{t_P-t_Q} \right)$. 
    Then we have the following bound on the difference in their Shannon entropy:
    \begin{equation}
        \abs{H(P) - H(Q)}\leq \tilde{T}\log_2(m) + h(\tilde{T}) + \abs{\eta(1-t_P) - \eta(1-t_Q)}.
    \end{equation}
    Here, $h(t) = -t \log_2 t - (1-t) \log_2 (1-t)$ is the binary entropy function, and $\eta(x) = -x \log_2 x$ is the entropy function. 
\end{lemma}
\begin{proof}
    Consider the following two probability distributions:
    \begin{equation}
        P' = p_1,\ldots, p_m, 1-t_P, \quad Q' = q_1, \ldots, q_m, 1-t_Q.
    \end{equation}
    Then we have $\frac{1}{2}\norm{P'-Q'}_1 = \frac{1}{2}\left( \norm{P-Q}_1 + \abs{t_P-t_Q} \right) = \tilde{T}$. 
    Thus, by Fannes-Audenaert inequality \cite{audenaert2007}, 
    \begin{align}
        \abs{H(P') - H(Q')}\leq \tilde{T}\log_2(m+1-1) + h(\tilde{T}) = \tilde{T}\log_2(m) + h(\tilde{T}).
    \end{align}
    Meanwhile, we have $H(P') = H(P) + \eta(1-t_P)$ and $H(Q') = H(Q) + \eta(1-t_Q)$. 
    Therefore, we have 
    \begin{equation}
        \begin{aligned}
            \abs{H(P) - H(Q)} & = \abs{H(P') - H(Q') + \eta(1-t_P) - \eta(1-t_Q)} \\
            & \leq \abs{H(P') - H(Q')} + \abs{\eta(1-t_P) - \eta(1-t_Q)} \\
            & \leq \tilde{T}\log_2(m) + h(\tilde{T}) + \abs{\eta(1-t_P) - \eta(1-t_Q)}. 
        \end{aligned}
    \end{equation}
    This completes the proof. 
\end{proof}

\begin{lemma}\label{lemma:: bound purity difference by l2 distance}
    Let $\rho$ and $\sigma$ be two quantum states with purities $t_{\rho}$ and $t_{\sigma}$, respectively. 
    Then we have the following bound on the difference in purities:
    \begin{equation}
        \abs{t_{\rho} - t_{\sigma}} \leq \left( \sqrt{t_{\rho}} + \sqrt{t_{\sigma}} \right) \norm{\rho - \sigma}_2. 
    \end{equation}
\end{lemma}
\begin{proof}
    By definition, we have 
    \begin{equation}
        \abs{t_{\rho} - t_{\sigma}} = \abs{\tr{\rho^2} - \tr{\sigma^2}} = \abs{\norm{\rho}^2_2 - \norm{\sigma}^2_2} = (\norm{\rho}_2 + \norm{\sigma}_2)\abs{\norm{\rho}_2 - \norm{\sigma}_2}. 
    \end{equation}
    By triangle inequality, we have $\abs{\norm{\rho}_2 - \norm{\sigma}_2} \leq \norm{\rho - \sigma}_2$. 
    So the result follows. 
\end{proof}

\begin{prop}\label{prop::continuity_bound of S_OE}
    Let $\rho$ and $\sigma$ be two bipartite quantum states of system $AB$. 
    Suppose the operator Schmidt ranks of both states are upper-bounded by $\chi_0$. 
    Suppose $\norm{\rho - \sigma}_2\leq \frac{1}{4}$, then the difference in their operator entanglement entropy can be upper bounded as follows:
    \begin{equation}\label{eqn:: l2_continuity_bound of S_OE}
        \abs{S_{OE}(\rho) - S_{OE}(\sigma)} \leq 2\norm{\rho - \sigma}_2 (\log_2(\chi_0)+2) + h(2\norm{\rho - \sigma}_2) + \eta(2\norm{\rho - \sigma}_2). 
    \end{equation}
    where $h(t) = -t \log_2 t - (1 - t) \log_2 (1 - t)$ is the binary entropy function, and $\eta(x) = -x \log_2 x$ is the entropy function. 
\end{prop}
\begin{proof}
    For a bipartite quantum state $\rho$ of system $AB$, we define $\mathcal{F}(\rho)$ to be the linear map from $A\overline{A}$ to $B\overline{B}$ obtained by reshaping the operator $\rho$, so that the operator Schmidt coefficients of $\rho$ are the singular values of $\mathcal{F}(\rho)$.
    Since reshaping only shuffles the matrix entries of an operator, it preserves the Hilbert-Schmidt norm: 
    \begin{equation}
        \norm{\mathcal{F}(\rho) - \mathcal{F}(\sigma)}_2 = \norm{\rho - \sigma}_2. 
    \end{equation}
    Denote by $\overrightarrow{\lambda} = \lambda_1\geq \lambda_2\geq \cdots\geq \lambda_{\chi_0}$ and $\overrightarrow{\mu} =\mu_1\geq \mu_2\geq \cdots\geq \mu_{\chi_0}$ the operator Schmidt coefficients of $\rho$ and $\sigma$, respectively. 
    Consider the two sub-normalized distributions $\overrightarrow{\lambda^2} = \{\lambda^2_i\}$ and $\overrightarrow{\mu^2} = \{\mu^2_i\}$. 
    We have $\sum_{i=1}^{\chi_0} \lambda_i^2 = t_{\rho}$ and $\sum_{i=1}^{\chi_0} \mu_i^2 = t_{\sigma}$. 
    At the same time, $H(\overrightarrow{\lambda^2}) = S_{OE}(\rho)$ and $H(\overrightarrow{\mu^2}) = S_{OE}(\sigma)$. 
    Thus, apply \cref{lemma:: Fannes-Audenaert for subnormalized distribution} and we obtain 
    \begin{equation}
        \abs{S_{OE}(\rho) - S_{OE}(\sigma)} = \abs{H(\overrightarrow{\lambda^2}) - H(\overrightarrow{\mu^2})}\leq \tilde{T}\log_2(\chi_0) + h(\tilde{T}) + \abs{\eta(1-t_{\rho}) - \eta(1-t_{\sigma})}, 
    \end{equation}
    where $\tilde{T} = \frac{1}{2}\left( \norm{\overrightarrow{\lambda^2}-\overrightarrow{\mu^2}}_1 + \abs{t_{\rho}-t_{\sigma}} \right)$. 
    In the following, we bound the first two terms then the last term in the above inequality. 
    By \cref{lemma:: bound purity difference by l2 distance}, we have 
    \begin{equation}
        \abs{t_{\rho}-t_{\sigma}} \leq 2 \norm{\rho - \sigma}_2
    \end{equation} 
    For the term $\norm{\overrightarrow{\lambda^2}-\overrightarrow{\mu^2}}_1 $, by Cauchy-Schwarz inequality, 
    \begin{align}
        \norm{\overrightarrow{\lambda^2}-\overrightarrow{\mu^2}}_1 & = \sum_{i=1}^{\chi_0} \abs{\lambda_i^2 - \mu_i^2} = \sum_{i=1}^{\chi_0} \abs{\lambda_i - \mu_i} \abs{\lambda_i + \mu_i} \leq \norm{\overrightarrow{\lambda}-\overrightarrow{\mu}}_2 \cdot \norm{\overrightarrow{\lambda}+\overrightarrow{\mu}}_2. 
    \end{align}
    By triangle inequality, we have 
    \begin{equation}
        \norm{\overrightarrow{\lambda}+\overrightarrow{\mu}}_2 \leq \norm{\overrightarrow{\lambda}}_2 + \norm{\overrightarrow{\mu}}_2\leq \sqrt{t_{\rho}} + \sqrt{t_{\sigma}} \leq 2.
    \end{equation}
    For the term $\norm{\overrightarrow{\lambda}-\overrightarrow{\mu}}_2$, Mirsky's inequality \cite{Mirsky1960} together with the fact that $\mathcal{F}$ preserves the Hilbert-Schmidt norm yield
    \begin{equation}
        \norm{\overrightarrow{\lambda}-\overrightarrow{\mu}}_2 \leq \norm{\mathcal{F}(\rho) - \mathcal{F}(\sigma)}_2 = \norm{\rho - \sigma}_2. 
    \end{equation}
    Therefore, we obtain 
    \begin{equation}
        \tilde{T} = \frac{1}{2}\left( \norm{\overrightarrow{\lambda^2}-\overrightarrow{\mu^2}}_1 + \abs{t_{\rho}-t_{\sigma}} \right)\leq \frac{1}{2}\left( 2\norm{\rho - \sigma}_2 + 2\norm{\rho - \sigma}_2 \right) = 2\norm{\rho - \sigma}_2.
    \end{equation}
    When $\norm{\rho - \sigma}_2 \leq 1/4$, we have $0\leq \tilde{T}\leq 2\norm{\rho - \sigma}_2 \leq 1/2$. 
    Since the binary entropy function $h(x)$ is monotonically increasing for $x \in [0, 1/2]$, we have
    \begin{equation}\label{eqn:: bound the first two terms}
       \tilde{T}\log_2(\chi_0-1) + h(\tilde{T})\leq 2\norm{\rho - \sigma}_2 \log_2(\chi_0) + h(2\norm{\rho - \sigma}_2).
    \end{equation}
    Lastly, by Fannes' inequality \cite{Fannes1973}[Lemma 2], we have 
    \begin{equation}\label{eqn:: bound the last term}
        \begin{aligned}
            \abs{\eta(1-t_{\rho}) - \eta(1-t_{\sigma})} &\leq 2\abs{(1-t_{\rho}) - (1-t_{\sigma})} + \eta((1-t_{\rho}) - (1-t_{\sigma})) \\
            &= 2\abs{t_{\rho} - t_{\sigma}} + \eta(\abs{t_{\rho} - t_{\sigma}}) \leq 4 \norm{\rho - \sigma}_2 + \eta(2\norm{\rho - \sigma}_2).
        \end{aligned}
    \end{equation} 
    Now combining \eqref{eqn:: bound the first two terms} and \eqref{eqn:: bound the last term}, we obtain the desired inequality. 
\end{proof}

\section{Dissipation driven by depolarizing noise}

\subsection{Purity decay under depolarizing noise}

We first consider single-qubit depolarizing noise channel defined as:
\begin{equation}
    \mathcal{N}(\sigma) = (1 - \lambda) \sigma + \lambda \frac{\mathbb{I}}{2},
\end{equation}
where $\lambda \in [0, 1]$ is the depolarizing probability and $\sigma$ is a single-qubit state.
For an $n$-qubit state $\rho$, if each qubit undergoes an independent depolarizing noise channel, the overall noise channel can be expressed as $\mathcal{N}^{\otimes n}(\rho)$.

For an $n$-qubit circuit with $L$ layers, if depolarizing noise is applied after each layer, the state $\rho_L$ after the $L$-th layer can be expressed as:
\begin{equation}
    \rho_L = \mathcal{N}^{\otimes n} \circ \mathcal{U}_L \circ \cdots \circ \mathcal{N}^{\otimes n} \circ \mathcal{U}_1 (\rho),
\end{equation}
where $\mathcal{U}_i$ is the unitary channel corresponding to the $i$-th layer of the circuit.

Using the hypercontractivity inequality for the depolarizing noise channel, we can derive an upper bound on the purity of the state $\rho_L$ after $L$ layers of depolarizing noise, as stated in the following lemma.

\begin{lemma}\label{lemma:hypercontractive_purity_bound}
Let
\begin{equation}
    \eta \coloneq 1-\lambda,
    \quad
    x_L \coloneq \eta^{2L}=(1-\lambda)^{2L},
    \quad
    D \coloneq 2^n .
\end{equation}
Then the purity of the state $\rho_L$, $\tr\{\rho_L^2\}$, is upper bounded as: 
\begin{equation}
    \tr\{\rho_L^2\}
    \leq
    D^{1-\frac{2}{1+x_L}}
    \left(
        \tr (\rho^{1+x_L})
    \right)^{\frac{2}{1+x_L}} .
\end{equation}
In particular, if $\rho$ is a pure state, we have
\begin{equation}
    \tr\{\rho_L^2\}
    \leq
    D^{-\frac{1-x_L}{1+x_L}}
    =
    2^{-n\frac{1-(1-\lambda)^{2L}}{1+(1-\lambda)^{2L}}}.
\end{equation}
Equivalently, defining
\begin{equation}
    \mu \coloneq -L\log(1-\lambda),
\end{equation}
we have
\begin{equation}
    \tr\{\rho_L^2\}
    \leq
    2^{-n\tanh \mu}.
\end{equation}
\end{lemma}
\begin{proof}
Let $D=2^n$ and define
\begin{equation}
    \norm{X}_{p,\tau}
    =
    \left(
        \frac{1}{D} \tr \{\abs{X}^p\}
    \right)^{1/p}.
\end{equation}
The hypercontractivity inequality~\cite{king2014hypercontractivity} for products of qubit depolarizing channels states that
\begin{equation}
    \norm{\mathcal N_\lambda^{\otimes n}(X)}_{q,\tau}
    \leq
    \norm{X}_{p,\tau},
\end{equation}
whenever $1<p\leq q$ and $p-1 \geq (1-\lambda)^2(q-1)$, and we can take $p_*=1+(1-\lambda)^2(q-1)$ to be the optimal input exponent for a given output exponent $q$.

Since each unitary channel $\mathcal U_l(X)=U_l XU_l^\dagger$ preserves the Schatten norms, we can apply this inequality layer by layer from the last layer to the first.

Let $r_L \coloneq 2$, the hypercontractivity inequality gives
\begin{equation}
    \norm{\mathcal N_\lambda^{\otimes n}\circ\mathcal U_L(\rho_{L-1})}_{r_L,\tau}
    \leq
    \norm{\mathcal U_L(\rho_{L-1})}_{r_{L-1},\tau},
\end{equation}
for $r_{L-1}=1+(1-\lambda)^2(r_L-1)=1+(1-\lambda)^2$.

The unitary layer does not change the norm $\norm{\mathcal U_L(\rho_{L-1})}_{r_{L-1},\tau}=\norm{\rho_{L-1}}_{r_{L-1},\tau}$, thus
\begin{equation}
    \norm{\rho_L}_{2,\tau}
    \leq
    \norm{\rho_{L-1}}_{1+(1-\lambda)^2,\tau}.
\end{equation}

We now repeat the same argument for the $(L-1)$-th depolarizing layer, which gives
\begin{equation}
    \norm{\rho_L}_{2,\tau}
    \leq
    \norm{\rho}_{1+(1-\lambda)^{2L},\tau}.
\end{equation}

Now repeat the same argument all the way to the first layer. 
We then obtain the following upper bound on the purity: 
\begin{equation}
    \tr\{\rho_L^2\}
    \leq
    D\norm{\rho}_{1+x_L,\tau}^2.
\end{equation}

Finally, since $\rho$ is a density operator, it is positive semidefinite.
Therefore $\abs{\rho}=\rho$, and
\begin{equation}
    \norm{\rho}_{1+x_L,\tau}
    =
    \left(
        \frac{1}{D} \tr(\rho^{1+x_L})
    \right)^{\frac{1}{1+x_L}}.
\end{equation}
Substituting this into the purity bound gives
\begin{equation}
    \tr\{\rho_L^2\}
    \leq
    D^{1-\frac{2}{1+x_L}}
    \left(
        \tr(\rho^{1+x_L})
    \right)^{\frac{2}{1+x_L}}.
\end{equation}

In particular, if $\rho$ is pure then $\tr(\rho^{1+x_L})=1$, hence
\begin{equation}
    \tr\{\rho_L^2\}
    \leq
    D^{1-\frac{2}{1+x_L}}
    =
    D^{-\frac{1-x_L}{1+x_L}}.
\end{equation}
Substituting $D=2^n$ and $x_L=(1-\lambda)^{2L}$ gives the stated bound.
Finally, since
\begin{equation}
    x_L=(1-\lambda)^{2L}=e^{-2\mu},
    \quad
    \mu=-L\log(1-\lambda),
\end{equation}
we have
\begin{equation}
    \frac{1-x_L}{1+x_L}
    =
    \frac{1-e^{-2\mu}}{1+e^{-2\mu}}
    =
    \tanh\mu.
\end{equation}
This completes the proof.
\end{proof}

\subsection{Operator entanglement bounds under depolarizing noise}

Combining the hypercontractive purity estimate with the purity-controlled OEE bound (\cref{theorem:maximum_entropy}) gives the following bounds for the OEE.
\begin{lemma}[OEE bounds from hypercontractivity]
\label{lemma:asymptotic_normalized_OE_bound_hypercontractivity}
Assume that the input state is pure, and fix a bipartition $A\mid B$. Let
\begin{equation}
    n_{\min}=\min\{\abs{A},\abs{B}\},
    \qquad
    x_L=(1-\lambda)^{2L},
    \qquad
    \alpha_L = \frac{1-x_L}{1+x_L},
    \qquad
    \beta_L = \frac{2x_L}{1+x_L},
    \qquad
    y_L=n\beta_L .
\end{equation}
If $n_{\min}=0$, then $S_{OE}(\rho_L)=\widetilde S_{OE}(\rho_L)=0$. If $n_{\min}\geq 1$ and $y_L\leq 2n_{\min}$, then
\begin{equation}\label{eqn:OE_bound_hypercontractivity}
    S_{OE}(\rho_L)\leq n2^{-n\alpha_L} + \left(2^{-n\alpha_L}-2^{-n}\right)\log_2\left(\frac{2^{2n_{\min}}-1}{2^{y_L}-1}\right),
\end{equation}
and
\begin{equation}\label{eqn:normalized_OE_bound_hypercontractivity}
    \widetilde S_{OE}(\rho_L) \leq y_L+ \left(1-2^{-y_L}\right) \log_2\left(\frac{2^{2n_{\min}}-1}{2^{y_L}-1}\right).
\end{equation}
If $n_{\min} \geq 1$ and $y_L\geq 2 n_{\min}$, then
\begin{equation}
    S_{OE}(\rho_L) \leq 2^{-n\alpha_L}\left(2n_{\min}+n\alpha_L\right),\qquad \widetilde S_{OE}(\rho_L) \leq 2n_{\min} .
\end{equation}
In particular, for a balanced cut $n_{\min}=\frac{n}{2}$, the first case always applies and gives
\begin{equation}
    S_{OE}(\rho_L) \leq \left(\log_2 \left(\frac{1}{2^{n(1-\alpha_L)}-1}\right)+2n\right) 2^{-n\alpha_L},
\end{equation}
and
\begin{equation}
    \widetilde S_{OE}(\rho_L) \leq n\beta_L+\left(1-2^{-n\beta_L}\right) \log_2\left(\frac{2^n-1}{2^{n\beta_L}-1}\right).
\end{equation}
\end{lemma}

\begin{proof}
Let $D=2^n$ and
\begin{equation}
    T_L\coloneq 2^{-n\alpha_L}.
\end{equation}
By the hypercontractive purity bound for layerwise product depolarizing noise, given by \cref{lemma:hypercontractive_purity_bound}, we have
\begin{equation}
    \tr\{\rho_L^2\} \leq T_L.
\end{equation}
The operator Schmidt rank across $A\mid B$ is at most $D_{\min}^2=2^{2n_{\min}}$, where $D_{\min}=2^{n_{\min}}$ and $D_{\max}=2^{n-n_{\min}}$. If $n_{\min}=0$, the rank is one and both entropies vanish.

Assume $n_{\min}\geq 1$. Since
\begin{equation}
    T_L D = 2^{n(1-\alpha_L)} = 2^{y_L},
\end{equation}
the condition $y_L \leq 2n_{\min}$ is equivalent to
\begin{equation}
    T_L \leq \frac{D_{\min}}{D_{\max}} .
\end{equation}
In this regime, monotonicity of the first-two-regime bound in \cref{theorem:property_maximum_entropy_part1} allows us to substitute the purity upper bound $T_L$ into \cref{theorem:maximum_entropy}. 
This gives
\begin{equation}
    \begin{aligned}
    S_{OE}(\rho_L) &\leq T_L\log_2D+ \left(T_L-\frac1D\right) \log_2\left(\frac{D_{\min}^2-1}{T_LD-1}\right)\\
    &= n2^{-n\alpha_L} + \left(2^{-n\alpha_L}-2^{-n}\right) \log_2\left(\frac{2^{2n_{\min}}-1}{2^{y_L}-1}\right),
    \end{aligned}
\end{equation}
and the normalized bound follows similarly:
\begin{equation}
    \widetilde S_{OE}(\rho_L) \leq
    \log_2(T_LD) + \left(1-\frac{1}{T_LD}\right) \log_2\left(\frac{D_{\min}^2-1}{T_LD-1}\right),
\end{equation}
which is the displayed expression because $T_LD = 2^{y_L}$.

If $y_L \geq 2n_{\min}$, then $T_L\geq \frac{D_{\min}}{D_{\max}}$. The two purity regimes in \cref{theorem:maximum_entropy} meet continuously at $D_{\min}/D_{\max}$, and the corresponding bounds are monotone on both sides by \cref{theorem:property_maximum_entropy_part1,theorem:property_maximum_entropy_part2}. Hence evaluating the third-regime expression at $T_L$ gives
\begin{equation}
    S_{OE}(\rho_L) \leq T_L\log_2\frac{D_{\min}^2}{T_L} = 2^{-n\alpha_L}(2n_{\min}+n\alpha_L),
\end{equation}
and
\begin{equation}
    \widetilde S_{OE}(\rho_L) \leq \log_2 D_{\min}^2 = 2n_{\min}.
\end{equation}

For a balanced cut, $2n_{\min}=n$ and $y_L=n\beta_L\leq n$, so the first case always applies. Substituting $n_{\min}=n/2$ into the first-case bounds and using
\begin{equation}
    n2^{-n\alpha_L} + 2^{-n\alpha_L}\log_2\left(\frac{2^n}{2^{n(1-\alpha_L)}-1}\right) = \left(\log_2 \left(\frac{1}{2^{n(1-\alpha_L)}-1}\right)+2n\right) 2^{-n\alpha_L}
\end{equation}
gives the stated balanced-cut specialization.
\end{proof}

Given the exponential dependence of the computational complexity of tensor-network simulation on the (normalized) operator entanglement, as discussed in \cref{sec:complexity_tensor_network}, the above result implies that the computational complexity of simulating the noisy quantum circuit with depolarizing noise can be significantly reduced compared to the noiseless case, especially when the number of layers $L$ is large.
Specifically, we have the following theorem:
\begin{theorem}[Depolarization-induced separation of simulation thresholds (\cref{thm:main_depolarizing_thresholds} in the main text)]\label{thm:: depolarizing_threshhold_full_statement}
Consider an $n$-qubit circuit of depth $L$ and a pure input state, with independent single-qubit depolarizing noise of fixed strength $\lambda>0$ applied after each layer, and let $\rho_L$ denote the output state. 
Fix any bipartition $A \mid B$ and let
\begin{equation}
    n_{\min}=\min\{\abs{A},\abs{B}\}.
\end{equation} 
If $n_{\min}=0$, then $S_{OE}(\rho_L)=\widetilde S_{OE}(\rho_L)=0$. Otherwise, the OEEs of $\rho_L$ exhibit two distinct crossover depths,
\begin{equation}
    L_{\mathrm{abs}} = \order{1}
    \qquad
    L_{\mathrm{rel}} = \order{\log n}
\end{equation}
such that $S_{OE}(\rho_L)= \order{\log{n}}$ for $L \geq  L_{\mathrm{abs}}$, whereas the same scaling for $\widetilde{S}_{OE}(\rho_L)$ emerges at $L \geq  L_{\mathrm{rel}}$.
\end{theorem}
\begin{proof}
Let
\begin{equation}
    D_{\min}=2^{n_{\min}}.
\end{equation}
The product depolarizing purity estimate is independent of the bipartition. Thus, with
\begin{equation}
    x_L=(1-\lambda)^{2L}, \quad \alpha_L=\frac{1-x_L}{1+x_L}, \quad \beta_L=\frac{2x_L}{1+x_L},
\end{equation}
using \cref{lemma:hypercontractive_purity_bound}, we have
\begin{equation}
    t_L\coloneq \tr\{\rho_L^2\} \leq 2^{-n\alpha_L}.
\end{equation}

We first prove the absolute-error, unnormalized-OEE statement. Across $A\mid B$, the normalized Schmidt distribution has support size at most $D_{\min}^2$, so
\begin{equation}
    \widetilde S_{OE}(\rho_L) \leq \log_2 D_{\min}^2 = 2n_{\min} \leq n .
\end{equation}
By \cref{lemma:normalized_oe_entropy},
\begin{equation}
    S_{OE}(\rho_L)= t_L\widetilde S_{OE}(\rho_L) - t_L\log_2 t_L \leq n t_L - t_L\log_2 t_L .
\end{equation}
Choose any constant depth $L_{\mathrm{abs}}$ for which $\alpha_{L_{\mathrm{abs}}}>0$. 
For all $L\geq L_{\mathrm{abs}}$, $t_L\leq T_n\coloneq 2^{-n\alpha_{L_{\mathrm{abs}}}}$. 
For sufficiently large $n$, $T_n\leq e^{-1}$, and the function $-u\log_2u$ is increasing on $0<u\leq e^{-1}$. Hence
\begin{equation}
    -t_L\log_2 t_L \leq -T_n \log_2 T_n = n\alpha_{L_{\mathrm{abs}}}T_n .
\end{equation}
Therefore
\begin{equation}
    S_{OE}(\rho_L) \leq n(1+\alpha_{L_{\mathrm{abs}}})2^{-n\alpha_{L_{\mathrm{abs}}}} = \order{\log n },
\end{equation}
after absorbing finitely many small values of $n$ into the constant. This proves the constant-depth upper bound for every bipartition.

We now prove the normalized-OEE statement. Let
\begin{equation}
    c_\lambda=-2\ln(1-\lambda), \qquad L_{\mathrm{rel}} = \left\lceil \frac{2\ln n-\ln\ln n}{c_\lambda} \right\rceil,
\end{equation}
and set $y_L=n\beta_L$. 
We have
\begin{equation}
    \beta_L= \frac{2e^{-c_\lambda L}}{1+e^{-c_\lambda L}}
    \leq 2e^{-c_\lambda L}.
\end{equation}
By the choice of $L_{\mathrm{rel}}$, for any $L\geq L_{\mathrm{rel}}$ we have
\begin{equation}
    e^{-c_\lambda L} \leq \frac{\ln n}{n^2},
\end{equation}
and therefore
\begin{equation}
    \beta_L \leq \frac{2\ln n}{n^2},
\end{equation}
which implies that 
\begin{equation}
    y_L = n\beta_L \leq \frac{2\ln n}{n}.
\end{equation}
Since the function $\frac{2\ln x}{x}$ is strictly decreasing for $x\geq 1$, for any positive integer $n$, we have $0<y_L\leq1$. Since $n_{\min}\geq1$, we also have $y_L\leq 2n_{\min}$, and therefore \cref{eqn:normalized_OE_bound_hypercontractivity} of \cref{lemma:asymptotic_normalized_OE_bound_hypercontractivity} applies:
\begin{equation}
    \widetilde S_{OE}(\rho_L)\leq y_L + \left(1-2^{-y_L}\right)\log_2\left(\frac{D_{\min}^2-1}{2^{y_L}-1} \right).
\end{equation}
Because $D_{\min}^2\leq2^n$, this is upper bounded by the balanced-cut expression
\begin{equation}
    \widetilde S_{OE}(\rho_L)\leq y_L + \left(1-2^{-y_L}\right)\log_2\left(\frac{2^n-1}{2^{y_L}-1}\right).
\end{equation}
Using $2^y-1\geq y\ln2$ and $1-2^{-y}\leq y$ for $0<y\leq1$, we obtain
\begin{equation}
    \widetilde S_{OE}(\rho_L) \leq y_L + y_L \log_2\left(\frac{2^n}{y_L \ln2}\right)
    = y_L + y_L \left(n+\log_2\frac1{y_L}-\log_2\ln2\right) = \order{y_L \left[ n+\log_2\frac1{y_L} \right]}.
\end{equation}
Using
\begin{equation}
    y_L=n\beta_L \leq \frac{2\ln n}{n},
\end{equation}
we obtain
\begin{equation}
    y_L n = \order{\log n}.
\end{equation}
Moreover,
\begin{equation}
    y_L\log_2\frac1{y_L} = \order{1}
\end{equation}
Hence
\begin{equation}
    \widetilde S_{OE}(\rho_L) = \order{\log n }
\end{equation}
for all $L\geq L_{\mathrm{rel}}$.

Finally, since $c_\lambda>0$ is a constant for fixed $\lambda>0$, we have
\begin{equation}
    L_{\mathrm{rel}}
    = \left\lceil \frac{2\ln n-\ln\ln n}{c_\lambda}\right\rceil = \order{\log n}.
\end{equation}
This proves the arbitrary-bipartition normalized-OEE upper bound at the same logarithmic depth scale.
\end{proof}

\begin{prop}[Sharpness of the logarithmic relative-error scale]\label{prop:main_depolarizing_log_scale_sharp}
The logarithmic depth scale in the preceding theorem is optimal at the level of normalized operator entanglement. Let $n$ be even and consider the balanced bipartition $A\mid B$ with $\abs{A}=\abs{B}=n/2$.
Let
\begin{equation}
    \Phi=\ketbra{\Phi^+}{\Phi^+},\qquad \ket{\Phi^+}=\frac{\ket{00}+\ket{11}}{\sqrt2},
\end{equation}
and initialize
\begin{equation}
    \rho_0=\Phi^{\otimes n/2},
\end{equation}
namely a product of Bell pairs across $A|B$. For fixed $\lambda\in(0,1)$, define
\begin{equation}
    \rho_L=(\mathcal D_\lambda^{\otimes n})^L(\rho_0), \qquad x_L=(1-\lambda)^{4L}.
\end{equation}
This is an admissible instance of the preceding noisy-circuit setting, obtained by taking every unitary layer to be the identity.
Then the normalized operator entanglement across $A|B$ is exactly
\begin{equation}\label{eq:main_product_bell_exact_noe}
    \widetilde S_{OE}(\rho_L) = \frac{n}{2} \left[ \log_2(1+3x_L) - \frac{3x_L}{1+3x_L}\log_2 x_L \right].
\end{equation}
Consequently, for every fixed $\delta\in(0,1)$, there exists $n_0(\delta,\lambda)$ such that for all $n\ge n_0(\delta,\lambda)$ and all
\begin{equation}\label{eq:main_product_bell_log_window}
    L\leq \frac{1-\delta}{4|\log(1-\lambda)|}\log n,
\end{equation}
one has
\begin{equation}\label{eq:main_product_bell_noe_lower}
    \widetilde S_{OE}(\rho_L)
    \ge
    \frac34(1-\delta)n^\delta\log_2 n.
\end{equation}
Thus $\widetilde S_{OE}(\rho_L)$ remains super-logarithmic throughout a logarithmic initial depth window.
Therefore the relative-error crossover in any uniform normalized-OEE upper bound cannot occur at depth $L=o(\log n)$.
\end{prop}

\begin{proof}
We first compute the normalized operator entanglement of one noisy Bell pair. The Bell state has the Pauli expansion
\begin{equation}
    \Phi = \frac{1}{4} \left( I\otimes I+X\otimes X-Y\otimes Y+Z\otimes Z \right).
\end{equation}
The single-qubit depolarizing channel satisfies
\begin{equation}
    \mathcal D_\lambda(I)=I,\qquad
    \mathcal D_\lambda(P)=(1-\lambda)P,
\end{equation}
for any Pauli operator $P\in\{X,Y,Z\}$.
Hence, if
\begin{equation}
    \tau_L=(\mathcal D_\lambda^{\otimes2})^L(\Phi),
    \qquad
    p_L=(1-\lambda)^{2L},
\end{equation}
then
\begin{equation}
    \tau_L=\frac{1}{4} \left(I\otimes I+p_L X\otimes X-p_L Y\otimes Y+p_L Z\otimes Z\right).
\end{equation}
Introduce the Hilbert-Schmidt orthonormal basis
\begin{equation}
    F_0=\frac{I}{\sqrt2},\qquad
    F_1=\frac{X}{\sqrt2},\qquad
    F_2=\frac{Y}{\sqrt2},\qquad
    F_3=\frac{Z}{\sqrt2}.
\end{equation}
Then
\begin{equation}
    \tau_L=\frac{1}{2} F_0\otimes F_0 +\frac{p_L}{2}F_1\otimes F_1-\frac{p_L}{2}F_2\otimes F_2+\frac{p_L}{2}F_3\otimes F_3 .
\end{equation}
The sign in the $Y\otimes Y$ term may be absorbed into one Schmidt vector, so the operator Schmidt coefficients of $\tau_L$ across $A_i \mid B_i$ are
\begin{equation}
    s_0=\frac{1}{2},\qquad
    s_1=s_2=s_3=\frac{p_L}{2}.
\end{equation}
Therefore
\begin{equation}
    \tr(\tau_L^2)=\sum_{\alpha=0}^3s_\alpha^2=\frac14(1+3p_L^2),
\end{equation}
and the normalized operator-Schmidt probabilities are
\begin{equation}
    q_0=\frac{1}{1+3p_L^2}, \quad q_1=q_2=q_3=\frac{p_L^2}{1+3p_L^2}.
\end{equation}
With $x_L=p_L^2=(1-\lambda)^{4L}$, the one-pair normalized OEE is
\begin{equation}
    \widetilde S_{OE}(\tau_L) = -\sum_{\alpha=0}^3q_\alpha\log_2 q_\alpha
    = \log_2(1+3x_L) - \frac{3x_L}{1+3x_L}\log_2x_L .
\end{equation}

For the full system, $\rho_0=\Phi^{\otimes n/2}$ and the product depolarizing channel factorizes over qubits, hence
\begin{equation}
    \rho_L=\tau_L^{\otimes n/2}.
\end{equation}
Under tensor products across the same bipartition, the normalized operator-Schmidt probabilities multiply. The Shannon entropy of a product distribution is additive, so
\begin{equation}
    \widetilde S_{OE}(\rho_L) =
    \frac {n}{2} \widetilde S_{OE}(\tau_L),
\end{equation}
which proves \eqref{eq:main_product_bell_exact_noe}.

It remains to prove the lower bound. Define
\begin{equation}
    f(x)=\log_2(1+3x)-\frac{3x}{1+3x}\log_2x, \quad 0<x \leq 1 .
\end{equation}
Since
\begin{equation}
    f'(x) = -\frac{3\log_2x}{(1+3x)^2}\geq 0 \qquad (0<x\le1),
\end{equation}
$f$ is increasing on $(0,1]$. Also, because the first term in $f$ is nonnegative,
\begin{equation}
    f(x) \geq \frac{3x}{1+3x}\log_2\frac1x .
\end{equation}
If $0< x \leq 1/3$, then $1+3x \leq 2$, and therefore
\begin{equation}\label{eq:main_product_bell_f_lower}
    f(x)\geq \frac{3}{2} x\log_2\frac{1}{x} .
\end{equation}
Assume \eqref{eq:main_product_bell_log_window}. Since $\log(1-\lambda)<0$, this gives
\begin{equation}
    x_L=(1-\lambda)^{4L} \geq n^{-(1-\delta)}.
\end{equation}
Set $y_n=n^{-(1-\delta)}$. For all sufficiently large $n$, $0< y_n \leq \frac{1}{3}$. By monotonicity of $f$ and by \eqref{eq:main_product_bell_f_lower},
\begin{equation}
    f(x_L) \geq f(y_n) \geq \frac{3}{2}y_n\log_2\frac1{y_n} = \frac{3}{2}n^{-(1-\delta)}(1-\delta)\log_2n .
\end{equation}
Combining this with \eqref{eq:main_product_bell_exact_noe} yields
\begin{equation}
    \widetilde S_{OE}(\rho_L) = \frac {n}{2}f(x_L) \geq \frac{3}{4}(1-\delta)n^\delta\log_2 n,
\end{equation}
which proves \eqref{eq:main_product_bell_noe_lower}. Since $n^\delta\log n$ grows faster than $\log n$ for every fixed $\delta>0$, the normalized OEE is super-logarithmic throughout the window \eqref{eq:main_product_bell_log_window}. This establishes the claimed $\Theta(\log n)$ lower scale for any uniform normalized-OEE crossover.
\end{proof}

Let $\chi_{\mathrm{abs}}(\rho_L)$ and $\chi_{\mathrm{rel}}(\rho_L)$ denote the bond dimensions required for fixed absolute error and fixed relative error, respectively.
As discussed in \cref{sec:complexity_tensor_network}, the bond dimension for fixed absolute error scales as $\chi_{\mathrm{abs}}(\rho_L) \sim \exp{S_{OE}(\rho_L)}$, while the bond dimension for fixed relative error scales as $\chi_{\mathrm{rel}}(\rho_L) \sim \exp{\widetilde{S}_{OE}(\rho_L)}$.
This separation reveals a noise-induced entanglement barrier: local noise suppresses the operator entanglement relevant to absolute accuracy much earlier than it suppresses the normalized operator entanglement relevant to relative accuracy.

\subsection{Efficient whole-trajectory simulation for 1D noisy circuits}\label{app:whole_trajectory}

\begin{lemma}[Simultaneous truncation to an MPO]\label{lemma:cutwise_tail_to_mpo}
Let $X$ be an operator on a one-dimensional chain of $n$ sites, and let $\chi$ be a positive integer.
For the cut between sites $j$ and $j+1$, let $e_j(\chi;X)$ be the squared truncation error of the ordered operator Schmidt decomposition of $X$ across that cut after keeping the largest $\chi$ coefficients.
The left-to-right tensor-train singular value decomposition (TT-SVD) produces an MPO $X_\chi$ with maximum bond dimension at most $\chi$ such that
\begin{equation}
    \norm{X-X_\chi}_2^2
    \leq
    \sum_{j=1}^{n-1} e_j(\chi;X).
\end{equation}
Consequently, if $e_j(\chi;X)\leq \delta\norm{X}_2^2$ for every bond $j$, then
\begin{equation}
    \norm{X-X_\chi}_2^2
    \leq
    (n-1)\delta\norm{X}_2^2 .
\end{equation}
\end{lemma}
\begin{proof}
Choose Hilbert--Schmidt orthonormal operator bases on the sites and vectorize $X$ as $|X\rangle\!\rangle$.
Vectorization is an isometry, so $\norm{|X\rangle\!\rangle}^2=\norm{X}_2^2$.
The case $n=1$ is trivial, so assume $n\geq2$.

For the cut $j\mid j+1$, let $A_j$ be the matrix obtained from the coefficient tensor of $|X\rangle\!\rangle$ by grouping sites $1,\ldots,j$ as row indices and sites $j+1,\ldots,n$ as column indices.
The singular values of $A_j$ are exactly the operator Schmidt coefficients of $X$ across this cut.
Hence, by the Eckart--Young theorem~\cite{eckart1936approximation},
\begin{equation}
    e_j(\chi;X)
    =
    \min_{\operatorname{rank} R\leq \chi}
    \norm{A_j-R}_F^2 .
\end{equation}

Apply the left-to-right truncated MPS decomposition to $|X\rangle\!\rangle$, and set $Q_0=1$.
At step $j$, let $Q_{j-1}$ be the left isometry obtained from the preceding steps.
Up to reshaping of row indices, the matrix decomposed at this step is
\begin{equation}
    B_j=(Q_{j-1}^\dagger\otimes I_j)A_j ,
\end{equation}
where $I_j$ is the identity on the doubled Hilbert space of site $j$.
Let $U_j$ be the isometry whose range is spanned by the leading $\min\{\chi,\operatorname{rank}B_j\}$ left singular vectors of $B_j$, and set
\begin{equation}
    Q_j=(Q_{j-1}\otimes I_j)U_j .
\end{equation}
The discarded squared norm is
\begin{equation}
    \eta_j^2
    =
    \min_{\operatorname{rank} S\leq\chi}\norm{B_j-S}_F^2 .
\end{equation}
For every matrix $R$ of rank at most $\chi$, the matrix $(Q_{j-1}^\dagger\otimes I_j)R$ also has rank at most $\chi$.
Since $Q_{j-1}^\dagger\otimes I_j$ is a contraction in Frobenius norm,
\begin{equation}
    \eta_j^2
    \leq
    \min_{\operatorname{rank} R\leq\chi}
    \norm{(Q_{j-1}^\dagger\otimes I_j)(A_j-R)}_F^2
    \leq
    e_j(\chi;X).
\end{equation}

Define orthogonal projections on the full doubled Hilbert space by
\begin{equation}
    \widehat{\Pi}_0=I,
    \qquad
    \widehat{\Pi}_j=Q_jQ_j^\dagger\otimes I_{j+1:n}.
\end{equation}
Because $\operatorname{ran}(Q_j)\subseteq\operatorname{ran}(Q_{j-1}\otimes I_j)$, these projections form a decreasing chain.
Consequently, the operators $D_j:=\widehat{\Pi}_{j-1}-\widehat{\Pi}_j$ are pairwise orthogonal projections.
Moreover,
\begin{equation}
    D_j
    =
    (Q_{j-1}\otimes I_j)(I-U_jU_j^\dagger)(Q_{j-1}^\dagger\otimes I_j)\otimes I_{j+1:n},
\end{equation}
so $\norm{D_j|X\rangle\!\rangle}^2=\norm{(I-U_jU_j^\dagger)B_j}_F^2=\eta_j^2$.
Setting $|\psi_\chi\rangle\!\rangle:=\widehat{\Pi}_{n-1}|X\rangle\!\rangle$, we obtain
\begin{equation}
    \norm{|X\rangle\!\rangle-|\psi_{\chi}\rangle\!\rangle}^2
    =
    \sum_{j=1}^{n-1}\eta_j^2
    \leq
    \sum_{j=1}^{n-1}e_j(\chi;X).
\end{equation}

Moreover, $\widehat{\Pi}_{n-1}\leq\widehat{\Pi}_j$ for every $j$, so $|\psi_\chi\rangle\!\rangle$ has Schmidt rank at most $\operatorname{rank}(Q_j)\leq\chi$ across every cut.
It therefore admits an MPS representation with maximum bond dimension at most $\chi$.
Reinterpreting each doubled physical index as a pair of operator indices gives an MPO $X_\chi$ with the same bond dimensions.
Since vectorization is an isometry,
\begin{equation}
    \norm{X-X_\chi}_2^2
    =
    \norm{|X\rangle\!\rangle-|\psi_{\chi}\rangle\!\rangle}^2
    \leq
    \sum_{j=1}^{n-1}e_j(\chi;X).
\end{equation}
This proves the stated estimate, and the final claim follows by substituting $e_j(\chi;X)\leq \delta\norm{X}_2^2$ for every $j$.
When $X$ is supplied as an MPO of maximum bond dimension $R$, first right-canonicalize its vectorization.
At cut $j$, the remaining suffix is an isometry $V_j$ from a bond space of dimension at most $R$ into the remaining doubled sites.
Consequently, $B_j=M_jV_j^\dagger$, where $M_j$ has at most $4\min\{\chi,R\}$ rows and $R$ columns for qubits.
Since $V_j^\dagger V_j=I$, the matrices $B_j$ and $M_j$ have the same nonzero singular values and left singular vectors.
Thus $O(n)$ factorizations of polynomial-size matrices implement the sweep without forming any $A_j$.
Exact decompositions here specify the mathematical truncation; the whole-trajectory construction below uses only local tensor updates on a finite initial segment.
\end{proof}

\begin{lemma}\label{lemma:purified_distance_from_purity}
    Let $\rho$ be a density matrix on a $d$-dimensional Hilbert space and put $t=\tr(\rho^2)$.
    With fidelity $F(\rho,\sigma)=\norm{\sqrt\rho\sqrt\sigma}_1$ and purified distance $P(\rho,\sigma)=\sqrt{1-F(\rho,\sigma)^2}$, we have
    \begin{equation}\label{eq:purified_maximally_mixed_bound}
    P(\rho,I/d)^2\le1-\frac1{dt}
    =\frac{\norm{\rho-I/d}_2^2}{\norm{\rho}_2^2}.
    \end{equation}
\end{lemma}

\begin{proof}
Let $p_1,\ldots,p_d$ be the eigenvalues of $\rho$.
H\"older's inequality with exponents $p = 3/2$ and $q = 3$ gives
\begin{equation}
 \begin{aligned}
 1=\sum_i p_i =\sum_i(\sqrt{p_i})^{2/3}(p_i^2)^{1/3} \le\left(\sum_i\sqrt{p_i}\right)^{2/3}
       \left(\sum_i p_i^2\right)^{1/3}.
 \end{aligned}
\end{equation}
Taking cubics of both sides yields $(\sum_i\sqrt{p_i})^2t\ge1$, and hence
\begin{equation}
 F(\rho,I/d)^2=\frac{(\sum_i\sqrt{p_i})^2}{d}\ge\frac1{dt}.
\end{equation}
This proves the purified-distance inequality.
Finally, $\tr\rho=1$ gives $\norm{\rho-I/d}_2^2=t-1/d$, proving the equality in \eqref{eq:purified_maximally_mixed_bound}.
\end{proof}

Consider an $n$-qubit circuit of depth $L$ on a one-dimensional chain,
\begin{equation}
 \rho_\ell=\Phi_\ell(\rho_{\ell-1}),\qquad
 \Phi_\ell=\mathcal D_\lambda^{\otimes n}\circ\mathcal U_\ell,\qquad
 \rho_0=\ketbra{0}{0}^{\otimes n},
\end{equation}
where $\lambda\in(0,1)$ is fixed and each unitary layer consists of disjoint gates on at most $\kappa$ consecutive sites, with fixed $\kappa$.

\paragraph*{Exact formulation of the algorithm.}
Set $d=2^n$ and let $T$ be the positive integer specified in the following proposition.
Perform the following updates:
\begin{enumerate}[label=\arabic*.,leftmargin=*,itemsep=3pt]
 \item Initialize the bond-one MPO for $\rho_0$ and output $\widehat\rho_0=\rho_0$.
 \item For $\ell=1,\ldots,\min\{L,T\}$, compute $\rho_\ell=\Phi_\ell(\rho_{\ell-1})$ by local MPO updates, retaining all tensor terms, and output $\widehat\rho_\ell=\rho_\ell$.
 \item If $L>T$, output $\widehat\rho_\ell=I/d$ for every $T<\ell\le L$.
\end{enumerate}
Thus the physical outputs are
\begin{equation}\label{eq:trajectory_physical_output}
 \widehat\rho_\ell=
 \begin{cases}
  \rho_\ell,&0\le\ell\le\min\{L,T\},\\[2pt]
  I/d,&T<\ell\le L.
 \end{cases}
\end{equation}
Every layer fixes $I/d$, so the internal state can remain $I/d$ after the switch.
The replacement at the switch introduces an approximation whose error is controlled below.

\begin{prop}[Constructive whole-trajectory simulation in purified distance]
\label{prop:whole_trajectory_purified_efficiency_general_circuit}
\label{prop:whole_trajectory_rel_efficiency_general_circuit}
Consider an $n$-qubit circuit with $n\ge1$ and depth $L\ge0$ on a one-dimensional chain, initialized in $\ketbra{0}{0}^{\otimes n}$, with disjoint unitary gates on at most $\kappa$ consecutive sites in each layer, followed by $\mathcal D_\lambda^{\otimes n}$, where $\lambda\in(0,1)$ and $\kappa$ are fixed.
Then for any tolerance $0<\varepsilon<1$, the simulation above produces density operators $\widehat\rho_0,\ldots,\widehat\rho_L$ such that, for every $0\le\ell\le L$,
\begin{equation}\label{eq:trajectory_physical_guarantees}
 \begin{aligned}
    P(\rho_\ell,\widehat\rho_\ell)&\le\varepsilon,\quad \norm{\rho_\ell-\widehat\rho_\ell}_2^2 \le\varepsilon^2\norm{\rho_\ell}_2^2.
 \end{aligned}
\end{equation}
The maximum bond dimension is $\mathrm{poly}(n,1/\varepsilon)$, independent of $L$.
Writing $q=1-\lambda$, one explicit cutoff is
\begin{equation}\label{eq:trajectory_purified_parameters}
 T=\left\lceil\frac{\ln(8n/\varepsilon^2)}{2\ln(1/q)}\right\rceil.
\end{equation}
\end{prop}

\begin{proof}
Write $q=1-\lambda$ and $d=2^n$. 
Put $t_\ell=\tr(\rho_\ell^2)$.
Then \cref{lemma:hypercontractive_purity_bound} gives 
\begin{equation}\label{eq:trajectory_purity_bound}
    dt_\ell\le2^{\frac{2nq^{2\ell}}{1+q^{2\ell}}}.
\end{equation}
In \eqref{eq:trajectory_purified_parameters}, $T$ is chosen such that $q^{2T}\le\varepsilon^2/(8n)$. 
Since $0<q<1$, for every $\ell\ge T$ we have 
\begin{equation}
   \frac{2nq^{2\ell}}{1+q^{2\ell}}\leq 2nq^{2\ell}\leq  2nq^{2T}\le\varepsilon^2/4. 
\end{equation}
Therefore, \cref{lemma:purified_distance_from_purity} yields
\begin{equation}\label{eq:trajectory_late_accuracy}
 \begin{aligned}
    P(\rho_\ell,I/d)^2 \le\frac{\norm{\rho_\ell-I/d}_2^2}{\norm{\rho_\ell}_2^2}
    =1-\frac1{dt_\ell} \le1-2^{-\varepsilon^2/4}
    \le\frac{\ln2}{4}\varepsilon^2
    <\frac{\varepsilon^2}{4}.
 \end{aligned}
\end{equation}
Here the second last inequality follows from $1-e^{-x}\le x$ for $x\ge0$.
Since the algorithm simulate each $\rho_l$ before $T$ exactly and the remaining outputs are $I/d$, \eqref{eq:trajectory_physical_guarantees} follows.

Since the iput state is a product state, it is a MPO with bond dimension one. 
At any cut of the chain, at most one gate per layer crosses the cut.
A gate supported on at most $\kappa$ sites has operator Schmidt rank at most $4^{\min\{a,b\}}\le2^\kappa$, where $a+b\le\kappa$ counts its sites on the two sides.
Conjugation therefore increases a density-MPO bond by at most $4^\kappa$, while on-site noise leaves spatial bonds unchanged.
Exact local tensor updates through time $\min\{L,T\}$ consequently give
\begin{equation}\label{eq:trajectory_purified_bonds}
 \widehat R_{\max}\le4^{\kappa T}
 \le4^\kappa\left(\frac{8n}{\varepsilon^2}\right)^{\frac{\kappa\ln4}{2\ln(1/q)}}.
\end{equation}
Every later output $I/d=(I/2)^{\otimes n}$ has bond dimension one.
The bound is polynomial because the noise strength and gate range are fixed.
\end{proof}

\begin{remark}[MPO construction]
    The initial segment can be represented by a locally purified tensor network with amplitude $A_\ell$, so that $\rho_\ell=A_\ell A_\ell^\dagger$ and $\norm{A_\ell}_2=1$.
Initialize $A_0=\ket{0}^{\otimes n}$, apply each gate to its system indices, and realize each single-qubit depolarizing channel by a local isometry with a four-dimensional environment.
The crossing-gate count in Step~2 bounds the spatial purification bond by $2^{\kappa T}$, while the environmental dimension at each site is at most $4^T$.
Both dimensions are polynomial in $n,1/\varepsilon$ because $T=O(\log(n/\varepsilon))$ for fixed noise strength and gate range.
Contracting each local environmental index between $A_\ell$ and its conjugate produces a positive semidefinite, trace-one MPO with spatial bond at most $4^{\kappa T}$.
For $L<T$ this construction stops at time $L$; for $L=0$ it consists only of the product initial state.
\end{remark}

\begin{corollary}[Trace-distance, observable, and absolute-error guarantees]
\label{cor:whole_trajectory_physical_error_guarantees}
\label{prop:whole_trajectory_abs_efficiency_general_circuit}
The physical outputs of \cref{prop:whole_trajectory_purified_efficiency_general_circuit}
satisfy, at every time,
\begin{equation}
    \norm{\rho_\ell-\widehat\rho_\ell}_1\le2\varepsilon,
    \qquad
    \norm{\rho_\ell-\widehat\rho_\ell}_2^2\le4\varepsilon^2.
\end{equation}
For any observable $O$,
\begin{equation}
    \abs{\tr[O(\rho_\ell-\widehat\rho_\ell)]}
    \le2\varepsilon\norm{O}_\infty.
\end{equation}
In particular, any prescribed squared absolute Hilbert--Schmidt tolerance
$\zeta\in(0,1)$ is achieved by running the algorithm with
$\varepsilon=\sqrt\zeta/2$, with polynomial bond dimension.
\end{corollary}
\begin{proof}
For normalized density matrices, the trace distance is bounded by purified
distance~\cite{tomamichel2012framework}:
\begin{equation}
    \frac12\norm{\rho_\ell-\widehat\rho_\ell}_1
    \le P(\rho_\ell,\widehat\rho_\ell)
    \le\varepsilon.
\end{equation}
The remaining inequalities follow from $\norm{Z}_2\le\norm{Z}_1$ and
$\abs{\tr(OZ)}\le\norm{O}_\infty\norm{Z}_1$.
\end{proof}

\section{Dissipation driven by general noise}\label{appendix:general_noise}

\subsection{Single-qubit noise}
Here we consider the general single qubit noise channel $\mathcal{N}$. Using the following lemma, we can obtain the normal form of single-qubit noise channels:
\begin{lemma}[Normal form of a quantum channel~\cite{king2001minimal,ruskai2002analysis}]
    Any single-qubit quantum channel $\mathcal{N}$ can be brought into a canonical form via pre- and post-unitary rotations such that 
    \begin{equation}
        \mathcal{N}(\cdot) = U \mathcal{N}'(V^\dagger (\cdot) V) U^\dagger,
    \end{equation}
    where $U$ and $V$ are single-qubit unitary operators, and $\mathcal{N}'$ is a quantum channel with the following Pauli transfer matrix $\left(\mathcal{S}_{\mathcal{N}'}\right)_{i,j}=\frac{1}{2}\tr{\sigma_i\mathcal{N}'(\sigma_j)}$ associated with the Pauli basis $\{\sigma_0 = \mathbb{I}, \sigma_1 = X, \sigma_2 = Y, \sigma_3 = Z\}$:
    \begin{equation}
        \mathcal{S}_{\mathcal{N}'} = \begin{pmatrix}
            1 & 0 & 0 & 0 \\
            t_X & D_X & 0 & 0 \\
            t_Y & 0 & D_Y & 0 \\
            t_Z & 0 & 0 & D_Z
        \end{pmatrix},
    \end{equation}
    where $D_i \in \mathbb{R}$, $t_i \in \mathbb{R}$ for $i = X,Y,Z$, and the entries of $\vec{D}$ have the same sign.
\end{lemma}

The parameters $D_i$ and $t_i$ for a legitimate quantum channel $\mathcal{N}'$ must satisfy certain constraints, for example, there is a contraction condition holding:
\begin{lemma}[Contraction coefficient in terms of the normal form parameters~\cite{mele2026noise}]
    For any single-qubit quantum channel $\mathcal{N}'$, the parameters $\vec{D}$ and $\vec{t}$ satisfy:
    \begin{equation}\label{eq:c:contraction_condition}
        c \coloneq \frac{1}{3} \left( t_X^2 + t_Y^2 + t_Z^2 + D_X^2 + D_Y^2 + D_Z^2 \right) \leq 1,
    \end{equation}
    and the equality holds if and only if the quantum channel is unitary. Furthermore, there is $\norm{\vec{t}}_2 \leq 1 $ and $\norm{\vec{D}}_\infty \leq 1$.
\end{lemma}

\subsection{Average case analysis}\label{app:general_average_case}

We consider a noisy brickwall circuit on an open chain of $n$ qubits, as illustrated in \cref{fig:random_circuit}. In each unitary layer $U_\ell$, two-qubit gates act on disjoint nearest-neighbor pairs, with the pairing alternating between successive layers. Each two-qubit gate is drawn independently from a unitary $2$-design~\cite{dankert2009exact}. Any unpaired endpoint qubits undergo single-qubit gates that are either fixed (including the identity) or sampled independently of the two-qubit gates.

Let $\mathcal U_\ell(X)=U_\ell XU_\ell^\dagger$. After each unitary layer, the same single-qubit noise channel $\mathcal N$ acts independently on every qubit. The resulting depth-$L$ circuit is
\begin{equation}\label{eq:general_noisy_circuit}
    \Phi = \left( \mathcal{N}^{\otimes n} \circ \mathcal{U}_{L}\right) \circ \cdots \circ \left( \mathcal{N}^{\otimes n} \circ \mathcal{U}_1 \right).
\end{equation}
Starting from an initial state $\rho_0=\rho$, we denote the state after the $\ell$-th noisy layer by
\begin{equation}
    \rho_\ell=\mathcal N^{\otimes n}\bigl(\mathcal U_\ell(\rho_{\ell-1})\bigr),
    \qquad \ell=1,\ldots,L.
\end{equation}

\begin{figure}[ht]
    \centering
    \includegraphics[width=0.5\textwidth]{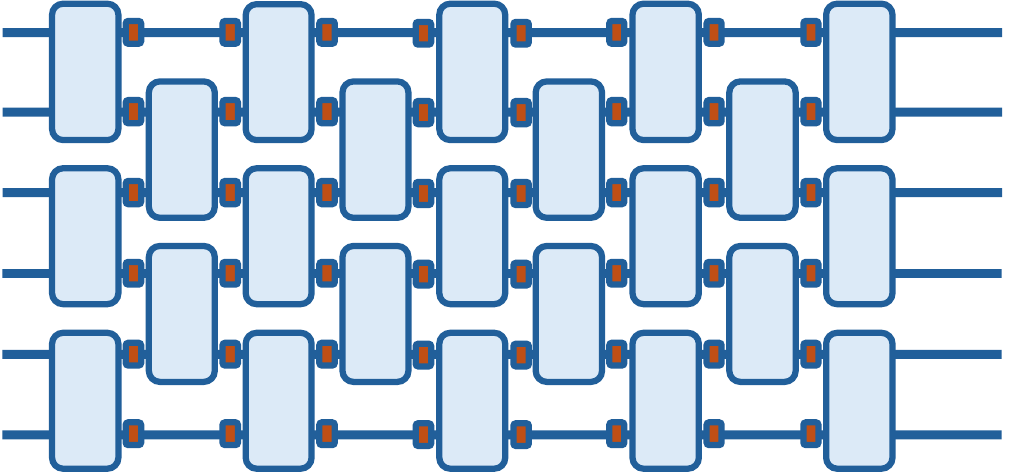}
    \caption{Schematic of a 1D circuit with nearest-neighbor two-qubit gates and noise applied after every layer of unitary gates. The noise is drawn in red color.}
    \label{fig:random_circuit}
\end{figure}

Averaging over the gates immediately before a noise layer gives a purity bound, which can be converted directly into an OEE bound.

\begin{lemma}[Average purity coefficient]
\label{lemma:average_purity_coefficient}
For a single-qubit channel $\mathcal N$ in canonical form, define
\begin{equation}
    a=\sum_{P\in\{X,Y,Z\}}t_P^2,\qquad
    d=\sum_{P\in\{X,Y,Z\}}D_P^2,\qquad
    \kappa_{\mathcal N}=\frac{1+a+d/3}{2}.
\end{equation}
Then $1/2\leq\kappa_{\mathcal N}\leq1$. Equality $\kappa_{\mathcal N}=1$ holds if and only if $\mathcal N$ is unitary or a pure-state replacement channel, $\mathcal N(X)=\nu\tr X$ with $\nu$ pure.
\end{lemma}
\begin{proof}
Write $t=(t_X,t_Y,t_Z)^{\mathsf T}$ and $D=\operatorname{diag}(D_X,D_Y,D_Z)$. For a pure input with Bloch vector $u\in\mathbb R^3$, $\lVert u\rVert_2=1$, the output Bloch vector is $t+Du$. Positivity of the output state gives $\lVert t+Du\rVert_2^2\leq1$ for every such $u$.

Now let $u$ be uniformly distributed over the unit sphere. Symmetry under $u\mapsto-u$ gives $\mathbb E[u_P]=0$ for each $P\in\{X,Y,Z\}$. Rotational symmetry implies that $\mathbb E[u_X^2]$, $\mathbb E[u_Y^2]$, and $\mathbb E[u_Z^2]$ are equal; since $u_X^2+u_Y^2+u_Z^2=1$, each equals $1/3$. Expanding the squared output Bloch-vector length and taking the expectation therefore gives
\begin{equation}
    \begin{aligned}
        \mathbb E\,\lVert t+Du\rVert_2^2
        &=\sum_{P\in\{X,Y,Z\}}\left(t_P^2+2t_PD_P\mathbb E[u_P]+D_P^2\mathbb E[u_P^2]\right)\\
        &=\sum_{P\in\{X,Y,Z\}}t_P^2+\frac13\sum_{P\in\{X,Y,Z\}}D_P^2
        =a+\frac d3.
    \end{aligned}
\end{equation}
Averaging the pointwise bound $\lVert t+Du\rVert_2^2\leq1$ thus yields
\begin{equation}\label{eq:average_bloch_length_bound}
    a+\frac d3\leq1,
\end{equation}
and hence $\kappa_{\mathcal N}=(1+a+d/3)/2\leq1$. Since $a,d\geq0$, we also have $\kappa_{\mathcal N}\geq1/2$.

If equality holds in \cref{eq:average_bloch_length_bound}, continuity and the pointwise upper bound imply $\lVert t+Du\rVert_2^2=1$ for every unit vector $u$. Expanding this identity gives
\begin{equation}
    a+2u^{\mathsf T}D^{\mathsf T}t+u^{\mathsf T}D^{\mathsf T}Du=1.
\end{equation}
Replacing $u$ by $-u$ leaves the constant and quadratic terms unchanged and reverses the sign of the linear term. Subtracting the two identities therefore gives $4u^{\mathsf T}D^{\mathsf T}t=0$ for every unit vector $u$, which implies $D^{\mathsf T}t=0$. Substituting this into the expanded identity yields
\begin{equation}
    u^{\mathsf T}D^{\mathsf T}Du=1-a
    \qquad\text{for every }\lVert u\rVert_2=1.
\end{equation}
Since $D$ is diagonal, choosing $u$ along each of the three coordinate axes gives $D_X^2=D_Y^2=D_Z^2=1-a$, and hence
\begin{equation}
    D^{\mathsf T}D=(1-a)I_3.
\end{equation}
When $a<1$, $D$ is invertible, hence $t=0$ and $D^{\mathsf T}D=I_3$. Thus $c=(a+d)/3=1$, and $\mathcal N$ is unitary by the characterization in \eqref{eq:c:contraction_condition}. When $a=1$, $D=0$, so the output is the same pure state for every input. The converse is immediate in both cases. Pre- and post-unitary rotations preserve $a$, $d$, and this classification, so the result applies to any single-qubit channel.
\end{proof}

\begin{lemma}[Average purity after noise channel layer]
\label{lemma:average_one_layer_purity}
Consider a unitary layer $\mathcal U$ in which disjoint two-qubit gates cover $n-r$ qubits, where $0\leq r\leq2$. Each two-qubit gate is drawn independently from a unitary $2$-design. The remaining $r$ qubits undergo single-qubit gates that are either fixed or sampled independently. For any $n$-qubit density operator $\omega$ independent of the gates in $\mathcal U$, we have
\begin{equation}\label{eq:average_one_layer_purity}
    \mathbb E\,\tr\left[\mathcal N^{\otimes n}\bigl(\mathcal U(\omega)\bigr)^2\right]
    \leq\kappa_{\mathcal N}^{\,n-r}
    \leq\kappa_{\mathcal N}^{\,n-2}.
\end{equation}
\end{lemma}
\begin{proof}
We first fix any random single-qubit gates on the uncovered sites. Since these gates are independent of the two-qubit gates, conditioning on them preserves the latter's $2$-design distributions. The bound below is uniform in the fixed single-qubit gates, so we can average over them at the end.
We may also assume that $\mathcal N$ is in canonical form: the pre-rotations can be absorbed into the gates of $\mathcal U$ without changing their $2$-design second moments, and the post-rotations do not affect purity.

Let $F$ swap two copies of a single qubit, let $\mathcal N^*$ denote the adjoint channel, and define
\begin{equation}\label{eq:average_purity_local_R}
    R=(\mathcal N^*\otimes\mathcal N^*)(F).
\end{equation}
For a fixed gate layer, write $\tau=\mathcal U(\omega)$. Pairing each qubit with its counterpart in the second copy, the swap identity gives
\begin{equation}
    \tr\left[\mathcal N^{\otimes n}(\tau)^2\right]
    =\tr\left[\tau^{\otimes2}R^{\otimes n}\right].
\end{equation}
Thus it remains to control the local factors $R$ after gate averaging.
Using $F=\frac12(I\otimes I+\sum_{P\in\{X,Y,Z\}}P\otimes P)$, together with $\mathcal N^*(I)=I$ and $\mathcal N^*(P)=t_PI+D_PP$, we obtain
\begin{equation}
    R=\frac12\left[(1+a)I\otimes I+
    \sum_{P\in\{X,Y,Z\}}\left(D_P^2P\otimes P+t_PD_P(P\otimes I+I\otimes P)\right)\right].
\end{equation}
To bound $\lVert R\rVert_\infty$, note that the swap operator satisfies $F^\dagger=F$ and $F^2=I$. Its eigenvalues are therefore $\pm1$, so both $I+F$ and $I-F$ are positive semidefinite. Since $\mathcal N$ is trace-preserving, its adjoint is unital, $\mathcal N^*(I)=I$. Moreover, $\mathcal N^*\otimes\mathcal N^*$ is completely positive and thus preserves positive semidefiniteness. Applying this map to $I\pm F$ gives
\begin{equation}
    I\pm R=(\mathcal N^*\otimes\mathcal N^*)(I\pm F)\geq0.
\end{equation}
Hence $-I\leq R\leq I$, so every eigenvalue of the Hermitian operator $R$ lies in $[-1,1]$. Since the operator norm of a Hermitian operator is the largest absolute value of its eigenvalues, we conclude that $\lVert R\rVert_\infty\leq1$.

Consider a two-qubit gate $G$ acting on sites $i,j$. Since $G$ is drawn from a unitary $2$-design, its second moments agree with those of a Haar-random unitary. For any fixed single-qubit unitaries $V_i,V_j$, the gate $(V_i\otimes V_j)G$ has the same second moments, by left invariance of Haar measure. The output purity involves two copies of each gate and its adjoint, so its average depends only on these second moments. Applying $V_i$ and $V_j$ after $G$ therefore leaves the average purity unchanged.

We apply this observation to every two-qubit gate in the layer. Let $C$ denote the $n-r$ covered sites, and choose single-qubit $2$-design rotations $V_j$ on these sites. Set $V_j=I$ on the uncovered sites and write $V=\bigotimes_{j=1}^n V_j$. Applying $V$ after the original unitary layer and before the noise gives
\begin{equation}\label{eq:average_purity_rotation_average}
    \mathbb E_{\mathcal U}\,\tr\left[\mathcal N^{\otimes n}(\tau)^2\right]
    =\mathbb E_{\mathcal U}\mathbb E_V\,\tr\left[\mathcal N^{\otimes n}(V\tau V^\dagger)^2\right],
    \qquad \tau=\mathcal U(\omega).
\end{equation}
We can thus keep the original gate layer fixed and first evaluate the inner average over the single-qubit rotations.

On a covered site $j$, cyclicity of the trace allows us to move $V_j$ from the state to the local factor $R$ defined in \cref{eq:average_purity_local_R}, turning it into $(V_j^\dagger\otimes V_j^\dagger)R(V_j\otimes V_j)$. We denote the average of this conjugated operator over $V_j$ by $\overline R$:
\begin{equation}
    \begin{aligned}
        \overline R
        &=\mathbb E_{V_j}\left[(V_j^\dagger\otimes V_j^\dagger)R(V_j\otimes V_j)\right]\\
        &=\frac12\left[(1+a)I\otimes I+\frac d3\sum_{P\in\{X,Y,Z\}}P\otimes P\right].
    \end{aligned}
\end{equation}
Using $\sum_{P\in\{X,Y,Z\}}P\otimes P=2F-I\otimes I$, we can rewrite this as
\begin{equation}
    \overline R=\frac{1+a-d/3}{2}I\otimes I+\frac d3F.
\end{equation}
The symmetric subspace consists of two-copy vectors unchanged by the swap, $F\ket{\psi}=\ket{\psi}$, while the antisymmetric subspace consists of vectors whose sign is reversed, $F\ket{\psi}=-\ket{\psi}$. These are the $+1$ and $-1$ eigenspaces of $F$, respectively. On each subspace, $\overline R$ therefore acts as multiplication by a scalar, giving the eigenvalues
\begin{equation}
    \begin{aligned}
        \lambda_+&=\frac{1+a-d/3}{2}+\frac d3
        =\frac{1+a+d/3}{2}=\kappa_{\mathcal N},\\
        \lambda_-&=\frac{1+a-d/3}{2}-\frac d3
        =\frac{1+a-d}{2}.
    \end{aligned}
\end{equation}
The inequality $d/3\leq1-a$ from \cref{eq:average_bloch_length_bound} gives $-\lambda_+\leq\lambda_-\leq\lambda_+$, and hence $\lVert\overline R\rVert_\infty=\kappa_{\mathcal N}$.

We now combine these local estimates to bound the output purity. Fix the original gate layer, so that $\tau=\mathcal U(\omega)$ is fixed, and write
\begin{equation}
    \rho_V=\mathcal N^{\otimes n}(V\tau V^\dagger)
\end{equation}
for the output after inserting the single-qubit rotations.
Let $F_n$ exchange the two $n$-qubit copies. When corresponding qubits in the two copies are paired, $F_n=\bigotimes_{j=1}^n F_j$, where $F_j$ is the single-qubit swap $F$ acting on the $j$-th pair. The swap identity expresses the purity as
\begin{equation}
    \tr(\rho_V^2)=\tr\left[\rho_V^{\otimes2}F_n\right].
\end{equation}
Moving the noise channels from the state to the swap operator by trace adjointness gives
\begin{equation}
    \begin{aligned}
        \tr(\rho_V^2)
        &=\tr\left[(V\tau V^\dagger)^{\otimes2}
        \left((\mathcal N^*)^{\otimes n}\otimes(\mathcal N^*)^{\otimes n}\right)(F_n)\right]\\
        &=\tr\left[(V\tau V^\dagger)^{\otimes2}R^{\otimes n}\right].
    \end{aligned}
\end{equation}
The second equality follows because each local swap $F_j$ is mapped to the operator $R$ defined in \cref{eq:average_purity_local_R}.
Next, cyclicity of the trace moves the inserted rotations to the operator side:
\begin{equation}
    \tr(\rho_V^2)
    =\tr\left[\tau^{\otimes2}
    (V^\dagger\otimes V^\dagger)R^{\otimes n}(V\otimes V)\right].
\end{equation}
Since $V=\bigotimes_{j=1}^n V_j$, the conjugated operator factorizes over the paired sites as
\begin{equation}
    (V^\dagger\otimes V^\dagger)R^{\otimes n}(V\otimes V)
    =\bigotimes_{j=1}^n R_j(V_j),
    \qquad
    R_j(V_j)=(V_j^\dagger\otimes V_j^\dagger)R(V_j\otimes V_j).
\end{equation}
On a covered site, $\mathbb E_{V_j}R_j(V_j)=\overline R$ by the calculation above. On an uncovered site, $V_j=I$, so $R_j(V_j)=R$.
Independence of the rotations across sites therefore gives
\begin{equation}
    \mathbb E_V\left[\bigotimes_{j=1}^n R_j(V_j)\right]
    =\bigotimes_{j=1}^n\mathbb E_{V_j}R_j(V_j)
    =T,
\end{equation}
where
\begin{equation}
    T=\bigotimes_{j=1}^n T_j,
    \qquad
    T_j=\begin{cases}
        \overline R,&j\in C,\\
        R,&j\notin C.
    \end{cases}
\end{equation}
Because $\tau$ is fixed while we average over $V$, linearity of the trace now yields
\begin{equation}
    \begin{aligned}
        \mathbb E_V\,\tr(\rho_V^2)
        &=\tr\left[\tau^{\otimes2}\,
        \mathbb E_V\left[\bigotimes_{j=1}^n R_j(V_j)\right]\right]\\
        &=\tr\left[\tau^{\otimes2}T\right].
    \end{aligned}
\end{equation}

To bound this expectation, note that $T$ is Hermitian, so $T\leq\lVert T\rVert_\infty I$. Since $\tau^{\otimes2}$ is positive with trace one,
\begin{equation}
    \tr\left[\tau^{\otimes2}T\right]
    \leq\lVert T\rVert_\infty\tr\left[\tau^{\otimes2}\right]
    =\lVert T\rVert_\infty.
\end{equation}
The operator norm is multiplicative under tensor products. There are $n-r$ covered sites and $r$ uncovered sites, so
\begin{equation}
    \lVert T\rVert_\infty
    =\prod_{j=1}^n\lVert T_j\rVert_\infty
    =\kappa_{\mathcal N}^{\,n-r}\lVert R\rVert_\infty^r
    \leq\kappa_{\mathcal N}^{\,n-r}.
\end{equation}
These steps use only positivity and normalization of $\tau$ and allow arbitrary correlations in the input $\omega$.

Finally, the bound holds for every fixed original gate layer. Using \cref{eq:average_purity_rotation_average} and averaging over that layer gives
\begin{equation}
    \mathbb E_{\mathcal U}\,\tr\left[\mathcal N^{\otimes n}\bigl(\mathcal U(\omega)\bigr)^2\right]
    =\mathbb E_{\mathcal U}\mathbb E_V\,\tr(\rho_V^2)
    \leq\kappa_{\mathcal N}^{\,n-r}.
\end{equation}
The same bound remains valid after averaging over any random single-qubit gates fixed at the start, because its right-hand side does not depend on their values. This proves the first inequality in \cref{eq:average_one_layer_purity}. The second follows from $r\leq2$ and $0<\kappa_{\mathcal N}\leq1$.
\end{proof}

\begin{theorem}[Average-case entanglement plateau in brickwall circuits (\cref{thm:main_average_case_plateau} in the main text)]
\label{cor:uniform_constant_oee}
Let $\mathcal N$ be a fixed single-qubit channel with $c<1$, and consider the noisy brickwall circuit on an open chain of $n\geq2$ qubits with a product initial state and $L\geq1$. Each two-qubit gate is drawn independently from a unitary $2$-design, independently across layers. Single-qubit gates at unmatched endpoints are fixed or sampled independently of these gates. Let $S_{OE}^{A_k}$ denote the OEE across $A_k\mid A_k^c$, where $A_k=\{1,\ldots,k\}$. Then
\begin{equation}\label{eq:average_nonunitary_oee_plateau}
    \Pr\left[
    S_{OE}^{A_k}(\rho_\ell)\leq2,
    \quad\forall k\in\{1,\ldots,n-1\},\quad
    \forall\ell\in\{0,\ldots,L\}\right]
    \geq1-Le^{-\Omega_{\mathcal N}(n)}.
\end{equation}
\end{theorem}
\begin{proof}
We first check depth $0$. Since $\rho_0$ is a product state, its operator Schmidt rank is one across every cut. Writing $t_0=\tr(\rho_0^2)$, we therefore have
\begin{equation}
    S_{OE}^{A_k}(\rho_0)=-t_0\log_2t_0
    \leq\frac{1}{e\ln2}<2,
    \qquad k=1,\ldots,n-1.
\end{equation}
This also covers mixed product initial states.
If $\mathcal N$ is a pure-state replacement channel, $\mathcal N(X)=\nu\tr X$ with $\nu$ pure, every noise layer prepares $\nu^{\otimes n}$. Hence $S_{OE}^{A_k}(\rho_\ell)=0$ for every cut and every $\ell\geq1$, and the claimed event holds deterministically.

It remains to consider channels that are not pure-state replacements. Since $c<1$ excludes unitary channels, \cref{lemma:average_purity_coefficient} gives $\kappa_{\mathcal N}<1$.
We first fix the single-qubit gates at all unmatched endpoints throughout the circuit. If these gates are random, this amounts to conditioning on their sampled values. The two-qubit gates are independent of the endpoint gates, so their joint distribution is unchanged by this conditioning. The estimates below hold with the same bound for every choice of endpoint gates, so averaging over these choices at the end preserves the bound.

Fix a layer $\ell\in\{1,\ldots,L\}$ and condition on the two-qubit gates in layers $1,\ldots,\ell-1$. The input $\rho_{\ell-1}$ is obtained by evolving $\rho_0$ through these earlier layers. Since the initial state, the noise channel, and all endpoint gates are already fixed, specifying the remaining gates in those layers completely determines $\rho_{\ell-1}$.
The two-qubit gates in layer $\ell$ are sampled independently of all earlier gates and of the endpoint gates. Their joint distribution is therefore unchanged by this conditioning. Thus we have a fixed input $\rho_{\ell-1}$ and an independent random gate layer to which \cref{lemma:average_one_layer_purity} applies. Since at most two sites are uncovered by two-qubit gates in a brickwall layer, we obtain
\begin{equation}
    \mathbb E_{\mathcal U_\ell}\,\tr\left[
    \mathcal N^{\otimes n}\bigl(\mathcal U_\ell(\rho_{\ell-1})\bigr)^2\right]
    \leq\kappa_{\mathcal N}^{\,n-2}.
\end{equation}
Here $\mathbb E_{\mathcal U_\ell}$ averages only over the two-qubit gates in layer $\ell$, with the preceding layers and all endpoint gates held fixed. The right-hand side does not depend on the preceding gates, so averaging over their values yields $\mathbb E\,\tr(\rho_\ell^2)\leq\kappa_{\mathcal N}^{\,n-2}$.
Set $t_\ell=\tr(\rho_\ell^2)$. Markov's inequality then gives
\begin{equation}
    \Pr\left[t_\ell>\frac1n\right]
    \leq n\,\mathbb E[t_\ell]
    \leq n\kappa_{\mathcal N}^{\,n-2}.
\end{equation}

If $t_\ell\leq1/n$, then $S_{OE}^{A_k}(\rho_\ell)\leq2$ for every spatial cut.
For any $n$-qubit state $\tau$, the operator Schmidt rank across $A_k\mid A_k^c$ is at most $4^{\min\{k,n-k\}}\leq2^n$, so its normalized OEE is at most $n$. Writing $t=\tr(\tau^2)$, \cref{lemma:normalized_oe_entropy} therefore gives
\begin{equation}
    S_{OE}^{A_k}(\tau)
    =t\widetilde S_{OE}^{A_k}(\tau)-t\log_2t
    \leq t(n-\log_2t).
\end{equation}
For $n\geq2$, the function $t(n-\log_2t)$ is increasing on $0<t\leq1$, since its derivative is $n-\log_2t-1/\ln2>0$. Thus, whenever $t_\ell\leq1/n$,
\begin{equation}
    S_{OE}^{A_k}(\rho_\ell)
    \leq\frac1n\left(n-\log_2\frac1n\right)
    =1+\frac{\log_2n}{n}\leq2
    \qquad\text{for every }k=1,\ldots,n-1.
\end{equation}

The desired OEE bound consequently holds at all depths and across all cuts whenever $t_\ell\leq1/n$ for every $\ell=1,\ldots,L$; depth $0$ was already handled above. A union bound over depths gives
\begin{equation}\label{eq:average_nonunitary_failure_probability}
    \begin{aligned}
        \Pr\left[\exists\,\ell\in\{1,\ldots,L\}:t_\ell>\frac1n\right]
        &\leq\sum_{\ell=1}^L\Pr\left[t_\ell>\frac1n\right]\\
        &\leq Ln\kappa_{\mathcal N}^{\,n-2}.
    \end{aligned}
\end{equation}
We only need to sum the failure probabilities over depths, because $t_\ell\leq1/n$ guarantees the OEE bound for all spatial cuts. The upper bound does not depend on the endpoint gates, so averaging over their random choices leaves it unchanged.
Finally, $\mathcal N$ is fixed and $\kappa_{\mathcal N}<1$, so
\begin{equation}
    n\kappa_{\mathcal N}^{\,n-2}
    =e^{-(n-2)|\ln\kappa_{\mathcal N}|+\ln n}
    =e^{-\Omega_{\mathcal N}(n)}.
\end{equation}
Together with the deterministic depth-$0$ bound, this proves \cref{eq:average_nonunitary_oee_plateau}.
\end{proof}

\subsection{Worst-case analysis}

We consider an $n$-qubit circuit with $L$ noisy layers,
\begin{equation}
    \Phi_\ell=\mathcal N^{\otimes n}\circ\mathcal U_\ell,
    \qquad
    \mathcal U_\ell(X)=U_\ell XU_\ell^\dagger,
    \qquad \ell=1,\ldots,L,
\end{equation}
where each $U_\ell$ is an arbitrary $n$-qubit unitary and the same fixed single-qubit noise channel $\mathcal N$ acts independently on all qubits after each unitary layer.
Starting from an arbitrary density operator $\rho_0=\rho$, we denote the state after the $\ell$-th noisy layer by
\begin{equation}
    \rho_\ell=\Phi_\ell(\rho_{\ell-1})
    =\left(\Phi_\ell\circ\cdots\circ\Phi_1\right)(\rho_0),
    \qquad \ell=1,\ldots,L.
\end{equation}
Thus $\rho_\ell$ is the state after both the $\ell$-th unitary layer and the following noise layer.
We first bound the output purity after one noise layer for an arbitrary input state, and then use this bound to control the OEE under strong contraction.

\begin{lemma}\label{lemma: exponential suppression of purity}
Let $\mathcal{N}$ be a single-qubit noise channel with $c<1/6$. 
Then for every integer $n\ge1$ and every $n$-qubit density operator $\omega$, we have 
\begin{equation}\label{eqn:: exponential suppression of purity}
 \tr\left[\mathcal{N}^{\otimes n}(\omega)^2\right]
 \le (\frac{1}{2} + 3c)^n.
\end{equation}
\end{lemma}

\begin{proof}
The post-unitary rotations in the canonical decomposition preserve output purity, and the pre-unitary rotations can be absorbed into the arbitrary input state, so it suffices to take $\mathcal{N}$ in canonical form.
Let $\mathrm{SWAP}$ be the two-qubit swap gate and $\mathcal{N}^*$ the adjoint of $\mathcal{N}$. 
Since $\mathcal{N}^*(\mathbb{I})=\mathbb{I}$ and $\mathcal{N}^*(P)=t_P\mathbb{I}+D_PP$ for $P\in\{X,Y,Z\}$, the Pauli expansion $\mathrm{SWAP}=\frac12\left(\mathbb{I}\otimes\mathbb{I}+\sum_{P\in\{X,Y,Z\}}P\otimes P\right)$ gives
\begin{equation}
\begin{aligned}
    R&\coloneq(\mathcal{N}^*\otimes\mathcal{N}^*)(\mathrm{SWAP}) =\frac12\Bigg[\left(1+\sum_{P\in\{X,Y,Z\}}t_P^2\right)\mathbb{I}\otimes\mathbb{I}\\
    &\qquad+\sum_{P\in\{X,Y,Z\}}\left(D_P^2P\otimes P+t_PD_P(P\otimes\mathbb{I}+\mathbb{I}\otimes P)\right)\Bigg].
\end{aligned}
\end{equation}
Using the operator norm, $\norm{P}=1$, the triangle inequality, and $2|t_PD_P|\le t_P^2+D_P^2$, we obtain
\begin{equation}
\begin{aligned}
    \norm{R}
    &\le\frac12\left[1+\sum_{P\in\{X,Y,Z\}}\left(t_P^2+D_P^2+2|t_PD_P|\right)\right]\\
    &\le\frac12+\sum_{P\in\{X,Y,Z\}}(t_P^2+D_P^2)
    =\frac12+3c.
\end{aligned}
\end{equation}

Let $\mathrm{SWAP}_n$ exchange the two $n$-qubit registers; pairing corresponding qubits identifies it with $\mathrm{SWAP}^{\otimes n}$.
The identity $\tr[(A\otimes B)\mathrm{SWAP}_n]=\tr(AB)$ and trace adjointness then give, with the same pairing of tensor factors,
\begin{equation}
\begin{aligned}
    \tr\left[\mathcal{N}^{\otimes n}(\omega)^2\right]
    &=\tr\left[\left(\mathcal{N}^{\otimes n}(\omega)\otimes\mathcal{N}^{\otimes n}(\omega)\right)\mathrm{SWAP}_n\right]\\
    &=\tr\left[(\omega\otimes\omega)R^{\otimes n}\right]\\
    &\le\norm{R^{\otimes n}}_\infty
    =\norm{R}_\infty^n
    \le\left(\frac12+3c\right)^n.
\end{aligned}
\end{equation}
Here $\omega\otimes\omega$ is positive with trace one, so the expectation of the Hermitian operator $R^{\otimes n}$ is at most its operator norm, without any product assumption on $\omega$.
\end{proof}

\begin{theorem}\label{thm:: purity induced bound of SOE under strong contraction}
Let $n\ge2$ and let $\mathcal{N}$ be a fixed single-qubit channel in canonical form with $c<1/6$.
Set $T_n=(\frac12+3c)^n$ and, for a nontrivial bipartition $A\mid B$, let $n_{\min}=\min\{|A|,|B|\}$.
For any $n$-qubit input state, arbitrary unitary layers, and every depth $1\leq\ell\leq L$ in the above circuit, the output satisfies
\begin{equation}\label{eq:strong_noise_purity_entropy_refined}
    S_{OE}(\rho_\ell)
    \le\begin{cases}
    nT_n+(T_n-2^{-n})\log_2\!\left(\dfrac{2^{2n_{\min}}-1}{2^nT_n-1}\right),
        & T_n\le2^{2n_{\min}-n},\\[4pt]
    T_n(2n_{\min}-\log_2T_n),
        & T_n>2^{2n_{\min}-n}.
    \end{cases}
\end{equation}
At $c=0$, the first branch is understood by continuity and equals $n2^{-n}$.
In particular, across every nontrivial bipartition,
\begin{equation}\label{eq:strong_noise_entropy_exponential}
    S_{OE}(\rho_\ell)
    \le n\left(\frac12+3c\right)^n\log_2\frac{2}{\frac12+3c} = \mathcal{O}(1).
\end{equation}
\end{theorem}

\begin{proof}
For $\ell\ge1$, write $\rho_\ell=\mathcal{N}^{\otimes n}(\omega)$, where $\omega=\mathcal U_\ell(\rho_{\ell-1})$ is the density operator immediately before the noise layer.
By \cref{lemma: exponential suppression of purity},
\begin{equation}
    2^{-n}\le t\coloneq\tr(\rho_\ell^2)\le T_n<1.
\end{equation}
For the chosen cut, $D_{\min}=2^{n_{\min}}$ and $D_{\max}=2^{n-n_{\min}}$.
The two entropy bounds in \cref{theorem:maximum_entropy} are increasing on their respective purity intervals by \cref{theorem:property_maximum_entropy_part1,theorem:property_maximum_entropy_part2}, and both equal $n2^{2n_{\min}-n}$ at their junction $t=2^{2n_{\min}-n}$.
Thus substituting $T_n$ for $t$ gives \cref{eq:strong_noise_purity_entropy_refined}.

For the simpler bound, the operator Schmidt rank is at most $2^{2n_{\min}}\le2^n$, so \cref{lemma:normalized_oe_entropy} yields
\begin{equation}
    S_{OE}(\rho_\ell)
    \le t(n-\log_2t)
    \le T_n(n-\log_2T_n)
    =n\left(\frac12+3c\right)^n\log_2\frac{2}{\frac12+3c}.
\end{equation}
The substitution is justified by $\frac{d}{dt}[t(n-\log_2t)]=n-\log_2t-1/\ln2>0$ for $n\ge2$ and $0<t\le1$. 
\end{proof}

\section{Generalization to higher dimensions}
\label{appendix:pepo}

This section makes precise the cut-rank and average bond dimension statements for projected entangled-pair operators~(PEPOs), the operator analogue of projected entangled-pair states~(PEPS)~\cite{orus2014practical,kshetrimayum2017simple}. The statements are cutwise: an actual PEPO of bounded virtual dimension imposes a Schmidt-rank constraint across a prescribed bipartition, while a Schmidt-rank truncation only induces an average boundary-bond dimension scale for that cut. 
They do not assert that such an average boundary-bond dimension estimate can be converted into an efficiently contractible PEPO, nor do they by themselves construct local PEPO tensors that realize the truncated Schmidt vectors or satisfy all cuts simultaneously.

Let $G=(V,E)$ be a finite interaction graph. A PEPO for an operator $X$ on the sites $V$ is a tensor-network representation with one physical input-output pair at each vertex and virtual indices on the edges. 
We consider a prescribed tensor-network cut: the PEPO is coarse-grained into two blocks supported on $A$ and $A^c$, with the virtual edges between the two blocks left open. If an edge $e$ has virtual dimension $\chi_e$, we denote this boundary-edge set by
\begin{equation}
    \partial A=\{e=(u,v)\in E:\,u\in A,\ v\in A^c\},
\end{equation}
and we write $a(A)=\abs{\partial A}$.

Let $X$ be represented by a PEPO on $G$ with virtual dimensions $\{\chi_e\}_{e\in E}$.
For the prescribed cut, fix all boundary virtual indices and contract the tensors and internal virtual indices within each of the two blocks. 
This produces an operator $L_{\boldsymbol{\alpha}}$ supported on $A$ and an operator $R_{\boldsymbol{\alpha}}$ supported on $A^c$, where $\boldsymbol{\alpha}$ is the multi-index formed by the boundary virtual indices. Therefore the PEPO can be written as
\begin{equation}\label{eq:pepo_schmidt_decomposition}
    X = \sum_{\boldsymbol{\alpha}} L_{\boldsymbol{\alpha}}\otimes R_{\boldsymbol{\alpha}},
\end{equation}
where the number of possible boundary multi-indices is $\prod_{e\in\partial A}\chi_e$. This expansion is not necessarily an orthonormal operator Schmidt decomposition, but it immediately bounds the actual operator Schmidt rank:
\begin{equation}
    \mathrm{rank}_{A \mid A^c}(X)
    \leq
    \prod_{e\in\partial A}\chi_e .
\end{equation}
To make the operator-entanglement notation explicit, let the orthonormal operator Schmidt decomposition of $X$ across the same cut be
\begin{equation}
    X
    =
    \sum_{\beta=1}^{\mathrm{rank}_{A \mid A^c}(X)}
    \lambda_{\beta}^{(A)}
    \widetilde L_{\beta}^{A}
    \otimes
    \widetilde R_{\beta}^{A},
\end{equation}
where $\{\widetilde L_{\beta}^{A}\}$ and $\{\widetilde R_{\beta}^{A}\}$ are Hilbert-Schmidt orthonormal families. Set
\begin{equation}
    t_X
    =
    \norm{X}_2^2
    =
    \sum_{\beta=1}^{\mathrm{rank}_{A|A^c}(X)}
    \bigl(\lambda_{\beta}^{(A)}\bigr)^2,
    \qquad
    p_{\beta}^{(A)}
    =
    \frac{\bigl(\lambda_{\beta}^{(A)}\bigr)^2}{t_X}.
\end{equation}
The normalized operator entanglement across $A|A^c$ is the Shannon entropy of this normalized operator-Schmidt distribution,
\begin{equation}
    \widetilde S_{OE}^{A}(X)
    =
    -\sum_{\beta=1}^{\mathrm{rank}_{A|A^c}(X)}
    p_{\beta}^{(A)}
    \log_2 p_{\beta}^{(A)} .
\end{equation}
The PEPO expansion above implies that this distribution is supported on at most $\prod_{e\in\partial A}\chi_e$ values, and therefore
\begin{equation}
    \widetilde S_{OE}^{A}(X)
    \leq
    \sum_{e\in\partial A}\log_2\chi_e .
\end{equation}
For uniform virtual dimension, this becomes $\widetilde S_{OE}^{A}(X)\leq a(A)\log_2\chi_{\mathrm{PEPO}}$. The unnormalized entropy differs only by the Hilbert-Schmidt norm factor: if $t=\norm{X}_2^2$, then $S_{OE}^{A}(X)=t\widetilde S_{OE}^{A}(X)-t\log_2 t$.

\begin{definition}\label{def:effective_boundary_bond_dimension}
    For a cut $A \mid A^c$ with $a(A)>0$ and a specified target Schmidt rank $R \geq 1$, the associated average boundary-bond dimension is defined as
    \begin{equation}
    \label{eq:effective_boundary_bond_dimension}
        \overline\chi_{\partial}(A) \coloneqq R^{\frac{1}{a(A)}}.
    \end{equation}
    If the target rank needs to be displayed explicitly, we write the same quantity as $\overline\chi_{\partial}(A;R)$.
    Equivalently,
    \begin{equation}
        \overline\chi_{\partial}(A)^{a(A)} = R,
        \quad
        \log_2 \overline\chi_{\partial}(A) = \frac{\log_2 R}{a(A)} .
    \end{equation}
\end{definition}
The average boundary-bond dimension records only the amount of boundary label space needed to store $R$ Schmidt labels after coarse-graining the two sides of the cut; it does not assert that the corresponding Schmidt vectors factor into local PEPO tensors. 

For any constant absolute error tolerance $\varepsilon > 0$ or relative error tolerance $\delta \in (0,1)$, we denote the corresponding Schmidt rank that meets the error criterion by $R_{\mathrm{abs}}$ or $R_{\mathrm{rel}}$, respectively (constructed by \cref{thm:D_abs_error,thm:D_rel_error}).
With these definitions, the average boundary-bond dimensions for the two error criteria are denoted by
\begin{equation}
    \overline\chi_{\partial}^{\mathrm{abs}}(A) \coloneqq R_{\mathrm{abs}}^{\frac{1}{a(A)}},
    \qquad
    \overline\chi_{\partial}^{\mathrm{rel}}(A) \coloneqq R_{\mathrm{rel}}^{\frac{1}{a(A)}}.
\end{equation}

\begin{prop}[Cutwise average boundary-bond dimension from Schmidt-rank truncation]
\label{prop:pepo_cutwise_scale}
Fix a bipartition $A|A^c$ with PEPO boundary size $a(A)>0$. 
The entropy-to-rank bounds of \cref{thm:D_abs_error,thm:D_rel_error} give the cutwise estimates
\begin{equation}
    \log_2\overline\chi_{\partial}^{\mathrm{abs}}(A)
    \leq
    \frac{1}{a(A)}
    \log_2
    \max\left\{1,\left\lceil
    \left(\tr{\rho^2}-\varepsilon\right)
    2^{S_{OE}^{A}(\rho)/\varepsilon}
    \right\rceil\right\}
    \leq
    \frac{1+S_{OE}^{A}(\rho)/\varepsilon}{a(A)},
\end{equation}
for any absolute error tolerance $\varepsilon > 0$, and
\begin{equation}
    \log_2\overline\chi_{\partial}^{\mathrm{rel}}(A)
    \leq
    \frac{1}{a(A)}
    \log_2
    \max\left\{1,\left\lceil
    (1-\delta)
    2^{\widetilde S_{OE}^{A}(\rho)/\delta}
    \right\rceil\right\}
    \leq
    \frac{1+\widetilde S_{OE}^{A}(\rho)/\delta}{a(A)},
\end{equation}
for any relative error tolerance $\delta\in(0,1)$.
For fixed error tolerances, these bounds scale as
\begin{equation}
    \log\overline\chi_{\partial}^{\mathrm{abs}}(A)
    =
    O\!\left(\frac{1+S_{OE}^{A}(\rho)}{a(A)}\right),
    \qquad
    \log\overline\chi_{\partial}^{\mathrm{rel}}(A)
    =
    O\!\left(\frac{1+\widetilde S_{OE}^{A}(\rho)}{a(A)}\right).
\end{equation}
\end{prop}
\begin{proof}
For a fixed bipartition $A|A^c$, \cref{thm:D_abs_error,thm:D_rel_error} are statements about truncating the ordered operator Schmidt spectrum across that cut. Their proofs do not use one-dimensional geometry; the same bounds apply to the PEPO cutwise approximation.

For the absolute-error estimate, when fixing the error tolerance $\varepsilon$,  \cref{thm:D_abs_error} gives an operator Schmidt rank
\begin{equation}
    R_{\mathrm{abs}}
    \leq
    \max\left\{1,\left\lceil
    \left(\tr{\rho^2}-\varepsilon\right)
    2^{S_{OE}^{A}(\rho)/\varepsilon}
    \right\rceil\right\}
\end{equation}
for $0<\varepsilon<\tr{\rho^2}$. Since $\lceil x\rceil\leq x+1$ and $\tr{\rho^2}-\varepsilon\leq1$, this implies
\begin{equation}
    R_{\mathrm{abs}}\leq 1+2^{S_{OE}^{A}(\rho)/\varepsilon}.
\end{equation}
By the definition of average boundary-bond dimension, this implies
\begin{equation}
    \log_2\overline\chi_{\partial}^{\mathrm{abs}}(A)
    =
    \frac{1}{a(A)}\log_2 R_{\mathrm{abs}} .
\end{equation}
Combining this with the entropy-to-rank bound gives
\begin{equation}
    \log_2\overline\chi_{\partial}^{\mathrm{abs}}(A)
    \leq
    \frac{1}{a(A)}
    \log_2\left(1+2^{S_{OE}^{A}(\rho)/\varepsilon}\right)
    \leq
    \frac{1}{a(A)}
    \left(1+\frac{S_{OE}^{A}(\rho)}{\varepsilon}\right),
\end{equation}
where the last inequality uses $\log_2(1+2^x)\leq1+x$ for $x\geq0$. If $\varepsilon\geq\tr{\rho^2}$, discarding the whole operator already meets the absolute Hilbert-Schmidt error tolerance, so the average boundary-bond dimension requirement is trivial.

The relative-error estimate follows in the same way from \cref{thm:D_rel_error},
\begin{equation}
    R_{\mathrm{rel}}
    \leq
    \max\left\{1,\left\lceil
    (1-\delta)
    2^{\widetilde S_{OE}^{A}(\rho)/\delta}
    \right\rceil\right\}.
\end{equation}
Since $\lceil x\rceil\leq x+1$ and $1-\delta\leq1$, we have
\begin{equation}
    \log_2\overline\chi_{\partial}^{\mathrm{rel}}(A)
    =
    \frac{1}{a(A)}\log_2 R_{\mathrm{rel}}
    \leq
    \frac{1}{a(A)}
    \left(1+\frac{\widetilde S_{OE}^{A}(\rho)}{\delta}\right).
\end{equation}
The final scaling relations follow for fixed $\varepsilon$ and $\delta$.
\end{proof}

For the depolarizing noise thresholds of \cref{thm:main_depolarizing_thresholds}, the operator purity bounds are independent of the spatial geometry.
Therefore, similar results for the average PEPO boundary-bond dimension hold for every cut, as stated in the following.
\begin{prop}[Depolarizing thresholds as average PEPO boundary-bond dimension bounds]
\label{prop:depolarizing_pepo_effective_boundary_bond_dimension}
Consider the depolarizing setting of \cref{thm:: depolarizing_threshhold_full_statement}. 
For every bipartition $A\mid A^c$, there exist two critical circuit depths $L_{\mathrm{abs}}$ and $L_{\mathrm{rel}}$ satisfying
\begin{equation}
    L_{\mathrm{abs}}=\order{1},
    \quad
    L_{\mathrm{rel}} = \Theta(\log n),
\end{equation}
such that for all $L\geq  L_{\mathrm{abs}}$ and $L\geq  L_{\mathrm{rel}}$, respectively, the operator entanglement across the cut is bounded by
\begin{equation}
    S_{OE}^{A}(\rho_L) = \order{\log n},
    \quad
    \widetilde S_{OE}^{A}(\rho_L) = \order{\log n}.
\end{equation}
\end{prop}

The preceding depolarizing estimate uses purity decay and is essentially independent of the geometry of the cut. We now add locality to the picture. In a local circuit, the only gates that can create new Schmidt labels across $A|A^c$ are gates whose support intersects both sides of the cut. For bounded-degree lattices and bounded-size gates, the number of such gates in one layer is proportional to the number $a(A)$ of boundary bonds, and each of them can increase the operator Schmidt rank by only a constant factor. This gives a purely local, cutwise rank-counting estimate.

\begin{prop}[Locality implies linear cutwise Schmidt-rank growth]
\label{prop:locality_cutwise_schmidt_rank_growth}
Consider a local circuit on a bounded-degree lattice of qubits, with gates supported on at most $\kappa=\order{1}$ sites, and fix a bipartition $A \mid A^c$ with PEPO boundary size $a(A)>0$.
Assume that each layer contains at most $C_{\partial} \cdot a(A)$ gates crossing the cut.
Locality implies that each crossing gate increases the operator Schmidt rank by at most the constant factor
\begin{equation}
    C_g
    \coloneqq 2^{4\lfloor\kappa/2\rfloor}.
\end{equation}
Writing the $\ell$-th noisy layer as $\Phi_\ell=\mathcal N_\ell\circ\mathcal U_\ell$, where $\mathcal N_\ell$ is a product of single-site quantum channels, define $\sigma_L = \Phi_L\circ\Phi_{L-1}\circ\cdots\circ\Phi_1(\ketbra{0}{0}^{\otimes n})$.
Then
\begin{equation}
    \mathrm{rank}_{A|A^c}(\sigma_L) \leq C_g^{C_{\partial} L a(A)},
\end{equation}
and therefore
\begin{equation}
    S_{OE}^{A}(\sigma_L) = \order{L a(A)}, \quad \widetilde S_{OE}^{A}(\sigma_L) = \order{L a(A)}.
\end{equation}
\end{prop}
\begin{proof}
Let a crossing gate or channel $\mathcal G$ have support $S=S_A\cup S_{A^c}$, where $S_A\subset A$ and $S_{A^c}\subset A^c$, with $\abs{S}\leq\kappa$.
After vectorizing operators, the local Liouville dimensions on the two sides are
\begin{equation}
    2^{2\abs{S_A}} \quad\text{and}\quad 2^{2\abs{S_{A^c}}}.
\end{equation}
The spaces of linear maps on these Liouville spaces consequently have dimensions $2^{4\abs{S_A}}$ and $2^{4\abs{S_{A^c}}}$.
Expanding $\mathcal G$ in a basis of linear maps on the smaller side gives
\begin{equation}
    \begin{aligned}
        \mathcal G&=\sum_{\mu=1}^{K}\mathcal A_\mu\otimes\mathcal B_\mu,\\
        K&\leq2^{4\min(\abs{S_A},\abs{S_{A^c}})}
        \leq2^{4\lfloor\kappa/2\rfloor}=C_g,
    \end{aligned}
\end{equation}
where $\mathcal A_\mu$ and $\mathcal B_\mu$ are linear maps on the respective local operator spaces.
If
\begin{equation}
    X=\sum_{r=1}^{R}X_r^{A}\otimes X_r^{A^c}
\end{equation}
has operator Schmidt rank at most $R$ across $A|A^c$, extending the local maps by identities on the other sites yields
\begin{equation}
    \mathcal G(X)=\sum_{r=1}^{R}\sum_{\mu=1}^{K}
    \mathcal A_\mu(X_r^{A})\otimes\mathcal B_\mu(X_r^{A^c}).
\end{equation}
Thus, $\mathrm{rank}_{A|A^c}(\mathcal G(X))\leq KR\leq C_gR$.

The product initial state has operator Schmidt rank one across every cut.
Gates and noise channels supported entirely inside $A$ or entirely inside $A^c$ do not increase the Schmidt rank across $A|A^c$, since they map each product term to another product term.
Each crossing gate increases the rank by at most $C_g$, and there are at most $C_{\partial}a(A)$ such gates per layer.
After $L$ layers,
\begin{equation}
    \mathrm{rank}_{A|A^c}(\sigma_L)\leq C_g^{C_{\partial}La(A)}.
\end{equation}
The normalized squared Schmidt coefficients form a probability distribution, so
\begin{equation}
    \widetilde S_{OE}^{A}(\sigma_L)
    \leq\log_2\mathrm{rank}_{A|A^c}(\sigma_L)
    \leq C_{\partial}La(A)\log_2C_g
    =\order{La(A)}.
\end{equation}
Writing $t=\tr{(\sigma_L)^2}\in(0,1]$, \cref{lemma:normalized_oe_entropy} gives
\begin{equation}
    \begin{aligned}
        S_{OE}^{A}(\sigma_L)
        &=t\widetilde S_{OE}^{A}(\sigma_L)-t\log_2t\\
        &\leq C_{\partial}La(A)\log_2C_g+\frac{1}{e\ln2}
        =\order{La(A)},\qquad L\geq1.
    \end{aligned}
\end{equation}
Here $-t\log_2t\leq1/(e\ln2)$ because its derivative $-\log_2t-1/\ln2$ is positive below $e^{-1}$ and negative above it.
The additive constant is absorbed since $La(A)\geq1$ for integer $L\geq1$.
At $L=0$, the state is the pure product initial state and both entropies vanish.
\end{proof}

The rank-counting estimate in \cref{prop:locality_cutwise_schmidt_rank_growth} has a direct per-boundary-bond interpretation. Since $\overline\chi_{\partial}(A)^{a(A)}$ is the cutwise Schmidt-rank capacity represented by the boundary indices, the local rank bound gives
\begin{equation}
    \log \overline\chi_{\partial}(A)
    \leq
    C_{\partial}L\log C_g
    =
    \order{L}.
\end{equation}
Thus a local depth window $L=\order{1}$ costs only a bounded average boundary-bond dimension, while $L=\order{\log n}$ costs at most a polynomial one. Combined with the depolarizing OEE bounds, this gives a uniform-in-depth statement: before the relevant depolarizing crossover the locality estimate controls the required boundary capacity, and after the crossover the depolarizing estimate controls it.

\begin{corollary}[Polynomial average PEPO boundary-bond dimension at all depths]
\label{cor:depolarizing_pepo_polynomial_effective_boundary_bond_dimension}
Fix constant error tolerances $\varepsilon>0$ and $\delta\in(0,1)$ in \cref{prop:pepo_cutwise_scale}. Suppose the circuit with depolarizing noise is described as in the setting of \cref{prop:depolarizing_pepo_effective_boundary_bond_dimension}, with a pure product initial state.
Then for every circuit depth $L\geq0$ and every nontrivial cut $A\mid A^c$ with $a(A)>0$, the average PEPO boundary-bond dimensions required to approximate $\rho_L$ at absolute Hilbert-Schmidt error $\varepsilon$ and relative error $\delta$ obey
\begin{equation}
    \log \overline\chi_{\partial}^{\mathrm{abs}}(A)=\order{\log n},
    \quad
    \log \overline\chi_{\partial}^{\mathrm{rel}}(A)=\order{\log n}.
\end{equation}
Consequently,
\begin{equation}
    \overline\chi_{\partial}^{\mathrm{abs}}(A)= \order{\mathrm{poly}(n)},
    \quad
    \overline\chi_{\partial}^{\mathrm{rel}}(A)= \order{\mathrm{poly}(n)},
\end{equation}
uniformly in $L$. 
\end{corollary}

The purity bound in \cref{thm:: purity induced bound of SOE under strong contraction} also controls the cutwise average PEPO boundary-bond dimension at fixed absolute Hilbert--Schmidt accuracy.

\begin{theorem}[Constant cutwise PEPO boundary dimension under strong contraction]
\label{thm:worst_case_pepo_boundary_dimension}
Consider a noisy local circuit on $n\geq2$ qubits on a bounded-degree interaction graph $G$, with layer maps
\begin{equation}
    \Phi_\ell=\mathcal N^{\otimes n}\circ\mathcal U_\ell,
\end{equation}
where each $\mathcal U_\ell$ is a product of disjoint two-qubit gates supported on edges of $G$, and the fixed single-qubit channel $\mathcal N$ satisfies $c(\mathcal N)<1/6$.
For any initial state $\rho_0$ and integer depth $L\geq1$, the output $\rho_L=\Phi_L\circ\cdots\circ\Phi_1(\rho_0)$ satisfies, across every bipartition $A\mid A^c$ with $a(A)>0$,
\begin{equation}
    S_{OE}^{A}(\rho_L)
    \leq n\left(\frac12+3c\right)^n
    \log_2\frac{2}{\frac12+3c}
    =\order{1}.
\end{equation}
Consequently, for every fixed absolute squared Hilbert--Schmidt error tolerance $\varepsilon>0$, the associated cutwise average boundary-bond dimension can be chosen so that
\begin{equation}
    \begin{aligned}
    \log_2\overline\chi_{\partial}^{\mathrm{abs}}(A)
    &\leq\frac{1+S_{OE}^{A}(\rho_L)/\varepsilon}{a(A)}
    =\order{\frac{1}{a(A)}},\\
    \overline\chi_{\partial}^{\mathrm{abs}}(A)&=\order{1}.
    \end{aligned}
\end{equation}
The bounds are uniform in $n$, $L$, and the cut, with constants depending only on the fixed noise channel and, for the boundary-bond dimension, on $\varepsilon$.
If $\rho_0$ is a product state, possibly mixed, the constant entropy and boundary-bond dimension bounds also hold at $L=0$.
\end{theorem}
\begin{proof}
For $L\geq1$, \cref{lemma: exponential suppression of purity} applied to the last noise layer gives $\tr{\rho_L^2}\leq(\frac12+3c)^n$.
The entropy argument in the proof of \cref{thm:: purity induced bound of SOE under strong contraction} uses only this purity bound and the operator Schmidt rank bound $2^n$, so it gives the stated entropy estimate independently of the circuit geometry.
Maximizing $x(\frac12+3c)^x$ over $x>0$ gives
\begin{equation}
    \sup_{x>0}x\left(\frac12+3c\right)^x
    =\frac{1}{e\ln\!\left(1/(\frac12+3c)\right)},
\end{equation}
which is finite because $c<1/6$.
At $L=0$, a product state has operator Schmidt rank one, so \cref{lemma:normalized_oe_entropy} gives $S_{OE}^{A}(\rho_0)=-t\log_2t\leq1/(e\ln2)$, where $t=\tr{\rho_0^2}$.
Finally, \cref{prop:pepo_cutwise_scale} gives the boundary-bond dimension estimate, and $a(A)\geq1$ makes it uniform over all cuts under consideration.
\end{proof}

\end{document}